\PassOptionsToPackage{nameinlink}{cleveref}
\documentclass[a4paper,UKenglish,cleveref,autoref,thm-restate]{lipics-v2021}

\nolinenumbers
\usepackage{latexsym,amsmath,amsthm,amssymb}

\usepackage{pifont}
\usepackage{enumerate}
\usepackage{subcaption}
\usepackage[T1]{fontenc}
\usepackage{bbm}
\usepackage{cite}
\usepackage{mathtools}
\usepackage{amsfonts}
\usepackage{pifont}
\usepackage{url}
\usepackage{multirow}
\usepackage{multicol}
\usepackage{fancybox}
\usepackage{colortbl}
\usepackage{soul}
\usepackage{color}
\usepackage{tikz}
\usepackage{tkz-berge}
\usetikzlibrary{calc} 
\usetikzlibrary{backgrounds}
\usetikzlibrary{arrows.meta}
\usepackage{rotating}
\usepackage{pdflscape}
\usepackage{afterpage}
\usepackage{comment}
\usepackage{etoolbox}
\usepackage{environ}
\usepackage{dirtytalk}
\usepackage[rgb]{xcolor}

\usepackage{cite}

\usepackage{algorithm}
\usepackage{algpseudocode}

\usepackage{todonotes}

\definecolor{mid-green}{rgb}{0.15,0.65,0.15}
\definecolor{dark-green}{rgb}{0.15,0.25,0.15}
\definecolor{dark-red}{rgb}{0.7,0.15,0.15}
\definecolor{dark-blue}{rgb}{0.15,0.15,0.9}
\definecolor{medium-blue}{rgb}{0,0,0.5}
\definecolor{gray}{rgb}{0.5,0.5,0.5}
\definecolor{color-Ig}{rgb}{0.15,0.7,0.15}
\definecolor{darkmagenta}{rgb}{0.30, 0.0, 0.30}
\definecolor{orange}{rgb}{1,0.5,0}

\hypersetup{
  colorlinks, linkcolor={dark-blue},
citecolor={mid-green}, urlcolor={dark-blue}}

\Crefname{remark}{Remark}{Remarks}
\Crefname{lemma}{Lemma}{Lemmas}

\newtheorem{invariant}{Invariant}
\Crefname{invariant}{Invariant}{Invariants}

\newtheorem{condition}{Condition}
\Crefname{condition}{Condition}{Conditions}

\newtheorem{vcvcrule}{Reduction Rule}

\Crefname{vcvcrule}{Rule}{Rules}

\newtheorem{isfvsrule}{Reduction Rule}

\Crefname{isfvsrule}{Rule}{Rules}

\newtheorem{istdrule}{Reduction Rule}

\Crefname{istdrule}{Rule}{Rules}

\newtheorem{isbdrule}{Reduction Rule}

\Crefname{isbdrule}{Rule}{Rules}

\Crefname{fvsfvsrule}{Rule}{Rules}

\usetikzlibrary{decorations,arrows,shapes}
\newcommand{\inners}{1.2pt}
\newcommand{\outers}{1pt}
\newcommand{\angled}[1]{\left\langle{#1}\right\rangle}
\newcommand{\pdkernelname}{PD}
\newcommand{\pdkernel}{{\pdkernelname{} kernel}}
\newcommand{\pdkernels}{{\pdkernelname{} kernels}}
\newcommand{\pdkernelization}{{\pdkernelname{} kernelization}}
\usepackage{complexity}
\newcommand{\Hard}{\text{hard}}
\newclass{\pNP}{paraNP}
\newcommand{\Hness}{hardness}
\newcommand{\NPH}{\NP\text{-}\Hard}
\newcommand{\NPHness}{\NP\text{-}\Hness}

\newclass{\Complete}{complete}
\newclass{\Total}{Total}
\newclass{\Delay}{Delay}
\newclass{\Cness}{completeness}
\newclass{\PPT}{PPT}
\newclass{\IncP}{IncP}
\newclass{\DelayP}{DelayP}
\newclass{\Pclass}{P}

\newfunc{\enum}{enum}
\newfunc{\tin}{in}
\newfunc{\tout}{out}
\newfunc{\bits}{bits}
\newfunc{\Conf}{conf}
\newcommand{\conf}[2]{\Conf_{#1}\left(#2\right)}
\newcommand{\confg}[3]{\Conf^{#3}_{#1}\left(#2\right)}
\newfunc{\Sol}{Sol}
\newfunc{\YES}{yes}
\newfunc{\NOi}{no}
\newfunc{\bdd}{bd}
\newfunc{\tw}{tw}
\newfunc{\td}{td}
\newfunc{\dcc}{dcc}
\newfunc{\dc}{dc}
\newfunc{\roott}{root}
\newfunc{\mul}{mul}
\newfunc{\fvs}{fvs}
\newfunc{\dirx}{dir}
\newfunc{\allx}{all}
\newfunc{\largex}{large}
\newfunc{\prop}{propagate}

\newfunc{\dist}{dist}
\newfunc{\diam}{diam}
\newfunc{\border}{border}
\newfunc{\dgr}{deg}
\newfunc{\degR}{degR}

\newfunc{\Rep}{Rep}
\newfunc{\DPx}{DP}

\newfunc{\leavesx}{leaves}

\newcommand{\pname}[1]{{\sc #1}}

\newfunc{\opt}{opt}
\newfunc{\rmcx}{rmc}

\newfunc{\ins}{ins}
\newfunc{\id}{id}
\newfunc{\shift}{shift}
\newfunc{\glue}{glue}
\newfunc{\proj}{proj}
\newfunc{\joinf}{join}
\newfunc{\idx}{index}
\newfunc{\Lift}{Lift}
\newfunc{\ePPT}{ePPT}
\newfunc{\up}{up}
\newfunc{\cbd}{cb}
\newfunc{\down}{down}

\def\A{\mathcal{A}}
\def\B{\mathcal{B}}
\def\C{\mathcal{C}}
\def\D{\mathcal{D}}
\def\F{\mathcal{F}}
\def\G{\mathcal{G}}
\def\Hc{\mathcal{H}}
\def\I{\mathcal{I}}
\def\P{\mathcal{P}}

\def\S{\mathcal{S}}
\def\T{\mathcal{T}}
\def\X{\mathcal{X}}
\def\Y{\mathcal{Y}}
\newcommand{\cb}[1]{#1_\cbd}
\newcommand{\bd}[1]{\bdd\left(#1\right)}
\newcommand{\xmark}{\ding{55}}
\newcommand{\mycmark}{\ding{52}}
\newcommand{\inpFPT}{input-\(\FPT\)}

\newcommand{\bigO}[1]{{\mathcal{O}\!\left(#1\right)}}

\newcommand{\probl}[3]{
  \begin{flushleft}
    \fbox{
      \begin{minipage}{0.97\linewidth}
        \noindent {\sc #1}\\
        {\bf Instance:} #2\\
        {\bf Question:} #3
      \end{minipage}}
    \medskip
  \end{flushleft}
}

\newcommand{\problenum}[3]{
  \begin{flushleft}
    \fbox{
      \begin{minipage}{0.97\linewidth}
        \noindent {\sc #1}\\
        {\bf Instance:} #2\\
        {\bf Enumerate:} #3
      \end{minipage}}
    \medskip
  \end{flushleft}
}

\NewEnviron{sproof}[1] {
    \begin{proof}[Proof of safeness of~#1]\BODY\end{proof}
}

\NewEnviron{invproof}[1] {
    \begin{proof}[Soundness of~#1]\BODY\end{proof}
}

\NewEnviron{cproof} {
    \begin{proof}[Proof of the claim]\BODY\end{proof}
}

\definecolor{goodred}{HTML}{CC6677}
\definecolor{goodblue}{HTML}{332288}
\definecolor{goodyellow}{HTML}{DDCC77}
\definecolor{goodgreen}{HTML}{117733}
\definecolor{goodcyan}{HTML}{88CCEE}
\definecolor{goodwine}{HTML}{882255}
\definecolor{goodteal}{HTML}{44AA99}
\definecolor{goodolive}{HTML}{999933}
\definecolor{goodpurple}{HTML}{AA4499}

\algnewcommand\algorithmicforeach{\textbf{for each}}
\algdef{S}[FOR]{ForEach}[1]{\algorithmicforeach\ #1\ \algorithmicdo}

\usetikzlibrary { decorations.pathmorphing, decorations.pathreplacing, decorations.shapes}

\newcommand{\titlemath}[1]{$#1$}

\title{A more versatile model for enumerative kernelization: a case study for  Vertex
Cover} 

\author{Marin Bougeret}{LIRMM, Université de Montpellier, CNRS, Montpellier,
France}{marin.bougeret@lirmm.fr}{https://orcid.org/0000-0002-9910-4656}{}

\author{Guilherme C. M. Gomes}{LIRMM, Université de Montpellier, CNRS, Montpellier,
  France\\
  Universidade Federal de Minas Gerais, Belo Horizonte,
Brazil}{gcm.gomes@dcc.ufmg.br}{ttps://orcid.org/0000-0002-7487-3475}{Funded by the
  European Union, project PACKENUM, grant number 101109317. Views and opinions expressed
  are however those of the author only and do not necessarily reflect those of the European
Union. Neither the European Union nor the granting authority can be held responsible for them.}

\author{Ignasi Sau}{LIRMM, Université de Montpellier, CNRS, Montpellier,
France}{ignasi.sau@lirmm.fr}{https://orcid.org/0000-0002-8981-9287}{French project ELIT
(ANR-20-CE48-0008-01).}

\authorrunning{ } 

\Copyright{} 




\category{} 

\relatedversion{} 

\acknowledgements{}

\EventEditors{John Q. Open and Joan R. Access}
\EventNoEds{2}
\EventLongTitle{42nd Conference on Very Important Topics (CVIT 2016)}
\EventShortTitle{CVIT 2016}
\EventAcronym{CVIT}
\EventYear{2016}
\EventDate{December 24--27, 2016}
\EventLocation{Little Whinging, United Kingdom}
\EventLogo{}
\SeriesVolume{42}
\ArticleNo{23}

\begin{document}

\maketitle

\begin{abstract}
  Enumerative kernelization is a relatively recent and promising area sitting at the intersection of parameterized complexity and enumeration algorithms, with two main models being proposed. The first, known as enum-kernels and due to Creignou et al. [Theory Comput. Syst., 2017], was too permissive, leading to constant-sized kernels for every problem solvable with \FPT{}-delay.
To remedy this, Golovach et al. [J. Comput. Syst. Sci., 2022] proposed the polynomial-delay enumeration kernelization  model that, while addressing the shortcoming of the previous one, appears to be too strict, which we believe is a central reason for the slow development that the area has enjoyed so far.
In this paper, we propose a new model for enumeration kernels, which we have called polynomial-delay (PD) kernels. It is more flexible than Golovach et al.'s kernels while still preserving their qualities; informally, it allows us to ignore ``bad'' solutions of the compressed instance when producing the solution set of the input instance, but still requires that the ``good'' solutions are lifted with polynomial-delay. After discussing the main properties of our model, we design a generic framework for vertex-subset problems to adapt decision kernels into PD kernels of the same size.
We showcase our model's increased versatility and the  expressive power of our framework on the \textsc{Enum Vertex Cover} problem, where we want to list all vertex covers of size at most $k$ of a given graph. In particular, we manage to generalize the kernelization dichotomy by Bougeret et al. [SIAM J. Discrete Math., 2022] about the existence of polynomial kernels for \textsc{Vertex Cover} parameterized by the vertex deletion distance to a minor-closed graph class, as well as the solution size and feedback vertex number parameterizations.
The second one, in particular, is significantly simpler than the kernel designed by Bougeret et al. [IPEC, 2025], requiring only a few lines for its lifting algorithm.
Beyond our framework, we also show how to generalize to the enumeration setting the kernel of Bougeret et al. [Algorithmica, 2019] for the vertex-deletion distance to $c$-treedepth.
\end{abstract}

\section{Introduction}

Enumeration problems are central objects in both practical and theoretical computer science, with the main example of the former being the branch-and-bound method~\cite{branch_bound} (and its relatives), universally used in integer programming solvers.
Instead of answering a question in either the positive or the negative, in this class of problems the goal is to generate all witnesses, or \emph{solutions}, that satisfy the given problem's constraints, or decide that no solution exists; we denote by $\Sol(x)$ the solution set of an instance $x$ of some fixed problem of interest.
Typically, the number of solutions will be exponential in the input size $n$, and thus \emph{input-sensitive} algorithms, whose running times are only analyzed with respect to $n$, are not powerful enough tools to capture all nuances intrinsic to enumeration problems; in particular, \emph{input-polynomial} algorithms, whose running times are of the form $n^\bigO{1}$, are inadequate models of efficient enumeration.
Consequently, enumeration algorithms are usually analyzed in the \emph{output-sensitive} context, i.e., their running times depend on both $n$ and the size of the output; in this setting, the most popular classes of efficiently solvable enumeration problems are those that admit \emph{incremental-polynomial time}~\cite{frequent_incremental,triangle_incremental} or \emph{polynomial-delay}~\cite{bagan_delay,fo_query_delay,maximal_ind_set_delay} algorithms, named $\IncP$ and $\DelayP$, respectively, for which the relation $\DelayP \subsetneq \IncP$ holds~\cite{strozecki_thesis}.
In $\IncP$, the time taken to output the $i$-th solution should be polynomial in $i + n$;  whereas in $\DelayP$, we want that, between startup and the first output (the precalculation), any two outputs, and between the last output and termination (postcalculation), the running time is bounded by $\poly(n)$.
These classes were extended by Creignou et al.~\cite{creignou_hard}, which established links between enumeration complexity and the polynomial hierarchy of Stockmayer~\cite{polynomial_hierarchy} to derive a hierarchy of hard enumeration problems and appropriate notions of reductions between them.
Incremental-polynomial time and polynomial-delay algorithms are unfeasible goals in many interesting cases; in particular, enumeration problems for which the existence of a single solution is an $\NPH$ problem do not admit such algorithms unless $\Pclass = \NP$.

Orthogonally to enumeration, the theory of parameterized complexity gained significant traction as a means of coping with \NPHness, with thousands of papers and some books~\cite{flum_grohe,cygan_parameterized,niedermeier_book,downey_fellows,book_kernels} dedicated to the subject.
In this framework, algorithms are analyzed according to both the total input size $n = |x|$ and a \emph{parameter} of interest $k = \kappa(x)$, that is also part of the input, with the efficient algorithms, known as \emph{fixed-parameter tractable} (\FPT), being those of running time $f(k) \cdot n^\bigO{1}$, for some computable function $f$.
This increase in computational power allows for the existence of efficient algorithms for several important $\NPH$ problems, such as \pname{Vertex Cover} and \pname{Feedback Vertex Set}, when parameterized by the solution size.
Complementarily, a rich lower bound theory capable of showing fine grained limits to this broader notion of tractability has also been developed, offering further insights between parameterized and non-parameterized decision problems.
An equivalent definition of tractability in this framework, and central to our work, is that of \emph{kernelization}. 
A \emph{kernel} is a polynomial-time algorithm that, given an instance $(x,k)$, outputs an equivalent instance $(y,\ell)$  such that $y,\ell \leq g(k)$ for some computable function $g$, known as the \emph{size} of the kernel.
Kernels can be seen as mathematically guaranteed \emph{pre-processing} algorithms, and as a way to explain the tremendous success that simple clean-up routines have in real-world applications, such as in integer optimization.
Polynomial size kernels are of particular interest, as they typically result in significant reductions in input size.

A very successful case study in kernelization is that of \pname{Vertex Cover}, arguably the most fundamental problem in parameterized complexity.
Abu-Khzam et al.~\cite{crown_decomp_fellows} presented the first kernel with $\bigO{k}$ vertices for the problem when parameterized by the solution size $k$, using the then-novel concept of \emph{crown decompositions} to improve on the previously best known kernel with $\bigO{k^2}$ vertices, implied by a result of Buss and Goldsmith~\cite{buss_kernel}.
\pname{Vertex Cover} has been the paradigmatic problem in the very active area of kernelization with {\em structural parameters}, that is, parameters that take into account structural properties of the input, and that can be much smaller than the solution size. The ultimate goal is to unveil the ``smallest'' parameters for which the problem admits a polynomial kernel. This research program was triggered by the breakthrough result of Jansen and Bodlaender~\cite{vc_fvs} showing the existence of a kernel with $\bigO{t^3}$ vertices for \pname{Vertex Cover} parameterized by the size $t$ of a given feedback vertex set.
Not long after, it became a textbook example to show that no polynomial kernel exists when parameterizing by the treewidth of the input graph~\cite{cross_composition,cygan_parameterized}, and so the main question was now how deep into the graph parameter hierarchy one could go, particularly when the parameterization was on the vertex-deletion distance to minor-closed graph classes.
In 2019, Bougeret and Sau~\cite{bougeret_is_td} took another step in this direction by presenting a polynomial kernel under the parameterization of vertex-deletion distance to a graph of treedepth $c$, for some constant $c$; see also~\cite{MajumdarRR18,FominS16,HolsK17}. Finally, Bougeret et al.~\cite{bridgedepth} proved that, under standard complexity hypotheses, \pname{Vertex Cover} admits a polynomial kernel when parameterized by the vertex-deletion distance to some minor-closed graph class $\F$ if and only if $\F$ has bounded \emph{bridgedepth}, a new parameter they introduced to this end.

Given the growth of the parameterized complexity field, it is no surprise then that there have been mounting efforts to extend it to the enumeration setting~\cite{meeks_oracle,fernau_param_enum,damaschke_cluster_editing,damaschke_full_kernels,golovach2022refined,enum_long_path,multicut_ipec,oscar_degeneracy,enum_dcut,creignou2017enum}, with \emph{\inpFPT}, \emph{incremental-\FPT}, and \emph{\FPT-delay} having their expected definitions; we also highlight the work of Creignou et al.~\cite{creignou2017enum} as being particularly influential in this endeavor.
Satisfactory notions of kernelization for enumeration problems, which we broadly call \emph{enumerative kernelization}, were not as quickly defined, even though kernelization itself did play a significant role in several works.
In the early 2000s, Fernau~\cite{fernau_param_enum} presented an \inpFPT{} algorithm for the enumeration of \emph{minimum} vertex covers when parameterized by the solution size $k$, one of the first parameterized enumeration algorithms in the literature.
Interestingly for our work, Fernau's algorithm uses a kernelization strategy based on the classical Buss kernel~\cite{buss_kernel,downey_fellows} to argue that we can reduce this problem to the enumeration of minimum vertex covers in a graph with $\bigO{k^2}$ vertices; the algorithm, however, uses standard parameterized decision complexity language and mechanisms, and so does not offer many other insights beyond its own existence.
A few years later, Damaschke~\cite{damaschke_full_kernels} investigated subset minimization problems, i.e., where the goal is to enumerate every subset of size at most $k$ of a given ground set that satisfies the desired property.
For example, \pname{Enum Vertex Cover} is the problem where, given a graph $G$ and integer $k$, the goal is to enumerate every subset of vertices of size at most $k$ that hits all edges of $G$.
For this type of problem, Damaschke introduced the notion of \emph{full kernels}: in a subset minimization problem, a full kernel is a set that contains the union of all minimal solutions of size at most $k$.
To the best of our knowledge, this is the first model of enumerative kernelization, even if it is not a general purpose one. 
Nevertheless, Damaschke applied this new formalism to \pname{Enum $c$-Hitting Set} parameterized by the maximum solution size $k$, thus showing an \inpFPT{} algorithm under this parameterization.
In a subsequent paper, Damaschke~\cite{damaschke_cluster_editing} successfully applied full kernels to \pname{Enum Cluster Edge-Editing}: the problem of enumerating all ways of making at most $k$ editions, with each edition being an edge addition or removal, such that every connected component of the resulting graph is a clique, when parameterized by $k$.
All of these results, in the end, imply that the target problems admit \inpFPT{} algorithms; in this direction, only very recently did Golovach et al.~\cite{golovach2022refined} define \emph{fully-polynomial enumeration kernels} (FPE kernels), which are equivalent to the existence of an \inpFPT{} algorithm, much like \FPT{} algorithms and kernels are equivalent in the decision world; we will not discuss input-sensitive enumeration algorithms or kernels further, as our focus is on delay-based algorithms.

Output-sensitive parameterized enumeration was little explored until Creignou et al.~\cite{creignou2017enum} defined \emph{enum-kernels}: a general purpose kernelization framework for parameterized enumeration problems that is applicable to a problem if and only if it admits an \FPT-delay algorithm.
Intuitively, the kernel has two parts: the first, $\A_1$, is the \emph{compression algorithm} and has the same constraints as a decision kernel, while the second $\A_2$, known as the \emph{lifting algorithm}, receives the input $(x,k)$, the output $(y,\ell)$ of $\A_1$, and a solution $Y \in \Sol(y)$, and must output, with $\left(f(k + \ell) \cdot \poly(|x| + |y| + |Y|)\right)$-delay, where $f$ is a computable function, a (possibly empty) subset $\S_Y \subseteq \Sol(x)$; additionally, the $\S_Y$'s must be disjoint and their union must be precisely $\Sol(x)$.
The authors then showcased their proposed kernelization model on \pname{Enum Vertex Cover} parameterized by the maximum solution size (called \pname{All-Vertex-Cover} in~\cite{creignou2017enum}) by proving that the classical Buss kernel~\cite{buss_kernel,downey_fellows} is an enum-kernel.
This model, however, has a fatal limitation.
As shown by Golovach et al.~\cite{golovach2022refined}, a parameterized enumeration problem admits an \FPT-delay algorithm if and only if it admits an enum-kernel of \emph{constant} size.
As such, there is little point in developing small enum-kernels besides as a certificate that the target problems has an \FPT-delay algorithm.
To overcome this issue, Golovach et al. proposed a new model for enumerative kernelization, which they called \emph{polynomial-delay enumeration kernel}: instead of \FPT-delay, the lifting algorithm is restricted to \emph{polynomial-delay} but must always output at least one solution.
With this alteration, they are still capable of proving that their kernels and \FPT-delay algorithms are equivalent, while gaining the critical property that a problem admits a kernel of constant size if and only if it admits a polynomial-delay algorithm.

In the same paper that introduces the kernelization model, Golovach et al.~\cite{golovach2022refined} provide these types of kernels for several enumeration variants of the \pname{Matching Cut} problem under different structural parameterizations; it is also indirectly discussed how the Buss kernel is, in fact, a polynomial-delay enumeration kernel for \pname{Enum Vertex Cover} of quadratic size.
Their work has since been followed by some further examples of these kernels.
In particular, Komusiewicz and Majumdar~\cite{enum_dcut} studied structural parameterizations of \pname{Enum $d$-Cut}, whose decision version is a direct generalization of \pname{Matching Cut}.
Simultaneously, Gomes et al.~\cite{multicut_ipec} generalized \pname{Matching Cut} in a new direction, known as \pname{Matching Multicut}, as well as some results of Golovach et al.~\cite{golovach2022refined} to \pname{Enum Matching Multicut}.
Outside of \pname{Matching Cut}-adjacent problems, we are only aware of two works on the subject. The first, by Komusiewicz et al.~\cite{enum_long_path}, is on an enumeration variant of \pname{Longest Path} under structural parameterizations.
In the second one, which only came out very recently, Bougeret et al.~\cite{vc_fvs_strong_pd_kernel} developed polynomial-sized polynomial-delay enumeration kernels for \pname{Enum Vertex Cover} and \pname{Enum Feedback Vertex Set} under their natural parameters $k$; for the first, they developed a lengthy lifting algorithm for the classical crown decomposition-based kernel to obtain a kernel with $2k$ vertices, while for the latter they designed a novel $\bigO{k^3}$ kernel.

While the enum-kernels of Creignou et al.~\cite{creignou2017enum} had too much lifting power,
Golovach et al.'s~\cite{golovach2022refined} polynomial-delay enumeration kernels have proven quite complex to be designed, with even what should be the simplest example, \pname{Enum Vertex Cover}, demanding long and cumbersome lifting algorithms.
By closely examining the latter model's examples, much of the difficulty seems to originate in the requirement that \emph{every} solution $Y$ is lifted to a \emph{non-empty} set $S_Y$.
Informally, this imposes the constraint that no ``noise'' is introduced in $\Sol(y)$ when compressing $x$ (typically, it implies that $|\Sol(y)| \le |\Sol(x)|$).
This is a critical limitation in our view, as it implies that polynomial-delay enumeration kernels cannot handle failure states, a common strategy in enumeration algorithms.
For example, in the ordered generation framework of Eiter et al.~\cite{eiter2003new}, initially developed for the \pname{Minimal Hypergraph Transversal Enumeration} problem, known as \pname{Trans-Enum} in the enumeration community, bad intermediate solutions exist and are pruned iteratively, and these may only be found deep within the algorithm's recursion tree.
Furthermore, the $e$-, $D$-, and $I$-reductions used by Creignou et al.~\cite{creignou_hard} to characterize their hierarchy of hard enumeration problems heavily rely on the possibility that the reduced instance has solutions unsuited to be lifted back to solutions of the input.

\smallskip
\noindent\textbf{Our contributions.} 
We propose a new enumerative kernelization model that sits between Creignou et al.'s and Golovach et al.'s; while it requires that lifting be performed with polynomial-delay (as in the latter), it also handles failure states, by allowing that some solutions of the compressed instance be ignored by the lifting procedure (as in the former). As evidenced by our results, this allowance leads to simpler algorithms while having the same nice theoretical guarantees of polynomial-delay enumeration kernels.
Our new enumerative kernelization model, which we call \emph{polynomial-delay kernel}, \pdkernel{} for short, can also be seen as a generalization of Golovach et al.'s polynomial-delay enumeration kernels; for convenience we rename the latter as \emph{strong \pdkernels{}}.
We begin by showing some fundamental properties of our model.
In particular, we prove that its existence indeed corresponds to \FPT-delay algorithms and that the size of the kernel is a meaningful metric to optimize (cf. \cref{thm:kernel_fpt_delay_equiv}).
To showcase our model's versatility, we introduce a generic framework for graph vertex-subsets problems, which allows for the systematic translation of decision kernels to their enumeration counterparts. This framework is based on
\cref{cond:good_decision_kernel,cond:decidable_trace,cond:poly_delay}, which encapsulate sufficient properties that, if satisfied by a decision kernel and the problem $\Pi$ at hand, immediately imply a \pdkernel{} of the same size for the problem of enumerating all solutions of $\Pi$; this is a highly desired property, as it is natural to try to leverage the extensive body of knowledge of decision kernelization to the enumeration setting. Intuitively, \cref{cond:good_decision_kernel} asks that the solution set of the input can be partitioned into similarity classes and that the kernelization preserves enough information to distinguish these classes; \cref{cond:decidable_trace} asks for an efficient algorithm that decides if an output solution should be a \emph{canonical solution} for a given class of input solutions; finally, \cref{cond:poly_delay} asks for a polynomial-delay algorithm that takes a canonical solution and lists all input solutions of the associated class.
We prove in the sufficiency of these conditions in \cref{thm:framework_kernel}.
As an evidence of our model's power, we apply the proposed framework in our main technical contribution: proving that the dichotomy for minor-closed graph classes of decision kernelization~\cite{bridgedepth} also holds in the enumeration setting.
That is, we show that, under standard complexity hypotheses, \pname{Enum Vertex Cover} admits a polynomial-sized \pdkernel{} when parameterized by the size of a given modulator $X$ of $G$ to a minor-closed graph class $\F$ (i.e. $G \setminus X$ has no minor forbidden by $\F$) if and only if $\F$ has bounded bridgedepth; this is formally stated in \cref{thm:dichotomy}.

\begin{restatable}{theorem}{thmdichotomy}
    \label{thm:dichotomy}
    Let $\F$ be a minor-closed graph class and assume $\NP \nsubseteq \coNP/\poly$.
    \pname{Enum Vertex Cover} parameterized by the size of a given modulator to $\F$ admits a \pdkernel{} of polynomial size if and only if $\F$ has bounded bridgedepth.
\end{restatable}

The positive part of our proof of \cref{thm:dichotomy} essentially boils down to showing that, modulo a very simple rule applied at the beginning of the compression step to identify trivial instances, \cref{cond:good_decision_kernel,cond:decidable_trace,cond:poly_delay} are readily satisfied by the kernel developed by Bougeret et al.~\cite{bridgedepth} to prove the decision version of \cref{thm:dichotomy}.
The hardness part of the proof follows immediately from the matching decision kernelization lower bound for \pname{Vertex Cover} proved in~\cite{bridgedepth}.

Beyond \cref{thm:dichotomy}, to illustrate how our model and framework can be used, we engage in the more interesting journey of reproving some of the most influential kernelization results for \pname{Vertex Cover} that led to the analogous dichotomy theorem.
Namely, we also revisit the parameterizations by: the size of the maximum size of the wanted vertex cover~\cite{buss_kernel,crown_decomp_fellows}, the feedback vertex number~\cite{vc_fvs}, and the vertex-deletion distance to graphs of treedepth at most $c$~\cite{bougeret_is_td}.
With the exception of the latter, all of our proofs are based on the developed framework.
This is due to the fact that the decision kernel of Bougeret and Sau~\cite{bougeret_is_td} relies on a bikernelization strategy; while it is possible to generalize our framework to encompass this type of strategy, it becomes overly cumbersome, and it is much simpler to prove that it can be accomplished in an ad-hoc manner.
We do so by introducing what we believe to be the first  bikernelization algorithm for an enumeration problem, as well as appropriate notions of reductions between enumeration problems that enable such an algorithm to exist.

\smallskip
\noindent\textbf{Further research.}
We hope that the added versatility of the introduced model will further stimulate research on enumerative kernelization.
We are interested in investigating enumeration versions of classical problems, such as \pname{Feedback Vertex Set}, and determining if and how other decision kernels can be adapted to our model.
Specifically for \pname{Feedback Vertex Set}, we are interested in the existence of quadratic \pdkernels{}, which would directly improve the cubic strong \pdkernel{} recently given by Bougeret et al.~\cite{vc_fvs_strong_pd_kernel}.
We also point out that a lower bound theory for enumerative kernelization (and for enumerative parameterized complexity as a whole!) would be of extreme benefit to the area.
Currently, we rely on lower bounds for related parameterized {\sl decision} problems to obtain enumeration lower bounds for both \FPT-delay and \pdkernelization{} algorithms.
It is highly unlikely, however, that this strategy is powerful enough to distinguish which parameterized enumeration problems are tractable or not, and which admit kernels of polynomial size or not.

\smallskip
\noindent\textbf{Organization.} In 
\cref{sec:new-model} we introduce our new model for enumerative kernelization and prove that it indeed satisfies the natural desirable properties that one may expect.
In \cref{sec:framework} we present the generic framework that allows to systematically lift decision kernels to enumeration kernels, based on 
\cref{cond:good_decision_kernel,cond:decidable_trace,cond:poly_delay}.
To illustrate the versatility of our model, in particular via the application of the generic framework, in \cref{sec:vcvc} we obtain a simple linear kernel for \pname{Enum Vertex Cover} parameterized by the solution size ({\sl much} simpler than the linear strong \pdkernel{} recently obtained by Bougeret et al.~\cite{vc_fvs_strong_pd_kernel}).
Before our structural parameterizations, we present a polynomial-delay algorithm for \pname{Enum Independent Set} that outputs the solutions in lexicographic order using flashlight search in \cref{sec:enum_is}; the output ordering is a key feature that we heavily rely on for our more complex kernels.
In \autoref{sec:isfvs}, we present a polynomial kernel for \pname{Enum Vertex Cover} by the feedback vertex number, also using our generic framework. 
Our polynomial kernel for \pname{Enum Vertex Cover} by the 
vertex-deletion distance to graphs of bounded treedepth, which does not rely on our framework, is presented in \cref{sec:vc_td}; this example shows that the applicability of our kernelization model is not limited to that of our framework, with the latter being one of many tools available.
In \cref{sec:isbd}, we prove our polynomial \pdkernel{} for \pname{Enum Vertex Cover} by the 
vertex-deletion distance to graphs of bounded bridgedepth, which, together with the decision lower bound of Bougeret et al.~\cite{bridgedepth}, implies \cref{thm:dichotomy}.
The kernel of \cref{thm:dichotomy} follows a very similar idea to the feedback vertex set parameterization, but it is much longer and the arguments a bit more involved.
Standard preliminaries about graphs~\cite{murty} and parameterized complexity~\cite{cygan_parameterized} can be found in \cref{sec:prelim}.

\subsection{Standard preliminaries}
\label{sec:prelim}

We denote $\{1, 2, \dots, n\}$ by $[n]$.
We use standard graph-theoretic notation, and we consider simple undirected graphs without loops or multiple edges; see~\cite{murty} for any undefined terminology.
When the graph is clear from the context, the degree (that is, the number of neighbors) of a vertex $v$ is denoted by  $\deg(v)$, and the number of neighbors of a vertex $v$ in a set $A \subseteq V(G)$ and its neighborhood in it are denoted by $\deg_A(v)$ and $N_A(v)$; we also define $N(S) = \bigcup_{v \in S} N(v) \setminus S$.
A \textit{matching} $M$ of $G$ is a subset of edges of $G$ such that no vertex of $G$ is incident to more than one edge in $M$; for simplicity, we define $V(M) = \bigcup_{uv \in M} \{u,v\}$ and refer to it as the set of \textit{$M$-saturated vertices}.
The \textit{subgraph of $G$ induced by $X$} is defined as $G[X] = (X, \{uv \in E(G) \mid u,v \in X\})$.
A \emph{feedback vertex set} $X \subseteq V(G)$ is such that $G \setminus X$ is a forest; the \emph{feedback vertex number} of $G$ is the size of a feedback vertex set of minimum size.
The vertex-deletion distance to $\mathcal{G}$ is the size of a minimum cardinality set $U \subseteq V(G)$ such that $G \setminus U = G[V(G) \setminus U]$ belongs to class $\mathcal{G}$; in this case, $U$ is called the $\mathcal{G}$-modulator.

We refer the reader to~\cite{downey_fellows,cygan_parameterized} for basic background on parameterized complexity, and we recall here only some basic definitions.
A \emph{parameterized problem} is a tuple $(L, \kappa)$ where $L \subseteq \Sigma^*$ is a decidable language and $\kappa: \Sigma^* \mapsto \mathbb{N}$ is a function called parameterization.
An \textit{instance} of a parameterized problem is string $x \in \Sigma^\star$, with $k = \kappa(x)$ being called its \emph{parameter}.
A parameterized problem is \emph{fixed-parameter tractable} (\FPT) if there exists an algorithm $\mathcal{A}$, a computable function $f$, and a constant $c$ such that, when given $(x,k)$,
$\mathcal{A}$ (called an \FPT\ \emph{algorithm}) correctly decides whether $x \in L$ in time bounded by $f(k) \cdot |x|^c$.
A fundamental concept in parameterized complexity is that of \emph{kernelization}; see~\cite{book_kernels} for a recent book on the topic. A kernelization
algorithm, or just \emph{kernel}, for a parameterized problem $\Pi $ takes as input~$(x,k)$ and, in time polynomial in $|x| + k$, outputs~$(y,\ell)$ such that $y, \ell \leq g(k)$ for some
function~$g$, and $x \in L$ if and only if $y \in L$.
The function~$g$ is called the \emph{size} of the kernel.
A kernel is called \emph{polynomial} (resp. \emph{quadratic, linear}) if the function $g(k)$ is a polynomial (resp. quadratic, linear) function in $k$.




\section{A new model for enumerative kernelization}
\label{sec:new-model}

Let $U$ be a finite set and define $2^U$ as the \emph{powerset} of $U$, i.e., the set of all subsets of $U$.
A collection $\P = \{P_1, \dots, P_k\}$ is a \emph{($k$-)partition} of $U$ if $\P \subseteq 2^U \setminus \{\emptyset\}$, $\bigcup_{i \in [k]} P_i = U$, and for every $P_i,P_j \in \P$, we have $P_i \cap P_j \neq \emptyset$ if and only if $i = j$.
We start with basic definitions.

\begin{definition}[Parameterized enumeration problem (Creignou et al.~\cite{creignou2017enum}, Golovach et al.~\cite{golovach2022refined})]
    \label{def:enum_problem}
    A \emph{parameterized enumeration problem} over a finite alphabet $\Sigma$ is a tuple $\Pi = (L, \Sol, \kappa)$, shorthanded by $\Pi^\kappa$, such that:
    \begin{enumerate}[i.]
        \item $L \subseteq \Sigma^\star$ is a decidable language.
        \item $\Sol \colon \Sigma^\star \mapsto 2^{\Sigma^\star}$ is a computable function such that, for every \emph{instance} $x \in \Sigma^\star$, $\Sol(x)$ is finite and we have $\Sol(x) \neq \emptyset$ if and only if $x \in L$.
        \item $\kappa\colon \Sigma^\star \mapsto \mathbb{N}$, is the \emph{parameterization}.
    \end{enumerate}
    An instance $x$ of $\Pi^\kappa$ is a \emph{\NOi-instance} if $\Sol(x) = \emptyset$, and is a \emph{\YES-instance} otherwise.
\end{definition}

\begin{definition}[\FPT-delay algorithm (Creignou et al.~\cite{creignou2017enum}, Golovach et al.~\cite{golovach2022refined})]
    \label{def:fpt_delay}
    Let $\Pi^\kappa$ be a parameterized enumeration problem.
    A \emph{fixed-parameter-tractable-delay algorithm}, or simply \emph{\FPT-delay} algorithm, is an algorithm $\A$ that, given an instance $x$ of $\Pi^\kappa$ and $k = \kappa(x)$, outputs $\Sol(x)$ such that, the precalculation time, the time between any two consecutive outputs, and the postcalculation time, are bounded by $f(k)\cdot|x|^\bigO{1}$ for some computable function $f$.
    If $f(k) \in \bigO{1}$, then we say that $\A$ is a \emph{polynomial-delay} algorithm.
\end{definition}

We are now ready to formally define our kernelization model.

\begin{definition}[Polynomial-delay kernel]
    \label{def:enum_kernels}
    Let $\Pi^\kappa$ be a parameterized enumeration problem.
    A \emph{polynomial-delay kernel}, \emph{PD kernel} for short, is a pair of algorithms $\A_1, \A_2$ that, given an instance $x$ of $\Pi^\kappa$ and $\kappa(x)$:
    \begin{itemize}
        \item Algorithm $\A_1$ runs in $\poly(|x| + \kappa(x))$-time and outputs an instance $y$ of $\Pi^\kappa$  satisfying $\kappa(y), |y| \leq f(\kappa(x))$, with $f$ a computable function, and $\Sol(y) \neq \emptyset$ if and only if $\Sol(x) \neq \emptyset$;
        \item Algorithm $\A_2$ receives $x$, $y = \A_1(x,\kappa(x))$, $\kappa(x)$, $\kappa(y)$, some $Y \in \Sol(y)$, and outputs, with polynomial-delay $g$ on its total input size, $S_Y \subseteq \Sol(x)$.
        Moreover, we require that $\{S_Y \mid Y \in \Sol(y), S_Y \neq \emptyset\}$ is a partition of $\Sol(x)$.
    \end{itemize}
    We call $\A_1$ the \emph{compression algorithm} and $\A_2$ the \emph{lifting algorithm}.
    We say that $f$ is the \emph{size} of the kernel, and $g$ is its \emph{delay}.
    The pair $(\A_1, \A_2)$ is a \emph{strong \pdkernel} if, for every $Y \in \Sol(y)$, $\A_2$ never outputs the empty set.
\end{definition}

\begin{remark}
    \label{rmk:strong_pde_equivalence}
    Strong \pdkernels{} and polynomial-delay enumeration kernels, as defined by Golovach et al.~\cite{golovach2022refined}, are the same objects. Moreover, in their definition, the $\Sol(x) \neq \emptyset \Leftrightarrow \Sol(y) \neq \emptyset$ condition is implied by $\A_2$ only being allowed to output non-empty sets.
\end{remark}

Observe that, \emph{a priori}, we could alter the definition of the compression algorithm to only require that $\Sol(x) \neq \emptyset$ implies $\Sol(y) \neq \emptyset$; the lifting algorithm could still be used to correctly identify that $\Sol(x) = \emptyset$ even if $\Sol(y) \neq \emptyset$. 
This, however, would configure a situation that to us seems unnatural, as we highlight in \cref{remark:something-bad-would-happen}.

\begin{remark}
\label{remark:something-bad-would-happen}
     Take the problem of enumerating all independent sets of size at least $t$ of a graph $G$ when parameterized by size of a modulator $X$ to graphs of treewidth two , with $(G, X, t)$ as its input. Outputting $(G[X], X, 0)$ is enough: given $A \in \Sol(G[X], X, 0)$ it suffices to use Courcelle's Theorem~\cite{courcelle_book} and flashlight search~\cite{strozecki_eatcs} to check if $G \setminus (N(A) \cup X)$ has an independent set of appropriate size and list all of the appropriate ones. However, it is believed that no kernel exists for the decision version of this problem \cite{CyganLPPS14}.
\end{remark}

Our next theorem shows that \pdkernels{} and \FPT-delay algorithms are equivalent. Our proof is similar to the one provided by Golovach et al.~\cite{golovach2022refined} for strong \pdkernels, but must account for the case where the output instance has more solutions than the input.

\begin{restatable}{theorem}{thmkernelfptdelayequiv}
    \label{thm:kernel_fpt_delay_equiv}
    A parameterized enumeration problem $\Pi^\kappa_\enum$ admits an $\FPT$-delay algorithm if and only if it admits a \pdkernel{}.
    Moreover, $\Pi^\kappa_\enum$ can be solved with polynomial-delay if and only if it admits a constant-sized \pdkernel{}. 
\end{restatable}

\begin{proof}
    Let us first prove the converse, i.e., that the existence of a \pdkernel{} $(\A_1, \A_2)$ implies the existence of an \FPT-delay algorithm.
    Let $x$ be our input instance and $y = \A_1(x, \kappa(x))$.
    As $\Sol$ is a computable function, it follows that there exists an algorithm that enumerates $\Sol(y)$ whose complexity is a function of $|y|$; let $g(|y|)$ be the total time spent in enumerating $\Sol(y)$.
    By definition, this means that $\Sol(y)$ can be enumerated in \FPT-time as $|y|,\kappa(y) \leq f(\kappa(x))$ for some computable function $f$ as guaranteed by $\A_1$.
    If the algorithm outputs nothing, then $\Sol(y) = \Sol(x) = \emptyset$ and the algorithm only used \FPT-time in total.
    Now, if $\Sol(y) \neq \emptyset$ and since $(\A_1, \A_2)$ is a \pdkernel, $\{\A_2(x,y,\kappa(x), \kappa(y), Y) \neq \emptyset \mid Y \in \Sol(y)\}$ is a partition of $\Sol(x)$.
    Note that, if the kernel is strong, then we are done, as the total delay will be at most the cost of computing $\Sol(y)$, as each $Y \in \Sol(y)$ can be lifted with polynomial-delay.
    If not, then it could be the case that several elements $Y \in \Sol(y)$ have $\A_2(x,y, \kappa(x), \kappa(y),Y) = \emptyset$ but, as $|\Sol(y)| \leq g(|y|) \leq g\left(f(\kappa(x))\right)$, the delay between outputting any two consecutive solutions of $x$ is at most \FPT.

    Before proceeding to the other direction of the proof, 
    let us define some important types of instances for the remainder of the proof.
    An instance $x$ of $\Pi$ is a \emph{trivial \NOi-instance} if $x$ is of minimum size while being a \NOi-instance; a \emph{trivial \YES-instance} is defined similarly.
    Note that, if $\Pi^\kappa$ admits both \NOi- and \YES-instances, then a trivial \NOi-instance and a trivial \YES-instance both have constant sizes that depend only on $\Pi^\kappa$ and both can be computed in constant time: for each instance size $t$, it suffices to check, for all strings $s \in \Sigma^t$, if $\Sol(s) = \emptyset$; as $\Pi_k$ is a fixed problem, $\Sigma$ is fixed and finite, and $\Sol$ is a computable function, such a process gives us our trivial instances.
    
    For the other direction, we proceed in a similar fashion to the classic proof that \FPT\ algorithms and kernels are equivalent~\cite{cygan_parameterized}.
    Let $x$ be an input instance to $\Pi^\kappa$, $k = \kappa(x)$, and $\D$ be an \FPT-delay algorithm that enumerates $\Sol(x)$ in $(f(k)\cdot|x|^c)$-delay, where $f$ is computable and $c = \bigO{1}$; due to the former, we may assume that $f(k)$ can be computed by an algorithm $\F$ in $g(k)$-time, and let $h(k) = \max\{f(k), g(k)\}$.
    We begin by running $\F$ for $n = |x|$ steps.
    If $\F$ does not compute $f(k)$ in $n-1$ steps, we have that $g(k) \geq n$ and so $x$ has bounded size.
    Consequently, if $\A_1(x)$ outputs $x$ and $\A_2(x,\kappa(x),x,\kappa(x))$  outputs $X$ for every $X \in \Sol(x)$ we have our kernelization algorithm.
    For the remainder of the proof, we may assume that $f(k)$ was computed in less than $n$ steps, i.e., $g(k) < n$.
    
    Note that a similar analysis is applicable if $f(k) \geq n$, so we may additionally assume that $f(k) < n$.
    In this case, however, we have that $\D$ runs with $n^{c+1}$-delay and correctly enumerates $\Sol(x)$.
    Thus, we run $\D$ until it either outputs a first solution of $x$, or terminates.
    In the latter case, $\A_1$ computes a trivial \NOi-instance and outputs it, with $\A_2$ being allowed to be arbitrary.
    If, however, $\Sol(x) \neq \emptyset$, $\A_1$ computes a minimum \YES-instance $y$ in constant time. Note that $|\Sol(y)|$ is constant and so $\Sol(y)$ can also be computed in $\bigO{1}$-time.
    Now, we run $\D$ until it outputs $t = \min\{|\Sol(y)|, |\Sol(x)|\}$ solutions of $x$.
    \begin{itemize}
        \item If $t = |\Sol(y)|$, then $|\Sol(y)| \leq |\Sol(x)|$, and $\A_2(x,\kappa(x),y,\kappa(y),Y)$ works as follows: (\textit{i}) first, compute $\Sol(y)$ in constant time and order it according to the computation order, with $Y_i$ being the $i$-th solution in this order; (\textit{ii}) if $Y = Y_i$ and $i < t$, output $\{X_i\}$, the singleton set with the $i$-th solution produced by $\D$, and note that this takes at most $(n^{c+1}\cdot i)$-time as we can just run $\D$ until it outputs $X_i$, which happens to be its $i$-th output; (\textit{iii}) finally, if $Y_i = Y_t$, then we use $\D$ to generate $\Sol(x) \setminus \{X_1, \dots, X_{t-1}\}$ with $(n^{c+1}\cdot t)$-delay using a similar strategy to the previous one.
        \item If, however $t = |\Sol(x)|$, then $t < |\Sol(y)|$ and we proceed with the construction of $\A_2$ as in the above case, except that, for $Y_i \in \Sol(y)$ with $i > t$, $\A_2$ outputs the empty set.
    \end{itemize}
    Observe that the above analysis implies that we output an instance of size at most $\bigO{h(k)}$. 
    For the second statement, we proceed as above.
    That is, given a \pdkernel{} and $|y|$ of constant size, we get that $|\Sol(y)|$ is bounded by a constant that depends only on $\Pi_k$.
    As such, the delay of running $\A_2$ for each $Y \in \Sol(y)$ is at most polynomial in $|x| + \kappa(x)$: between any two solutions of $y$ for which $\A_2$ does not output the empty set, we run $\A_2$ in polynomial time a constant number of times and then enumerate with polynomial-delay.
    In the other direction, the arguments are the same, sufficing to note that $f(k)$ and $g(k)$ are constants.
\end{proof}

In~\cite{enum_long_path}, it was proved that every strong \pdkernel{} admits \FPT{} precalculation time and polynomial-delay; intuitively, their result guarantees that searching for small strong \pdkernels{} with fast lifting algorithms is a worthwhile pursuit.
We generalize their result with \cref{rmk:new_kernel_size_time}.
This coincides with the proof of Creignou et al.~\cite{creignou2017enum} that every problem solvable by an \FPT-delay algorithm can be solved with \FPT{} precalculation time and polynomial-delay.

\begin{restatable}{remark}{rmknewkernelsizetime}
    \label{rmk:new_kernel_size_time}
    Let $\Pi^\kappa_\enum = (L, \Sol, \kappa)$ be a parameterized enumeration problem such that there exists a $T(|x|)$-time algorithm $\F$ to compute $\Sol(x)$, where $x$ is an instance of $\Pi^\kappa_\enum$.
    If $\Pi^\kappa_\enum$ admits a \pdkernel{} of size $f$ and delay $g$, then $\Pi^\kappa_\enum$ admits an algorithm with $T(f(\kappa(x))) \cdot g(|x| + \kappa(x) + 2f(\kappa(x)))$ precalculation time and delay $g$.
\end{restatable}

\begin{proof}
    Let $k = \kappa(x)$ and $(\A_1, \A_2)$ be the algorithms making up our kernel of size $f(k)$ and delay $g(\cdot)$.
    We devise our \FPT-delay algorithm as follows:
    \begin{enumerate}
        \item Initialize $\Y = \emptyset$.
        \item Run $\A_1(x,k)$ to obtain $(y,\ell)$ in $\poly(|x| + k)$-time.
        \item Run $\F(y)$ to obtain $\Sol(y)$ in $T(|y|)$-time, which is at most $T(f(k))$.
        \item For each $Y \in \Sol(y)$, we run $\A_2(x,y,k,\ell,Y)$ until it either produces a first output or terminates; for each $Y$ in the first case, we set $\Y \gets \Y \cup \{Y\}$.
        As $\Sol(y)$ can be computed in $T(f(k))$-time, it follows that $|\Sol(y)| \leq T(f(k))$; as $\A_2$ only runs until its first output or termination, we have that it only used at most $g(|x| + |y| + k + \ell)$-time, which is at most $\poly(|x| + k + 2f(k))$-time, in each call.
        Consequently, this step was performed in $T(f(k)) \cdot g(|x| + \kappa(x) + 2f(\kappa(x)) + |Y|)$ time; by assuming that $|Y|$ is linearly bounded by $f(|x|)$, we have that the previous running time is equivalent to $T(f(k)) \cdot g(|x| + \kappa(x) + 2f(\kappa(x)))$.
        This concludes the proof of the bound on the precalculation time.
        \item For each $Y \in \Y$, we now run $\A_2(x,y,k,\ell,Y)$, which has $g(|x| + |y| + k + \ell +|Y|)$-delay, which concludes the proof. \qedhere
    \end{enumerate}
\end{proof}

In light of \cref{thm:kernel_fpt_delay_equiv} and \cref{rmk:new_kernel_size_time}, we note that, if $\Pi^\kappa_\enum$ admits a polynomial-delay algorithm $\P$, then $f(k) \in \bigO{1}$, which implies that $T(f(\kappa(x))) \cdot g(|x| + \kappa(x) + 2f(\kappa(x))) \in \bigO{g(|x|)}$; that is, after we kernelize the input into $y$, we obtain the following polynomial-delay algorithm: \textit{(i)} run $\P$ until termination to obtain a list $\Sol(y)$ ($\bigO{1}$); \textit{(ii)} for each $S \in \Sol(y)$, run the lifting algorithm until termination or first output, with $\L$ being the list of solutions that fall in the latter case ($\bigO{g(|x|)}$); \textit{(iii)} finally, for each surviving solution in $L$, run the lifting algorithm itself. Since $|L| \in \bigO{1}$, step \textit{(iii)} runs with polynomial-delay $g$.
We also point out that \textit{(ii)} is the reason for the $g(\cdot)$ precalculation overhead in \cref{rmk:new_kernel_size_time}.

\section{A framework for translating decision kernels into \pdkernels{}}
\label{sec:framework}

While decision kernelization has a large literature from which to draw when designing a new kernel, the same is yet untrue for enumerative kernelization.
As such, it is highly desired to understand when one can translate a kernel for a decision problem to one of its enumeration counterparts.
In this section, we take a first step in this direction, identifying three sufficient conditions that allow the design of a \pdkernel{} from an existing decision kernel.
We concern ourselves only with \emph{graph vertex-subset problems}, i.e., problems where the solutions are precisely subsets of vertices.
While this may seem restrictive, several cornerstone problems in parameterized complexity belong to this class; as examples, we point to both \pname{Vertex Cover} and \pname{Clique}.
For this section, let $\Pi^\kappa$ be a parameterized graph vertex-subset problem with parameterization $\kappa$ and $\Pi^\kappa_{\enum}$ be the problem of enumerating all solutions of $\Pi^\kappa$. %

Formally, both $\Pi^\kappa$ and $\Pi^\kappa_\enum$ have as input the triple $(G, \I, k)$, where $G$ is the graph of interest, $k = \kappa(G, \I)$ is the parameter, and $\I$ is any other meaningful problem input, e.g. $\I$ could include a feedback vertex set of $G$ and/or the threshold size $t$ of the desired solution (for instance, of an independent set).
To derive a \pdkernel{} for $\Pi^\kappa_\enum$, it is necessary that a kernel exists for $\Pi^\kappa$, as a kernel for the former implies one for the latter.
In this spirit,  \cref{cond:good_decision_kernel,cond:decidable_trace,cond:poly_delay} are a set of sufficient properties that enable us to derive a kernel for $\Pi^\kappa_\enum$ from a kernel to $\Pi^\kappa$.
\cref{cond:good_decision_kernel} is the most technical one, but intuitively we ask that the $\Pi^\kappa$ kernel preserve some key information to the enumeration process, which we capture by the concepts of \emph{core} of the kernel and its \emph{good traces}.
\cref{cond:decidable_trace} allows us to efficiently separate the good traces from the bad ones; moreover, it allows us to identify a representative solution for each good trace.
Finally, given a good trace $Y$ of the core in the original instance, \cref{cond:poly_delay} allows us to list all solutions that intersect its core precisely at~$Y$.
We refer to \cref{fig:cond_diagram} for a diagram depicting these interactions.

\begin{condition}
    \label{cond:good_decision_kernel}
    $\Pi^\kappa$ admits a compression algorithm $\A_1$ for which the following holds.
    \begin{enumerate}
        \item Given $(G, \I, k)$, $\A_1$ outputs an equivalent $(H, \mathcal{L}, \ell)$ with $|H| \leq |G|$.
        \item \label{item:special_condition} Let $\Lambda_H \subseteq V(H)$, $\lambda\colon \Lambda_H \to V(G)$ be an injection, and $\Lambda_G = \lambda(\Lambda_H) = \bigcup_{v \in \Lambda_H} \{\lambda(v)\}$, i.e., $\lambda$ is a bijection from $\Lambda_H$ to $\Lambda_G$.%
        We say that $(\Lambda_H, \lambda)$ is a \emph{core} for $\A_1$ if, for every $Y_G \subseteq \Lambda_G$ for which there exists $S \in \Sol(G, \I, k)$ with $S \cap \Lambda_G = Y_G$, there exists $S' \in \Sol(H, \mathcal{L}, \ell)$ with 
        $Y_H = S' \cap \Lambda_H$ and $\lambda(Y_H) = Y_G$. Moreover, $\Lambda_H$ and $\lambda$ can be computed in $\poly(|G|, k, |\I|, |H|, \ell, |\mathcal{L}|)$-time.
        All such $Y_H \subseteq \Lambda_H$ and $\lambda(Y_H) \subseteq \Lambda_G$ that satisfy \cref{item:special_condition} are the \emph{good traces} of the core.
    \end{enumerate}
\end{condition}

We remark that, if in a kernel for $\Pi^\kappa$ we have that $V(H) \subseteq V(G)$ (which is the case for all kernels of this paper except for the one in \autoref{sec:isbd}) and take $\lambda$ as the identity function, \cref{cond:good_decision_kernel} just asks for the existence of a set $\Lambda_H \subseteq V(H)$ such that each of its subsets that is contained in a solution of $(G, \I, k)$ must also belong to some solution of $(H, \mathcal{L}, \ell)$.

\begin{condition}
    \label{cond:decidable_trace}
    Let $\A_1$ be the algorithm of \cref{cond:good_decision_kernel}, $(G, \I, k)$ be an input of $\Pi^\kappa_\enum$, and $(H, \mathcal{L}, \ell)$ be the output of $\A_1(G, \I, k)$.
    There exists a \emph{choosing algorithm} $\A_c$ that:
    \begin{enumerate}
        \item Runs in $\poly(|G|, k, |H|, \ell)$-time.
        \item Given $(G, \I, k, H, \mathcal{L}, \ell)$ and a single $S' \in \Sol(H, \mathcal{L}, \ell)$, $\A_c$ returns either \YES{} or \NOi.
        \item For every good trace $Y_H \subseteq \Lambda_H$, there is a unique $S_{Y_H} \in \Sol(H, \mathcal{L}, \ell)$ with $S_{Y_H} \cap \Lambda_H = Y_H$ and $\A_c(G, \I, k, H, \mathcal{L}, \ell, S_{Y_H})$ returning \YES; such an $S_{Y_H}$ is the \emph{canonical solution} for $Y_H$.
        \item For every $Y_H \subseteq \Lambda_H$ and $S' \in \Sol(H, \mathcal{L}, \ell) \setminus \{S_{Y_H}\}$ with $S' \cap \Lambda_H = Y_H$, $\A_c$ returns \NOi.
    \end{enumerate}
\end{condition}

\begin{condition}
    \label{cond:poly_delay}
    Let $(G, \I, k)$ be an input of $\Pi^\kappa_\enum$, $(H, \mathcal{L}, \ell)$ be the output of $\A_1(G, \I, k)$, and $\Lambda_G$ be defined as in \cref{cond:good_decision_kernel}.
    There exists an \emph{enumeration algorithm} $\A_e$ that, given a good trace $Y_G \subseteq \Lambda_G$, $\A_e$ produces the set $\{S \in \Sol(G, \I, k) \mid S \cap \Lambda_G = Y_G\}$ with $\poly(|G|, k, \I, |H|, \ell, \mathcal{L}, |Y_G|)$-delay.
\end{condition}

  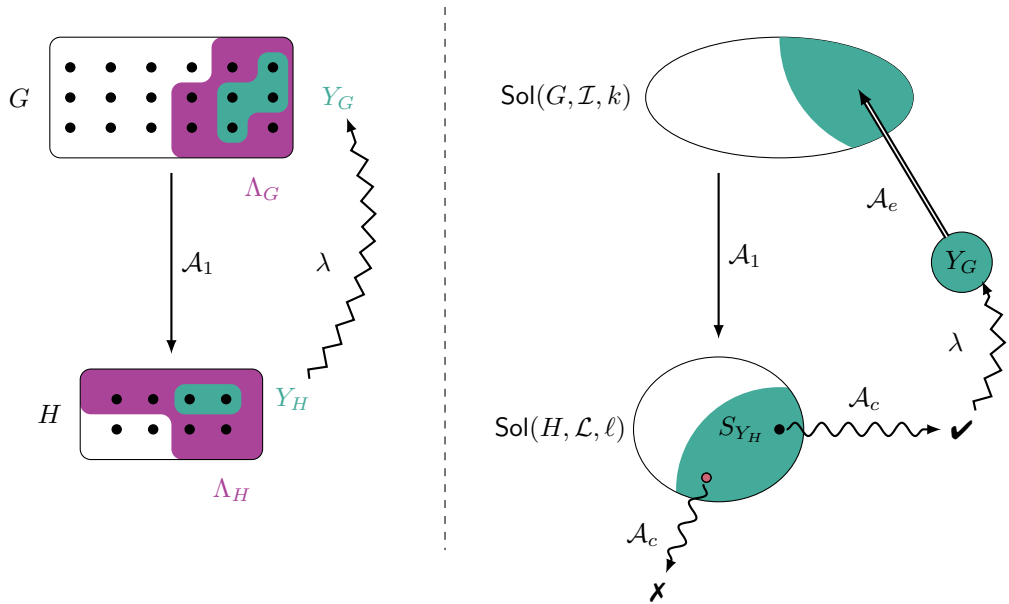
\begin{figure}[!htb]
    \centering
    \begin{tikzpicture}[scale=0.8]
      \GraphInit[unit=3,vstyle=Normal]
      \SetVertexNormal[Shape=circle, FillColor=black, MinSize=2pt]
      \tikzset{VertexStyle/.append style = {inner sep = \inners, outer sep = \outers}}
      \SetVertexLabelOut

      \newcommand{\goodcolorb}{goodteal}
      \newcommand{\goodcolora}{goodpurple}
      \newcommand{\goodcolorc}{goodyellow}

      \begin{scope}
          \draw[rounded corners] (0,0) rectangle (4,2);
          \node at (-0.5, 1) {$G$};
          \node at (3.5, -0.5) {\textcolor{\goodcolora}{$\Lambda_G$}};
          
          \node (yg) at (4.75, 1) {\textcolor{\goodcolorb}{$Y_G$}};
          
          \foreach \j in {1,...,6} {
                \pgfmathsetmacro{\x}{2*\j/3 - 0.33}
            \foreach \i in {1,...,3} {
                \pgfmathsetmacro{\y}{\i/2}
                \Vertex[NoLabel,x=\x, y=\y]{i\i\j}
            }
          }

          \begin{scope}[on background layer]
            \fill[\goodcolora, rounded corners] (2,0) -- (2,1.25) -- (2.66, 1.25) -- (2.66, 2) -- (4,2) -- (4,0) -- (2,0) -- (2,1.25);
            \fill[\goodcolorb, rounded corners] ($(i15)-(0.25, 0.25)$) -- ($(i25)+(-0.25, 0.25)$) -- ($(i26)+(-0.25, 0.25)$) -- ($(i36)+(-0.25, 0.25)$) -- ($(i36)+(0.25, 0.25)$) -- ($(i26)+(0.25, -0.25)$) -- ($(i15)+(0.25, 0.25)$) -- ($(i15)+(0.25, -0.25)$) -- ($(i15)-(0.25, 0.25)$) -- ($(i25)+(-0.25, 0.25)$);
          \end{scope}

          \draw[thick,-latex] (2, -0.25) -- (2, -3.25) node[midway,right] {$\A_1$};
      \end{scope}

      \begin{scope}[yshift=-5cm,xshift=0.5cm]
          \draw[rounded corners] (0,0) rectangle (3,1.5);
          \node at (-0.5, 0.75) {$H$};
          \node at (2.5, -0.5) {\textcolor{\goodcolora}{$\Lambda_H$}};
          
          \node (yh) at (3.5, 1) {\textcolor{\goodcolorb}{$Y_H$}};

          \Edge[style={decorate, -latex, decoration={zigzag, post length = 2mm}, bend right}](yh)(yg)
          \node at (4, 3.27) {$\lambda$};
          
          \foreach \j in {1,...,4} {
                \pgfmathsetmacro{\x}{\j*0.6}
            \foreach \i in {1,...,2} {
                \pgfmathsetmacro{\y}{\i/2}
                \Vertex[NoLabel,x=\x, y=\y]{o\i\j}
            }
          }
          
          \begin{scope}[on background layer]
            \fill[\goodcolora, rounded corners] (0, 1.5) -- (0, 0.75) -- (1.5, 0.75) -- (1.5, 0) -- (3,0) --(3, 1.5) -- (0, 1.5) -- (0, 0.75);
            \fill[\goodcolorb, rounded corners] ($(o23)-(0.25, 0.25)$) rectangle ($(o24)+(0.25, 0.25)$);
          \end{scope}
      \end{scope}

      \begin{scope}[xshift=6.5cm]
        \draw[dashed] (0, 2.5) -- (0,-6.5);
      \end{scope}

      \begin{scope}[xshift=12cm]
        \begin{scope}[yshift=1cm]
          \draw ellipse (2.2cm and 1cm) (0, 0);
          \node at (-3.5, 0) {$\Sol(G, \I, k)$};
          \begin{scope}
              \clip ellipse (2.2cm and 1cm) (0, 0);
              \fill[\goodcolorb] (2cm,1cm) circle (2cm);
          \end{scope}

          \node (t1) at (1.2cm, 0.3cm) {};
          
          \draw[thick,-latex] (-1cm, -1.25cm) -- (-1cm, -4cm) node[midway,right] {$\A_1$};
        \end{scope}

        \begin{scope}[yshift=-4.5cm,xshift=-1cm]
          \draw ellipse (1.4cm and 1.2cm) (0, 0);
          \node at (-2.6,0) {$\Sol(H, \mathcal{L}, \ell)$};
          \begin{scope}[on background layer]
              \clip ellipse (1.4cm and 1.2cm) (0, 0);
              \fill[\goodcolorb] (1cm,-1cm) circle (1.7cm);
          \end{scope}


          \Vertex[NoLabel,x=-0.2,y=-0.8]{s1}
          \AddVertexColor{goodred}{s1}
          \node (m1) at (-1cm, -2.7cm) {\xmark};
          
          \Edge[style={decorate, decoration={snake, post length=1mm},-latex}](s1)(m1)
          \node at (-1.25cm, -1.75cm) {$\A_c$};
          
          \Vertex[x=1,y=0,L={{$S_{Y_H}$}}, Lpos=180, Ldist=-1pt]{s3}

          \node (c1) at (4cm,0cm) {\mycmark};
          
          \Edge[style={decorate, decoration={snake, post length=1mm},-latex}](s3)(c1)
          \node at (2.4cm, 0.5cm) {$\A_c$};

          \node (ygr) at (4, 2.77) {};
          \node (ygrb) at ($(ygr)+(-30:0.3)$) {};
          \node (ygru) at ($(ygr)+(120:0.3)$) {};
          \draw[fill=\goodcolorb] (ygr) circle (0.5cm);
          \node at (ygr) {$Y_G$};
          \Edge[style={decorate, -latex, decoration={zigzag, post length = 2mm}, bend right}](c1)(ygrb)
          \node at (3.9cm, 1.45cm) {$\lambda$};

          \Edge[style={-latex, double}](ygru)(t1)
          \node at (2.7cm, 3.81cm) {$\A_e$};
        \end{scope}
      \end{scope}
    \end{tikzpicture}
    \caption{Interactions between the components of \cref{cond:good_decision_kernel,cond:decidable_trace,cond:poly_delay} on the input and output and on the solutions and solution sets. The lower purple-shaded area indicates the set $\Lambda_H$ of the output graph $H$, and the upper area $\Lambda_G = \lambda(\Lambda_H)$; teal-shaded areas correspond to the good traces $Y_H \subseteq \Lambda_H$ and $Y_G = \lambda(Y_H)$ and to the solutions with $Y_H$ or  as their trace.
    Straight arrows represent the compression algorithm given by \cref{cond:good_decision_kernel} on both the input and solution sets; wavy arrows the choosing algorithm of \cref{cond:decidable_trace}; zigzagging arrows the $\lambda$ function of kernel's core; double arrow the action of the enumeration algorithm of \cref{cond:poly_delay} upon the good trace $Y_G$. \label{fig:cond_diagram}}
    \end{figure}

Note that, by combining all three conditions, we have that $\{\{S \in \Sol(G, \I, k) \mid S \cap \Lambda_G = \lambda(Y)\} \mid Y \text{ is a good trace of } \Lambda_H\}$ is a partition of $\Sol(G, \I, k)$.
Our conditions illustrate the two main challenges when working with enumerative kernelization.
Namely, having a large core is helpful, as it becomes easy to preserve good traces and lift them back to the input instance; however, we want to compress the instance as much as possible, so, in particular, we must work with a core that is much smaller than the input's.

Before proceeding to the  main theorem of this section, we highlight that a similar set of conditions can be obtained for strong \pdkernels{}, but there are two main challenges in specifying such a collection.
The first is on the compression algorithm: besides preserving all the good traces, the kernel of \cref{cond:good_decision_kernel} would also have to exclude \emph{every} non-good trace, which is a feature rarely found in decision kernels; in particular, when applying reduction rules, we typically restrict ourselves to only preserving some maximal/minimal solutions.
The second is on the lifting algorithm: by requiring that it always produces some output, we must also be able to partition the set of solutions with a same trace; this puts additional strain on the algorithm, making its design significantly more complicated.

\begin{restatable}{theorem}{thmframeworkkernel}
    \label{thm:framework_kernel}
    Let $\Pi^\kappa$ be parameterized decision graph vertex-subset problem and $\Pi^\kappa_\enum$ its parameterized enumeration version.
    If there exist $\A_1, \A_e, \A_c$ that satisfy \cref{cond:good_decision_kernel,cond:decidable_trace,cond:poly_delay} and $\A_1$ is a kernel, then $\Pi^\kappa_\enum$ admits a \pdkernel{} of the same size as $\A_1$ and delay bounded by the maximum between the delay of $\A_e$ and the running time of $\A_c$.
\end{restatable}

\begin{proof}
    Let $(G, \I, k)$ be an input of $\Pi_\enum$.
    Observe that, by definition, $(G, \I, k)$ is also a valid input of $\Pi$ and $(G, \I, k)$ is a \YES-instance of $\Pi$ if and only if $\Sol(G, \I, k) \neq \emptyset$.
    As such, to obtain our \pdkernel{}, we begin by setting $\A_1$ (given in \cref{cond:good_decision_kernel}), as our compression algorithm.
    Let $(H, \mathcal{L}, \ell) = \A_1(G, \I, k)$ and $(\Lambda_\A, \lambda)$ be the core of the kernel; by the definition of a decision kernel, $|H|, \ell, \mathcal{L} \leq f(k)$. Let $\Lambda_G = \lambda(\Lambda_\A)$. 
    If $\A_1$ returns a \NOi-instance of $\Pi$, then there is nothing to lift and we are done, so assume this is not the case. Given a solution $S' \in \Sol(H, \mathcal{L}, \ell)$, the lifting algorithm, $\A_2$, works as follows: first, it runs algorithm $\A_c$ of \cref{cond:decidable_trace}.
    If $\A_c$ returns \YES{}, then $\A_2$ runs algorithm $\A_e$ of \cref{cond:poly_delay} with $Y_G = \lambda(S' \cap \Lambda_\A)$ as input and returns its output, otherwise $\A_2$ returns the empty set.
    This concludes the description of the \pdkernel{} $(\A_1, \A_2)$.
    
    Let us prove that $(\A_1, \A_2)$ is indeed a \pdkernel{}; to do so, it suffices to show that all solutions of $(G, \I, k)$ are enumerated with polynomial-delay by $\A_2$ without repetitions, as $\A_1$ already satisfies the conditions on the compression algorithm of \cref{def:enum_kernels}.
    First, observe that $\A_2$ runs in $\max \{\poly(|G|, k, \I, |H|, \ell, \mathcal{L}, S'), \poly(|G|, k, \I, |H|, \ell, \mathcal{L}, Y_H)\}$-precalculation time and $\poly(|G|, k, \I, |H|, \ell, \mathcal{L}, Y_H)$ delay, where $S' \in \Sol(H, \mathcal{L}, \ell)$ and $Y_H = S' \cap \Lambda_\A$.
    Let $S_1 \in \Sol(G, \I, k)$, $Y_H = \lambda^{-1}(S_1 \cap \Lambda_G)$, and $S' \in \Sol(H, \mathcal{L}, \ell)$ be a solution of $(H, \mathcal{L}, \ell)$ with $S' \cap \Lambda_\A = \lambda_H$ for which algorithm $\A_c$ returned \YES, i.e., it is the canonical solution for $Y_H$; by \cref{cond:decidable_trace}, $S'$ exists and is unique.
    As $\A_2$ is applied to every solution of $(H, \mathcal{L}, \ell)$, it is applied to $S'$.
    By our construction of $\A_2$, $\A_e$ will be run on $Y_H$.
    That is, $\A_2$ will output the set $\{\Sol(G, \I, k) \mid S \cap \Lambda_G = \lambda(Y_H)\}$, which contains $S_1$, when given $S'$.
    This implies that every solution of $\Sol(G, \I, k)$ is output by some call of our lifting algorithm.
    Note that $S_1$ is output only when $\A_2$ is given $S'$: any other $S'' \in \Sol(H, \mathcal{L}, \ell)$ for which $\A_c(|G|, k, \I, |H|, \ell, \mathcal{L}, S'')$ outputs \NOi{}, with $S'' \cap \Lambda_\A \neq Y_H$ or $S'' \cap \Lambda_\A = Y_H$, does not produce $S_1$; the first follows from the construction of $\A_2$, while the latter additionally follows from \cref{cond:poly_delay}.
    Finally, $\A_2$ produces no other output: $\A_c$ only returns \YES{} on a unique instance for each good trace, and $\A_e$ only outputs solutions of $(G, \I, k)$.    
\end{proof}

As our kernels rely on reduction rules, we prove that they respect \cref{cond:good_decision_kernel} in order to leverage \cref{thm:framework_kernel}.
The core of a \pdkernel{} is what links its compression and lifting parts, which allows for a clean decoupling of these phases and simplifies the design of the kernel.
A reduction rule is \emph{safe} if the corresponding compression algorithm satisfies the first item of \cref{cond:good_decision_kernel}, and it is \emph{sound} if it satisfies the second.
For most reduction rules we prove only their soundness, overall shortening our proofs.
When both safeness and soundness are needed, e.g., when we introduce a novel reduction rule, we prove them separately.
\cref{lem:condition_one_composes} shows that algorithms satisfying \cref{cond:good_decision_kernel} compose (given that they have compatible cores).

\begin{restatable}{lemma}{lemconditiononecomposes}
    \label{lem:condition_one_composes}
    Let $\A, \B$ be algorithms satisfying \cref{cond:good_decision_kernel} for the same problem $\Pi^\kappa$, $(G, \I, k)$ be an input of $\Pi^\kappa$,
    $(H, \mathcal{L}, \ell) = \B(G, \I, k)$, and $(\Lambda_\B, \lambda_\B)$ be a core for $\B$. If there is some core $(\Lambda_\A, \lambda_\A)$ for $\A$ such that $\lambda_\A(\Lambda_\A) = \Lambda_\B$, then the algorithm $\C$ obtained by applying $\B$ and then $\A$ satisfies \cref{cond:good_decision_kernel} with core $(\Lambda_\A, \lambda_\B \circ \lambda_\A)$.
\end{restatable}

\begin{proof}
    Let $(G, \I, k)$ be an input of $\Pi$, $(F, p, \P) = \A(H, \mathcal{L}, \ell) = \C(G, \I, k)$, 
    Let us show that each item of \cref{cond:good_decision_kernel} is satisfied in the same order they were stated:
    \begin{enumerate}
        \item As both $\A$ and $\B$ receive and output instances of $\Pi$, so does $\C$.
        \item Let $\Lambda_G = \lambda_\C(\Lambda_\A)$.
        We argue that $(\Lambda_\A, \lambda_\C)$ is a core for $\C$. As $\lambda_\A$ and $\lambda_\B$ are bijective functions, it follows that $\lambda_\C: \Lambda_\A \mapsto \Lambda_G$ is also bijective, i.e., $\lambda_\C^{-1}$ is unique and well defined.
        Now take $S_G \in \Sol(G, \I, k)$.
        By definition, there are $S_H \in \Sol(H, \mathcal{L}, \ell)$, and $S_F = \Sol(F, p, \P)$ such that $\lambda_\A(S_F \cap \lambda_F) = S_H \cap \lambda_\B = Y_H$ and $\lambda_\B(Y_H) = S_G \cap \Lambda_G$; that is $\lambda_\C(S_F \cap \lambda_F) = S_G \cap \Lambda_G$, as desired.
        \item As $\C = \A \circ \B$, it follows immediately that the core of $\C$ can be computed in the time taken to compute the cores of $\A$ and $\B$. \qedhere
    \end{enumerate}
\end{proof}

\begin{observation}
    \label{obs:good_core_subset}
    If $\A$ is an algorithm that satisfies \cref{cond:good_decision_kernel} for core $(\Lambda_H, \lambda)$, then for every $\Lambda'_H \subseteq \Lambda_H$ it holds that $\A$ satisfies \cref{cond:good_decision_kernel} for the core $(\Lambda'_H, \lambda')$, where $\lambda'$ is the restriction of $\lambda$ to $\Lambda'_H$.
\end{observation}

As all our kernels are defined by applying exhaustively some reduction rules, let us now explain how we use  \autoref{lem:condition_one_composes} and \autoref{obs:good_core_subset} to 
argue that, once the soundness of each rule is proved, this implies the soundness of the kernel that is defined by the composition of  these rules (using an appropriate core).

Let us start with the parameterization by the solution size.
In this setting, all reduction rules have the following structure.
Given an input $(G,t)$ of the rule, it outputs $(G',t')$, and we  choose as the core the set $V(G')$ with the identity injection.
Then, given the input of the kernel $(G,t)$, and  the output $(H,\ell)$ after applying exhaustively all rules, \autoref{obs:good_core_subset} implies, if we choose $V(H)$ as the core (for the whole kernel), that all rules are sound using $V(H)$ as the core, and \autoref{lem:condition_one_composes} implies that the overall kernel is sound as well.
Similarly, moving to structural parameterizations, all our rules also have a common structure.
Given an input $(G,X,t)$ of the rule, where $X$ is the modulator corresponding to a given parameterization, the rule outputs $(G',X',t')$, where $X'$ may be a subset or a superset of $X$, and we always choose as the core the set $X' \cap X$, again with the identity injection.
Then, given the input of the kernel $(G,X,t)$, and the output $(H,X',\ell)$  after applying exhaustively all rules, \autoref{obs:good_core_subset} implies that, if $X'$ is the kernel's core, all rules are sound using $X'$ as the core, and \autoref{lem:condition_one_composes} implies that the overall kernel is sound as well.
\section{The natural parameterization for \pname{Enum Vertex Cover}}
\label{sec:vcvc}

Let us formally define our target problem.
\problenum{Enum Vertex Cover}{A graph $G$ and an integer $k$.}{All vertex covers of $G$ of size at most $k$.}

\vspace{-.2cm}

We use crown decompositions to obtain a simple linear \pdkernel{}  for \pname{Enum Vertex Cover} parameterized by $k$.
While it has a slightly higher constant on the number of vertices than the strong \pdkernel{} presented by Bougeret et al.~\cite{vc_fvs_strong_pd_kernel} (3 versus 2), our lifting algorithm takes only a few lines, while theirs is eight pages long.
We start with \cref{obs:vcCore}, which states that if a compression algorithm only removes vertices, and, in addition, has a ``forward safeness proof'' that just restricts the solution to the output graph, then we can keep a large core.
Recall that having large cores can only make \cref{cond:poly_delay} easier to satisfy.
\begin{observation}
    \label{obs:vcCore}
    Let $\A_1$ be a \pname{Vertex Cover} kernel that: (i) given $(G,k)$, outputs $(G',\ell)$ where $G' \coloneqq G\setminus X$ for some $X \subseteq V(G)$; and (ii) for every $S \in \Sol(G,k)$, $S\setminus X \in \Sol(G',\ell)$.
    Then, by setting $V(G')$ and the identity injection as the core, we satisfy \cref{cond:good_decision_kernel}.
\end{observation}


We define crowned graphs in \cref{def:crowned_graph} and a strengthening of crown decompositions in \cref{def:heavy_crown}, where the body is small or, equivalently, that the crown is large.

\begin{definition}
  \label{def:crowned_graph}
  A graph is a \emph{crowned graph} if its vertex set can be partitioned into $H \cup C$, such that $C$ is an independent set, and such that there is a matching $M$ between $H$ and $C$
  with $|M|=|H|$. The set $C$ is called the \emph{crown}, and $H$ is called the \emph{head}.
  Given a graph $G$, a \emph{crown decomposition}~\cite[p. 26]{cygan_parameterized} of width $t$ is a partition $(C, H, B)$ of $V(G)$ such that $G[H \cup C]$ is
  a crowned graph (with head $H$), $N_G(C)=H$, and $|H| = t$; set $B$ is the \emph{body}.
\end{definition}

\begin{definition}[Heavy-crown decomposition]
    \label{def:heavy_crown}
    Given a graph $G$, a \emph{heavy-crown decomposition} of width $t$ is a crown decomposition $(C, H, B)$ of $G$ of width $t$ where $|B|+2|H| \leq 3t$.
\end{definition}

We remark that heavy-crown decomposition is in fact the original definition of Chor et al.~\cite{original_crown} for crown decompositions.
The proof of \cref{lem:heavy_crown_poly} is given as the proof of the Crown Lemma in the more recent book by Fomin et al.~\cite{book_kernels}; while they do not explicitly state it in terms of heavy-crown decompositions, one can easily check that it indeed leads to the conclusion that one can find such a decomposition in polynomial time.

\begin{lemma}
    \label{lem:heavy_crown_poly}
    Let $t \in \mathbb{N}$. There is a polynomial-time algorithm that, given $t$, and a graph $G$ with no isolated vertices and with $|V(G)| \geq 3t+1$, either: (\textit{i}) computes a matching of size at least $t+1$ of $G$, or (\textit{ii}) computes a heavy-crown decomposition of width at most $t$.
\end{lemma}

Let us now explain why we consider the heavy-crown variant instead of the classical one (where $B$ may be arbitrarily large).
In the textbook kernel for \pname{Vertex Cover}, the algorithm finds a crown decomposition, and as there is always at least one solution that contains $H$, it restarts on $(G',k')=(G \setminus (H \cup C),k-|H|)$).
Such a rule verifies the hypothesis of \cref{obs:vcCore}, and thus we could use it in a \pdkernel{} with the core $(V(G'), \id)$, and \cref{cond:good_decision_kernel} would be satisfied. However, proving that \cref{cond:poly_delay} holds would become challenging (see how it is done in \cite{vc_fvs_strong_pd_kernel}) as it would require, given a good trace $Y_G \subseteq V(G)\setminus (H \cup C)$, to enumerate with polynomial delay all vertex covers $S$ of $G$ of size at most $k$ such that $S \setminus (H \cup C)=Y_G$, which is non-trivial as $G[H \cup C]$ may be complex (even if $C$ is an independent set). 
Thus, we prefer to rely on heavy-crown decompositions, which allows us to only remove vertices from $C$ (as in \cite{ABUKHZAM2010524}) instead of $H \cup C$  to obtain our kernel.

We consider the kernel for \pname{Vertex Cover} based on the following two rules, and then show that we can indeed use this kernel to satisfy \cref{cond:good_decision_kernel,cond:decidable_trace,cond:poly_delay} to get a \pdkernel{}.

\begin{vcvcrule}
    \label{vcvcrule:deg_0}
    Let $(G, k)$ be an input instance of \pname{Vertex Cover}. If $v$ is an isolated vertex of $G$, remove $v$ from $G$ and set $k \gets k$.
\end{vcvcrule}


\begin{vcvcrule}
    \label{vcvcrule:unmatched_crown}
    Let $(G, k)$ be an input of \pname{Vertex Cover} and $(C, H, B)$ be a heavy-crown decomposition of $G$ of width $t \leq k$.
    Let $M^{\star}$ be an $H$-saturating matching between $H$ and $C$ (meaning $|M^{\star}|=|H|$), and $L = C \setminus V(M^{\star})$. Then, output $(G \setminus L, k)$.
\end{vcvcrule}

Note that the soundness proofs of \cref{vcvcrule:deg_0,vcvcrule:unmatched_crown} (choosing as the core the vertex set of the output graph as announced in \autoref{sec:framework}) is immediate by \autoref{obs:vcCore}.


Given an instance $(G,k)$ of \pname{Vertex Cover}, the compression algorithm $\A_1$ is as follows.
First, we apply \cref{vcvcrule:deg_0} exhaustively.
If $|V(G)|\le 3k$, we output $(G',k')=(G,k)$.
Otherwise, apply \cref{vcvcrule:unmatched_crown} once using the algorithm of \cref{lem:heavy_crown_poly} with $t=k$.
If it finds a matching of size $k+1$ then we output a trivial \NOi-instance. Otherwise, output $(G',k')$, the output of \cref{vcvcrule:unmatched_crown}.
\cref{lem:vcvc_bound} is immediate.

\begin{lemma}
    \label{lem:vcvc_bound}
    Let $(G, k)$ be an input of \pname{Vertex Cover} and $(G', k) = \A_1(G, k)$.
    It holds that $|V(G')| \leq 3k$ and $\Sol(G,k) \neq \emptyset$ if and only if $\Sol(G', k') \neq \emptyset$.
\end{lemma}



We are now ready to derive our \pdkernel{} using the framework of \cref{sec:framework}.

\begin{restatable}{theorem}{thmvcvckernel}
    \label{thm:vc_vc_kernel}
    \pname{Enum Vertex Cover} admits a \pdkernel{} with at most $3k$ vertices when parameterized by the maximum size of a solution $k$. 
\end{restatable}
\begin{proof}
Let us prove how the three conditions of \cref{sec:framework} can be satisfied.
Given an instance $(G,k)$ of \pname{Vertex Cover}, let $(G',k)$ be the output of $\A_1$. 
Notice first that, as explained at the end of \cref{sec:framework}, $\A_1$ 
verifies \cref{cond:good_decision_kernel} with core $(V(G'),\id)$.
Towards verifying \cref{cond:decidable_trace}, 
let $Y' \in \Sol(G', k)$, $L = V(G) \setminus V(G')$, and observe that, by our reduction rules, $L$ is an independent set.
First, we construct $Y = Y' \cup (L \cap N_G(V(G') \setminus Y'))$, that is, we add to $Y'$ the vertices of $L$ that are incident to edges not covered by $Y'$, i.e., they are mandatory in any solution of $G$ that intersects the core at $Y'$ (cf. \cref{fig:vcvc_extension}); note that this implies that $Y$ is a vertex cover of $G$.
As such, if $|Y| > k$, we have that $Y'$ is not a good trace.
As there is exactly one solution of $(G', k')$ with each good trace and we can decide if a trace is good or not in polynomial time, \cref{cond:decidable_trace} is satisfied.
Finally, \cref{cond:poly_delay} is also immediate. Indeed, given a good trace $Y' \subseteq V(G')$, we define $Y$ as above, and take $I = L \setminus Y$.
As $Y$ is already a vertex cover, it suffices to enumerate all subsets of $I$ of size at most $k-|Y|$.
\end{proof}

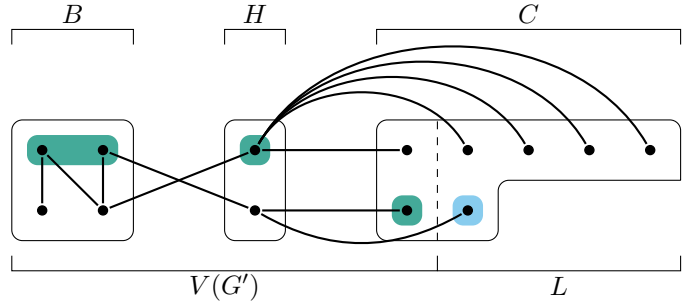
\begin{figure}[t]
    \centering
    \begin{tikzpicture}[scale=0.8]
      \GraphInit[unit=3,vstyle=Normal]
      \SetVertexNormal[Shape=circle, FillColor=black, MinSize=2pt]
      \tikzset{VertexStyle/.append style = {inner sep = \inners, outer sep = \outers}}
      \SetVertexNoLabel
      
      \begin{scope}
        \draw[rounded corners] (0, 0) rectangle (2, 2);
        
          \foreach \j in {1,...,2} {
                \pgfmathsetmacro{\x}{\j-0.5}
            \foreach \i in {1,...,2} {
                \pgfmathsetmacro{\y}{\i-0.5}
                \Vertex[x=\x, y=\y]{b\i\j}
            }
          }
      \end{scope}
      
      \begin{scope}[xshift=3.5cm]
        \draw[rounded corners] (0, 0) rectangle (1, 2);
        
          \foreach \j in {1,...,1} {
                \pgfmathsetmacro{\x}{\j-0.5}
            \foreach \i in {1,...,2} {
                \pgfmathsetmacro{\y}{\i-0.5}
                \Vertex[x=\x, y=\y]{h\i\j}
            }
          }
      \end{scope}

      \begin{scope}[xshift=6cm]
        \draw[rounded corners] (5, 1) -- (5, 2) -- (0, 2) -- (0, 0) -- (2, 0) -- (2, 1) -- (5, 1);
        \draw[dashed] (1, 0) -- (1,2);
        
          \foreach \j in {1,...,1} {
                \pgfmathsetmacro{\x}{\j-0.5}
            \foreach \i in {1,...,2} {
                \pgfmathsetmacro{\y}{\i-0.5}
                \Vertex[x=\x, y=\y]{c\i\j}
            }
          }
          \foreach \j in {2,...,5} {
            \pgfmathsetmacro{\x}{\j-0.5}
            \Vertex[x=\x, y=1.5]{l2\j}
            \Edge[style={bend left=60}](h21)(l2\j)
          }
        \Vertex[x=1.5, y=0.5]{l12}
      \end{scope}
      \Edge[style={bend right=30}](h11)(l12)
      \Edges(h11,c11)
      \Edges(h21,c21)
      \Edges(b11,b21,b12,b22,h11)
      \Edges(h21,b12)
      
      \newcommand{\solcolor}{goodteal}
      \begin{scope}[on background layer]
          \fill[\solcolor, rounded corners] ($(b21)-(0.25, 0.25)$) rectangle ($(b22)+(0.25, 0.25)$);
          \fill[\solcolor, rounded corners] ($(h21)-(0.25, 0.25)$) rectangle ($(h21)+(0.25, 0.25)$);
          \fill[\solcolor, rounded corners] ($(c11)-(0.25, 0.25)$) rectangle ($(c11)+(0.25, 0.25)$);
          \fill[goodcyan, rounded corners] ($(l12)-(0.25, 0.25)$) rectangle ($(l12)+(0.25, 0.25)$);
          
      \end{scope}
      \draw (0, -0.25) -- (0, -0.5) -- (7, -0.5) -- (7, -0.25);
      \node at (3.5, -0.77) {$V(G')$};
      
      \draw (0, 3.25) -- (0, 3.5) -- (2, 3.5) -- (2, 3.25);
      \node at (1, 3.75) {$B$};
      
      \draw (3.5, 3.25) -- (3.5, 3.5) -- (4.5, 3.5) -- (4.5, 3.25);
      \node at (4, 3.75) {$H$};
      
      \draw (7, -0.5) -- (11, -0.5) -- (11, -0.25);
      \node at (9, -0.75) {$L$};
      
      \draw (6, 3.25) -- (6, 3.5) -- (11, 3.5) -- (11, 3.25);
      \node at (8.5, 3.75) {$C$};
    \end{tikzpicture}
    \caption{Illustration of the lifting algorithm of \cref{thm:vc_vc_kernel}, with $k=5$. The teal-shaded vertices correspond to $Y' \in \Sol(G', k)$, and the cyan-shaded vertex to the mandatory vertex that must be added to $Y'$ to obtain $Y \in \Sol(G, k)$. If we take $Y'$ as the unshaded vertices of $G'$, then we must add the top four vertices of $L$ to $Y'$ to obtain $Y$, but then $|Y| > k$, and so $Y \notin \Sol(G, k)$.\label{fig:vcvc_extension}}
\end{figure}We are now ready to derive our \pdkernel{} using the framework of \cref{sec:framework}.
\section{Structural parameterizations of \pname{Enum Vertex Cover}}
\label{sec:enum_is}

For our structural parameterizations, instead of directly kernelizing \pname{Enum Vertex Cover}, we opt to work on the dual problem (as done in decision kernelization), defined below.

\problenum{Enum Independent Set}{A graph $G$ and an integer $t$.}{All independent sets of $G$ of size at least $t$.}

When parameterized by their respective solution sizes, \pname{Enum Vertex Cover} and \pname{Enum Independent Set} behave very differently. However, note that one can reduce from an instance $(H, k)$ of the former to an instance $(G, t)$ and preserve \emph{all} structural parameters: it suffices to set $G \coloneqq H$ and $t \coloneqq n - k$; as $G = H$ and $X \in \Sol(H, k)$ if and only if $V(G) \setminus X$ is an independent set of size at least $t$.
We prove a simple algorithm that helps in much of our work in \cref{thm:generic_lifting}; it can be seen as an application of the flashlight search method~\cite{strozecki_eatcs}, but we need some additional guarantees on the order of the output.
Let $U$ be a set, $\sigma$ be a total ordering of $U$, and $X,Y$ be two subsets of $U$.
We say that $X$ is \emph{lexicographically smaller with respect to $\sigma$} than $Y$ if one of the following is true: (i) $X \subset Y$ and the vertex  $u \in Y \setminus X$ of minimum  $\sigma(u)$ satisfies $\sigma(u) > \sigma(v)$ for every $v \in X$; or (ii) $\arg\min_{u \in X \triangle Y} \{\sigma(u)\}$ belongs to $X$, where $X \triangle Y$ is the symmetric difference of $X$ and $Y$.

\begin{restatable}{lemma}{thmgenericlifting}
    \label{thm:generic_lifting}
    Let $\mathcal{G}$ be a hereditary graph class such that \pname{Independent Set} can be solved in $f_\mathcal{G}(n)$-time for every $n$-vertex graph $G \in \mathcal{G}$.
    Then, there is a $(f_\mathcal{G}(n) \cdot n)$-delay algorithm that solves an instance $(G, t)$ of \pname{Enum Independent Set} and, if given a total ordering $\sigma$ of $V(G)$, outputs the solutions in lexicographic order with respect to $\sigma$.
\end{restatable}

\begin{proof}
    We apply the folklore technique to deduce a polynomial-delay algorithm from an algorithm to decide the extension problem, which we detail here for the sake of completeness.
    The corresponding extension problem here is defined as follows.
    
    \probl{Independent Set Extension}{A graph $G$, two subsets $M,P \subseteq V(G)$, and an integer $t$}{Is there an independent set of $G$ that avoids $P$, contains $M$, and has size at least $t$?}
    
    As $\mathcal{G}$ is hereditary, \pname{Independent Set Extension} on $\mathcal{G}$ reduces to the instance $(G \setminus (N_G[M] \cup P), t - |M|)$ of \pname{Independent Set} on $\mathcal{G}$.
    Consequently, as \pname{Independent Set} can be solved on $\mathcal{G}$ in $f_\mathcal{G}(n)$-time, we can efficiently provide an answer to $(G \setminus (N_G[M] \cup P), t - |M|)$.

    To solve \pname{Enum Independent Set}, we use a recursive branching algorithm where each node of the branching tree is of the form $(G,M,P,t)$, with the entries having the same meaning as the input to \pname{Independent Set Extension}.
    If $V(G) \setminus (N[M] \cup P)$ is empty, then the algorithm does nothing on that node, which is then a leaf of the branching tree.
    First, we pick the vertex $v \in V(G) \setminus (N[M] \cup P)$ of minimum $\sigma(v)$ to branch on and solve two instances of \pname{Independent Set Extension}, namely $I_\ell = (G, M \cup \{v\}, P, t)$, and $I_r = (G, M, P \cup \{v\}, t)$, each in $f_\mathcal{G}(n)$-time.
    We know that $I_\ell$ is a \YES-instance if and only if there is \emph{at least one} independent set of $G$ of size at least $t$ that contains $M \cup \{v\}$; note that if $G[M \cup \{v\}]$ is not independent we simply output that $I_\ell$ is a \NOi-instance.
    The analysis for $I_r$ is symmetric.
    Thus, before moving further down the recursion tree, we know that if $I = (G, M, P, t)$ is a \YES-instance of \pname{Independent Set Extension}, then at least one of $I_\ell,I_r$ is a \YES-instance; moreover, their solution sets are disjoint.
    Consequently, if we recursively call our branching algorithm for \pname{Enum Independent Set} on $I_x$ ($x \in \{\ell, r\}$) if and only if $I_x$ is a \YES-instance for \pname{Independent Set Extension}, then we are certain to find at least one appropriate independent set of $G$ when exploring $I_x$ or its descendants.
    To guarantee our ordering, we first recurse on $I_\ell$, the root of the \emph{left-subtree} of $I$, then recurse on $I_r$, the root of the  \emph{right-subtree}.
    Producing the output is not simply a matter of printing at the leaves, as this would break the lexicographic ordering requirement. We instead proceed as follows: if $I_\ell$ is a \YES-instance and $|M \cup \{v\}| \geq t$, then we output $M \cup \{v\}$; that is, we only produce some output when adding vertices to the current partial solution.
    
    As such, to solve \pname{Enum Independent Set}, we first check if $(G,\emptyset, \emptyset, t)$ is a \YES-instance of \pname{Independent Set Extension}; in the positive case, we call the above branching algorithm on $(G,\emptyset, \emptyset, t)$, otherwise we terminate immediately and output nothing.
    Let us first show by induction on $n = |V(G)|$ that the algorithm produces our solutions in lexicographic order according to $\sigma$ if we have a \YES-instance.
    If $n = 0$, then the claim follows immediately.
    Now, let $v$ be the first vertex of $\sigma$, $G' = G \setminus \{v\}$, and $G_v = G \setminus N[v]$.
    If $(G_v, t-1)$ is a \NOi-instance of \pname{Independent Set}, then we are done: $\Sol(G, t) = \Sol(G', t)$ which, by induction, is correctly computed in lexicographic order.
    Otherwise, every solution of $G$ that contains $v$ must be listed before every solution of $G'$, i.e., $\{\{v\} \cup S \mid S \in \Sol(G_v, t-1) \}$ lexicographically precedes every $R \in \Sol(G', t)$.
    By induction, $\Sol(G_v, t-1)$ is correctly computed in lexicographic order; since $(G, \{v\}, \emptyset, t)$ is explored before $(G, \emptyset, \{v\}, t)$ and their descendants correspond to $\Sol(G_v, t-1)$ and $\Sol(G', t)$, respectively, we have that $\Sol(G, t)$ is computed in lexicographic order.    
    
    Now, let us show that this strategy yields an algorithm with the claimed delay; note that the height of the branching tree is at most the size of a maximum independent set of $G$, which could be close to $n$.
    Take two consecutive solutions $S_1, S_2$, produced by the algorithm in this order, with $I_i = (G, S_i, P_i, t)$, $i \in \{1,2\}$, being the corresponding nodes in the branching tree; we use $I_i$ to refer to both the node as well as the \pname{Independent Set} instance.
    We break our analysis in four cases:
    \begin{enumerate}
        \item If $S_1$ is the first solution, then we made at most $2n$ calls to the $f_\mathcal{G}(n)$-time-algorithm that solves \pname{Independent Set} in the path from root to $I_1$.
        \item If $I_1$ is an ancestor of $I_2$, then we can use a similar reasoning as in the preceding case since $S_2$ is the lexicographically smallest solution of $G \setminus (N[M_1] \cup P_1)$.
        \item Otherwise, let $C = (G, M_C, P_C, t)$ be the lowest common ancestor of $I_1$ and $I_2$.
        By design, $I_1$ is in the left-subtree of $C$ and $I_2$ in the right-subtree.
        Note that, as $I_2$ is not a descendant of $I_1$, it holds that no vertex can be added to $S_1$, i.e., it is maximal, and so $V(G) \setminus (N(S_1) \cup P_i) = \emptyset$, so no recursive calls are made from $I_1$.
        Now, for each node in the path from $I_1$ to $C$, we may have spent $f_\mathcal{G}(n)$-time to check if another right-subtree was viable; these $\bigO{n}$ checks, however, will fail until the algorithm returns to $C$.
        At this point the check will succeed, producing node $R = (G, M_R, P_R, t)$, which is an ancestor of, but is not, $I_2$; the latter follows from the fact that the algorithm only produces an output when discovering a left-subtree, and $R'$ is the root of a right-subtree.
        Now, we can analyze as in the first case: $S_2$ is the first lexicographic solution of $G \setminus (N[M_R] \cup P_R)$ and so is produced within the claimed time.
        \item Finally, if $S_2$ is the last solution, then our recursion will unroll until the root, potentially performing $\bigO{n}$ queries to  the $f_\mathcal{G}(n)$-time-algorithm for \pname{Independent Set}, which again only takes $(f_\mathcal{G}(n) \cdot n)$-time. \qedhere
    \end{enumerate}    
\end{proof}
\section{Parameterizing by the feedback vertex number}
\label{sec:isfvs}

In this section, we show how to apply the invariants of \cref{sec:framework} to \pname{Enum Vertex Cover} parameterized by the size of a given feedback vertex set.
We adapt Jansen and Bodlaender's proof~\cite{vc_fvs} that \pname{Independent Set} has a cubic kernel under our desired parameterization, which immediately implies the same result for \pname{Vertex Cover}.
Interestingly, very little work is needed in order to guarantee that their kernel satisfies \cref{cond:good_decision_kernel}, and the only modification we perform is in the identification of trivial instances of the decision problem, i.e., where the forest obtained by removing the feedback vertex set has a large enough independent set.
Throughout this section, we denote by $(G, X, t)$ the input to \pname{Independent Set}, where we want to decide if there exists an independent set of $G$ of size at least $t$, knowing that $F = 
G \setminus X$ is a forest.
Being given $X$ is not a limitation of our method, as we can compute a 2-approximation of a minimum feedback vertex set of $G$ in polynomial time~\cite{bafna_fvs_apx}. The same can be safely assumed for finding a constant-factor approximation of a modulator to graphs of bounded treedepth or bridgedepth, as discussed in the corresponding articles~\cite{bridgedepth,bougeret_is_td}.
In fact, in this section we may assume a stronger condition on the modulator $X$, which we state in \cref{lem:better_fvs}.
This assumption is used to simplify the application of our framework, but was already explored as~\cite[Lemma 1]{vc_fvs}.

\begin{lemma}[{\cite[Lemma 1]{vc_fvs}}]
    \label{lem:better_fvs}
    Let $X$ be a feedback vertex set of $G$, $(C, H, B)$ be a crown decomposition of $G$ computed using the Nemhauser-Trotter algorithm~\cite{nemhauser_crown}, $G' = G[B]$, $\hat{X} = X \cap B$, $M$ be a maximum matching of $F' \coloneqq G' \setminus \hat{X}$, and $I = F' \setminus V(M)$ be an independent set.
    Then $X \cup I$ is a feedback vertex set of $G$ of size at most $2|X|$ such that $G' \setminus X$ is a forest with a perfect matching.
\end{lemma}

For the entirety of this section, we assume that the initial feedback vertex set $X$ of our input instance already exhibits the property in the conclusion of \cref{lem:better_fvs}.

\subsection{Verifying \cref{cond:good_decision_kernel}}

Our goal is to prove the following lemma, which effectively gives us a good decision kernel to generalize to the enumeration setting.

\begin{lemma}
    \label{lem:vc_fvs_good_decision_kernel}
    \pname{Independent Set} admits a polynomial kernel $\A_1$ of cubic size parameterized by the size of a given feedback vertex set $X$, with core $(X', \id)$, where $X' \subseteq X$ and $\id$ is the identity injection.
\end{lemma}

To obtain $\A_1$, we divide our analysis in two cases: one where the problem is trivial and another where it is not. Recall that, as discussed in \autoref{sec:framework}, for each rule with $(G,X,t)$ as input and $(G',X',t')$ as output, we will prove soundness using $X \cap X'$ paired with the identity function as our a core.
Let us handle the trivial case with the following rule, where $\alpha(G)$ is defined as the size of a maximum independent set of $G$.

\begin{isfvsrule}
    \label{rrule:isfvs_easy}
    If $\alpha(F) \geq t$, then the kernelization algorithm immediately stops and we return the instance $(G[X], X, 0)$.
\end{isfvsrule}

\begin{sproof}{\cref{rrule:isfvs_easy}}
    As $F$ is a forest, one can easily check in linear time using dynamic programming whether $\alpha(F) \geq t$.
    If this is the case, then our kernel is the instance $(G[X], X, 0)$; the choice for $t' = 0$ is just for convenience, and setting $t = 1$ would also work, although with a bit of additional technical details.
    As $(G, X, t)$ and $(G[X], X, 0)$ are \YES-instances, we have our safeness.
 \end{sproof}

\begin{invproof}{\cref{rrule:isfvs_easy}}
Given $S \in \Sol(G, X, t)$, as $S' = S \setminus F$ is such that $S' \in \Sol(G[X], X, 0)$, we obtain that $S'$ is a good trace of the core $X$.
\end{invproof}

We remark that it is not clear how to leverage this capability when designing {\sl strong} \pdkernels{}, as we must guarantee that not only do all good traces survive the compression process, but no non-good trace is allowed to surface.
From now on, we can assume $\alpha(F) < t$.
The remainder of the section boils down to showing that $\A'_1$, \emph{when $\alpha(F) < t$}, satisfies \cref{cond:good_decision_kernel}.
As in~\cite{vc_fvs}, our next step is showing that we can transform the input instance into an equivalent instance $(G', X', t')$ where $F' = G' \setminus X'$ has a perfect matching by making some changes to a given crown decomposition of $G$.
The following rule is a rephrased statement of \cite[Lemma 1]{vc_fvs}; we remark that it is equivalent to the classical crown decomposition reduction rule, where we remove the head and the crown of the graph, keeping only the body.

\begin{isfvsrule}[{\cite[Lemma 1]{vc_fvs}}]
    \label{rrule:perf_matching}
    Let $(C, H, B)$ be a crown decomposition of $G$ computed using the Nemhauser-Trotter algorithm. Return the instance $(G[B], X \cap B, t - |C|)$.
\end{isfvsrule}

\begin{invproof}{\cref{rrule:perf_matching}}
    Let $t'= t - |C|$ and $X' = X \cap B$.
    Observe that, as $G[H \cup C]$ is a crowned graph, every $S \in \Sol(G, X, t)$ satisfies $|S \cap (H \cup C)| \leq |C|$.
    Consequently, $S' \coloneqq S \setminus (H \cup C)$ is a solution of $(G', X', t')$ and, moreover, $S \cap X' = S' \cap X'$.
    As such, for every $Y \subseteq X'$ for which there exists $S^\star \in \Sol(G, X, t)$ with $S^\star \cap X' = Y$ there is some $S'' \in \Sol(G', X', t')$ with $S'' \cap X' = Y$.
\end{invproof}

To simplify our notation, we assume that $(G, X, t)$ is the instance we are working with, and $(G^\star, X^\star, t^\star)$ is the input instance of the kernelization algorithm.

\begin{definition}[{\cite[Definition 2]{vc_fvs}}]
    \label{def:chunks}
    Let $(G, X, t)$ be an instance of \pname{Enum Independent Set}. We say that $Y \subseteq X$ is a \emph{chunk} if $|Y| \leq 2$ and, if $Y = \{a,b\}$, then $ab \notin E(G)$. 
\end{definition}

\begin{definition}[{\cite[Definition 3]{vc_fvs}}]
    \label{def:conflicts}
    The number of \emph{conflicts} of a given $Y \subseteq X$ on a subforest $F' \subseteq F$ is defined as $\conf{F'}{Y} = \alpha(F') - \alpha(F' \setminus N(Y))$.
\end{definition}

The first two reduction rules of \cite{vc_fvs} can be condensed in the following way.
While the statement is very similar, our proof is a strengthened version of theirs, due to \cref{rrule:isfvs_easy}.
This is needed as, in their proof, it is argued that the existence of a solution that has a vertex of $X$ with a large $\Conf$ value implies the existence of another solution that does not use this vertex; this would lead to a replacement strategy and break \cref{cond:good_decision_kernel}.
Concluding that no element of $X$ has a large $\Conf$ value is a key step in the argument bounding the size of the kernel, and so we cannot afford to not perform the removal described in \cref{rrule:isfvs_chunk_removal}.

\begin{isfvsrule}
    \label{rrule:isfvs_chunk_removal}
    If there is a chunk $C$ of $X$ such that $\conf{F}{C} \geq |X|$ then: either remove $C$ from $G$ if $|C| = 1$ or add an edge between the elements of $C$ if $|C| = 2$.
\end{isfvsrule}

\begin{invproof}{\cref{rrule:isfvs_chunk_removal}}
Let $X'=X \setminus \{v\}$ if $C=\{v\}$, or $X'=X$ if $|C|=2$.
We first prove that, if $A \subseteq X$ is independent and has $\conf{F}{A} \geq |X|$, then there is no independent set of $G$ of size at least $t$ that contains $A$.
    By definition, $\alpha(G[A \cup F \setminus N(A)]) = |A| + \alpha(F \setminus N(A))$. Recall that, as \cref{rrule:isfvs_easy} is not applicable, $\alpha(F) < t$; the strict inequality in the last two lines of the following equations follow from this property.
    As such, by our hypothesis that $\conf{F}{A} \geq |X|$, we have:
    \begin{align*}
        \conf{F}{A} &\geq |X|\\
        \alpha(F) - \alpha(F \setminus N(A)) &\geq |X|\\
        t - \alpha(F \setminus N(A)) &> |X| \geq |A|\\
        t &> |A| + \alpha(F \setminus N(A)).
    \end{align*}    
    We can now easily check \autoref{cond:good_decision_kernel}.
    Indeed, given a solution $S$ of $(G, X, t)$, as $S$ does not contain $C$, then $S$ is a solution of $(G,X',t)$ (with the same trace as in $G$). 
\end{invproof}


The next rules are as given in \cite{vc_fvs}, so we only prove their soundness. 

\begin{isfvsrule}
    \label{rrule:isfvs_tree_removal}
    Given $(G,X,t)$ (and $F=G\setminus X)$, if there is a connected component $F'$ of $G[F]$ with $\conf{F'}{Y} = 0$ for every chunk $Y$, then remove $F'$ and set $t \gets t - \alpha(G[F'])$.
\end{isfvsrule}

\begin{invproof}{\cref{rrule:isfvs_tree_removal}}
    Let $(G', X, t')$ be the output of the rule; observe first that $S \in \Sol(G, X, t)$ implies that $S \setminus F' \in \Sol(G', X', t')$.
    Moreover, as $F' \cap X = \emptyset$, no good trace of $X$ is lost upon application of the rule.
\end{invproof}

\begin{definition}
    A pair of vertices $x,y \in V(G) \setminus X$ is \emph{$X$-blockable} if there is some chunk $Y$ of $X$ such that $x,y \in N(Y)$.
\end{definition}

Observe that, if $\{x,y\}$ is \emph{not} $X$-blockable, then $N_X(x) \cap N_X(y) = \emptyset$, otherwise $\{u\} \subseteq N_X(x) \cap N_X(y)$ is a chunk of $X$ with $x,y \in N(u)$.
Moreover, for every $u \in N_X(x)$ and every $v \in N_X(y)$, we have that $uv \in E(G)$, otherwise $Y = \{u,v\}$ is a chunk with $x,y \in N(Y)$.

\begin{isfvsrule}
    \label{rrule:isfvs_red_purple_edges}
    Given $(G,X,t)$ (and $F=G\setminus X)$, 
    if there exists $uv \in E(G[F])$, $u,v$ are not $X$-blockable, and $\deg_F(u), \deg_F(v) \leq 2$, then perform the following operations on $G$ (illustrated in \cref{fig:isfvs_red_purple_edge}): (i) if $u$ has a neighbor $z \in F \setminus \{v\}$, add all edges from $z$ to $N_X(v)$; (ii) if $v$ has a neighbor $w \in F \setminus \{u\}$, add all edges from $w$ to $N_X(u)$; (iii) if $z,w$ exist, then add the edge $tw$; and (iv) delete $u,v$ and set $k \gets k - 1$.
\end{isfvsrule}

\begin{invproof}{\cref{rrule:isfvs_red_purple_edges}}
    Let us show that, if $(G, X, t)$ admits a solution, at least one of them intersects $\{u,v\}$; thus, let $S \in \Sol(G, X, t)$ but $S \cap \{u,v\} = \emptyset$.
    Since $\{u,v\}$ is not $X$-blockable, we have that at most one of $N_X(u), N_X(v)$ intersects $S$.
    \begin{itemize}
        \item If $(N_X(u) \cup N_X(v)) \cap S = \emptyset$, then either $S \cup \{v\}$ or $S \cup \{v\} \setminus \{w\}$ are solutions of size at least $|S|$, depending on whether $w$ exists and is in $S$.
        \item If, however, $N_X(v) \cap S \neq \emptyset$, then either $S \cup \{u\}$ or $S \cup \{u\} \setminus \{z\}$ is a solution of $(G, X, t)$.
        \item The case where $N_X(u) \cap S \neq \emptyset$, is symmetric.
    \end{itemize}
    As we now have that $|\{u,v\} \cap S| = 1$ and no vertex of $X$ was touched upon in the previous analysis, setting $S' = S \setminus \{u,v\}$ guarantees that: (i) we do not exceed the $t' = t-1$ bound, and (ii) the good traces of $X$ are maintained in the output instance.
\end{invproof}
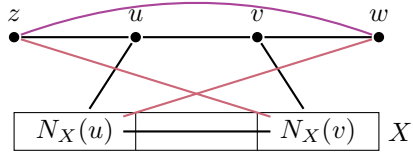
\begin{figure}[!htb]
    \centering
    \begin{tikzpicture}[xscale=0.8, yscale=0.5]
      \GraphInit[unit=3,vstyle=Normal]
      \SetVertexNormal[Shape=circle, FillColor=black, MinSize=2pt]
      \tikzset{VertexStyle/.append style = {inner sep = \inners, outer sep = \outers}}
      \SetVertexLabelOut

      \newcommand{\xcolor}{goodteal}
      \newcommand{\rcolor}{goodcyan}
      \newcommand{\lcolor}{goodyellow}

      \draw (0, -0.5) rectangle (6, 0.5);
      \node (nu) at (1, 0) {$N_X(u)$};
      \draw (2, -0.5) -- (2, 0.5);
      \draw (4, -0.5) -- (4, 0.5);
      \node (nv) at (5, 0) {$N_X(v)$};
      \node at (6.35,0) {$X$};

      \Vertex[x=0, y=2.5, Lpos=90, Math]{z}
      \Vertex[x=2, y=2.5, Lpos=90, Math]{u}
      \Vertex[x=4, y=2.5, Lpos=90, Math]{v}
      \Vertex[x=6, y=2.5, Lpos=90, Math]{w}
      \Edges(z,u,v,w)
      \Edges(v,nv)
      \Edges(u,nu)
      \Edges(nu,nv)
      \SetUpEdge[color=goodred]
      \Edges(z,nv)
      \Edges(w,nu)
      \SetUpEdge[color=goodpurple]
      \Edges[style={bend left}](z,w)
    \end{tikzpicture}
    \caption{Applicable scenario for \cref{rrule:isfvs_red_purple_edges}. The two red edges going from $z,w$ to $X$ and the purple edge $zw$ are the possible edges introduced by the reduction rule. Edges between the two $N(\cdot)$'s indicate that all edges exist between their elements. \label{fig:isfvs_red_purple_edge}}
\end{figure}
\begin{isfvsrule}
    \label{rrule:isfvs_red_edges}
    Given $(G,X,t)$ (and $F=G\setminus X)$, 
    if there are distinct vertices $z,u,v,w \in F$ with $\deg_F(u) = \deg_F(v) = 3$, $N_F(z) = \{u\}$, $N_F(w) = \{v\}$, $uv \in E(G)$, $p \in N_F(u)$, $q \in N_F(v)$, and none of the pairs $\{u,z\}$, $\{v,w\}$, and $\{z,w\}$ are $X$-blockable, do the following, illustrated in \cref{fig:isfvs_red_edges}: (i) add all edges from $p$ to $N_X(z)$; (ii) add all edges from $q$ to $N_X(w)$; and (iii) remove $z,u,v,w$ from $G$ and set $t \gets t-2$.
\end{isfvsrule}

\begin{invproof}{\cref{rrule:isfvs_red_edges}}
    As in \cref{rrule:isfvs_red_purple_edges}, let us just show that if $\Sol(G, X, t) \neq \emptyset$, then there is a solution $S$ containing one of the pairs $\{z,w\}$, $\{z,v\}$, or $\{u,w\}$.
    Suppose that $S$ contains none of these pairs.
    \begin{itemize}
        \item If $S \cap N_X(z) = \emptyset$ and $S \cap N_X(w) = \emptyset$, then $S \cup \{z,w\} \setminus \{u,v\} $ satisfies our constraints, as at most two vertices of $\{u,v,z,w\}$ belong to any solution.
        \item Assuming $S \cap N_X(z) \neq \emptyset$, as $\{z,w\}$ and $\{z,u\}$ are $X$-blockable, we have that $S \cap (N_X(w) \cup N_X(u)) = \emptyset$.
        Thus, $S' = S \cup \{u, w\} \setminus \{p,v\}$ is independent in $G$; now, note that $|S'| \geq |S|$ as at most two vertices of $\{p,u,v,w\}$ can be in an independent set of $G$.
        \item The case where $S \cap N_X(w) \neq \emptyset$ is symmetric to the former.
    \end{itemize}
    It is now a matter of proving that $S' \coloneqq S \setminus \{t,u,v,w\}$ is a solution of the resulting instance. We omit this piece of the prove as it is identical to the original proof, and only remark that $S \cap X = S' \cap X$, so all good traces are preserved by this rule.
\end{invproof}
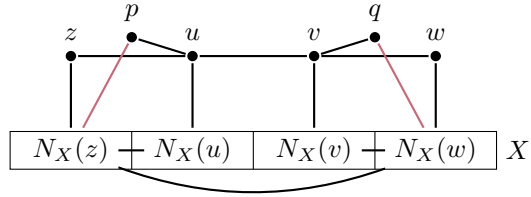
\begin{figure}[!htb]
    \centering
    \begin{tikzpicture}[xscale=0.8, yscale=0.5]
      \GraphInit[unit=3,vstyle=Normal]
      \SetVertexNormal[Shape=circle, FillColor=black, MinSize=2pt]
      \tikzset{VertexStyle/.append style = {inner sep = \inners, outer sep = \outers}}
      \SetVertexLabelOut

      \newcommand{\xcolor}{goodteal}
      \newcommand{\rcolor}{goodcyan}
      \newcommand{\lcolor}{goodyellow}

      \draw (-1, -0.5) rectangle (7, 0.5);
      \node (nt) at (0, 0) {$N_X(z)$};
      \draw (1, -0.5) -- (1, 0.5);
      \node (nu) at (2, 0) {$N_X(u)$};
      \draw (3, -0.5) -- (3, 0.5);
      \node (nv) at (4, 0) {$N_X(v)$};
      \draw (5, -0.5) -- (5, 0.5);
      \node (nw) at (6, 0) {$N_X(w)$};
      \node at (7.35,0) {$X$};

      \Vertex[x=0, y=2.5, Lpos=90, Math, L={z}]{t}
      \Vertex[x=2, y=2.5, Lpos=90, Math]{u}
      \Vertex[x=1, y=3, Lpos=90, Math]{p}
      \Vertex[x=4, y=2.5, Lpos=90, Math]{v}
      \Vertex[x=5, y=3, Lpos=90, Math]{q}
      \Vertex[x=6, y=2.5, Lpos=90, Math]{w}
      \Edges(t,u,v,w)
      \Edges(t,nt)
      \Edges(u,nu)
      \Edges(v,nv)
      \Edges(w,nw)
      \Edges(nt,nu)
      \Edges(nv,nw)
      \Edges[style={bend right}](nt,nw)
      \Edges(u,p)
      \Edges(v,q)
      \SetUpEdge[color=goodred]
      \Edges(p,nt)
      \Edges(q,nw)
    \end{tikzpicture}
    \caption{Applicable scenario for \cref{rrule:isfvs_red_edges}. Red edges from $p,q$ to $X$ are the edges added in $E(G') \setminus E(G)$. Edges between two $N(\cdot)$'s indicate that all edges exist between their elements. \label{fig:isfvs_red_edges}}
\end{figure}
Our compression algorithm consists of applying \cref{rrule:isfvs_easy,rrule:perf_matching} once, followed by exhaustive applications of \cref{rrule:isfvs_chunk_removal}, which is in turn followed by exhaustively applying the other rules.
For the remainder of \cref{sec:isfvs}, let $(H, X', \ell)$ be our output instance.
Note that $X'$ is obtained from $X$ by the vertex removals performed by \cref{rrule:isfvs_chunk_removal}, which belong to no solution of $(G, X, t)$.
As explained at the end of \autoref{sec:framework}, 
$(H, X', \ell)$ is a kernelized instance which is sound when using $X'$ and the identity injection as the core of the compression algorithm, this implies \cref{lem:vc_fvs_good_decision_kernel}.

\subsection{Lexicographic lifting}
\label{sec:isfvs_lex_lifting}

To leverage \cref{thm:framework_kernel}, it remains to show that \cref{cond:decidable_trace,cond:poly_delay} hold for \pname{Enum Independent Set}; we do so in this section.
Recall that $(G^{\star}, X^\star, t^{\star})$ is the instance received by our kernelization algorithm before the application of \cref{rrule:perf_matching}, $(G, X, t)$ is the one after it, and $(H, X', \ell)$ is the compressed instance.
Our overall strategy is to proceed in three steps. First, for each good trace $Y \subseteq X$, we characterize a canonical solution $A_Y \in \Sol(H, X', \ell)$; for our purposes, we define $A_Y$ to be the lexicographically smallest solution of $(H, X', \ell)$ with respect to a given ordering $\sigma$ of $V(H)$ such that $X \cap A_Y = Y$.
Second, given a good trace $Y$, we use \cref{thm:generic_lifting} to generate the set $\S_Y = \{S \in \Sol(G, X, t) \mid S \cap X = Y\}$ with polynomial-delay; we remark that, by being careful, it is possible to perform this step with linear delay as $G \setminus X$ is a forest; intuitively, it is possible to avoid rerunning the $\bigO{|V(G)|}$-time dynamic programming algorithm for \pname{Indpendent Set} in each recursive call of the algorithm of \cref{thm:generic_lifting}, and instead just lookup its partial tables in $\bigO{1}$.
Finally, for each $S \in \S_Y$, we apply a modification of the lifting algorithm of \cite{vc_fvs_strong_pd_kernel} for the \pname{Enum Vertex Cover} problem; in this modification, given an independent set $S$ of $G$ (the body of a crown decomposition, cf. \cref{rrule:perf_matching}) of size at least $t$, we compute the set $\{S^\star \in \Sol(G^\star, X^\star, t^\star) \mid S^\star \cap V(H) = S\}$ with polynomial delay.

\begin{lemma}
    \label{lem:fvs_lift_filter}
    There is a polynomial-time algorithm that, given $(G^{\star}, X^\star, t^{\star})$, $(G, X, t)$, $(H, X', \ell)$, and $Y \subseteq X$: (i) decides if $Y$ is a good trace; and (ii) computes $A_Y \in \Sol(G, X, t)$.
\end{lemma}

\begin{proof}
    Suppose that \cref{rrule:isfvs_easy} was not applied.
    For the first claim, using \cref{thm:generic_lifting}, we can easily check if $(G \setminus (X \cup N_{G}(Y)), t - |Y|)$ has a solution $S_0$; in the affirmative, we may immediately conclude that $Y$ is a good trace, as $S = S_0 \cup Y$ is a solution of size at least $t$ and $S \cap X = Y$, and $S$ can be extended to at least one solution $S^\star \in \Sol(G^{\star}, X^\star, t^{\star})$ by \cite[Lemmas 8, 12, 16]{vc_fvs_strong_pd_kernel}.
    For the second claim, let us apply \cref{thm:generic_lifting} again, but this time to the instance $(H \setminus (X \cup N(Y)), \ell - |Y|)$.
    Let $Z$ be its first output; observe that $A_Y \coloneq Y \cup Z$ is the lexicographically smallest solution of $(H, X', \ell)$ with $A_Y \cap X = Y$, and so it is the canonical solution for the good trace $Y$.

    On the other hand, if \cref{rrule:isfvs_easy} was applied, then observe that $V(H) = X'$ and $G = G^\star$.
    As such, each trace can be verified for goodness as above, and since each trace appears in exactly one solution of $(H, X', t')$ it is trivial to identify a canonical solution.
\end{proof}

We are now ready to define our lifting algorithm.

\begin{lemma}
    \label{lem:fvs_lex_lifiting}
    There is a polynomial-delay algorithm that, given $(G^{\star}, X^\star, t^{\star})$, $(H, X', \ell)$, and a good trace $Y$ of $X' \cap X^\star = X'$, computes $\{S \in \Sol(G^{\star}, X^\star, t^{\star}) \mid S \cap X' = Y\}$.
\end{lemma}

\begin{proof}
    First, we apply \cref{thm:generic_lifting} to $(G \setminus (X \cup N_{G}(Y)), t - |Y|)$, which is guaranteed to find a solution as $Y$ is a good trace.
    Now, for each solution $S$ generated this way, we run the polynomial-delay algorithm of Bougeret et al.~\cite[Lemmas 8, 12, 16]{vc_fvs_strong_pd_kernel} to list all solutions of $(G^\star, X^\star, t^\star)$ that extend $S$. This concludes the proof.
\end{proof}

\cref{thm:fvs_lex_kernel} follows immediately from \cref{thm:framework_kernel}, \cref{lem:vc_fvs_good_decision_kernel,lem:fvs_lift_filter,lem:fvs_lex_lifiting}, and \cite[Theorem 2]{vc_fvs}.

\begin{theorem}
    \label{thm:fvs_lex_kernel}
    There is a cubic \pdkernel{} for \pname{Enum Independent Set} when parameterized by the size of a given feedback vertex set.
\end{theorem}

We remark that it is possible to avoid the algorithm of Bougeret et al.~\cite{vc_fvs_strong_pd_kernel} in the proof of \cref{lem:fvs_lex_lifiting} (and consequently of \cref{thm:fvs_lex_kernel}), but this entails modifications to \cref{lem:better_fvs} and \cref{rrule:perf_matching} to preserve the hypothesis that $G^\star \setminus X^\star$ has a perfect matching.
In particular, instead of removing the crown and head of the decomposition (as done in \cref{rrule:perf_matching}) we would only remove the \emph{unmatched} crown vertices; i.e., we would apply \cref{vcvcrule:unmatched_crown}.
While such a modification is viable, it increases the bound on the size of $X^\star$ from $2|X|$ to $3|X|$, where $X$ is the given feedback vertex set of the input graph, which would lead to an increase in the constants associated with the overall size of the kernel.



\section{Parameterizing by the vertex-deletion distance to \titlemath{c}-treedepth}
\label{sec:vc_td}
\subsection{Bikernelization for enumeration}
\label{sec:gen_bikernel}

A strategy sometimes employed in the design of decision kernels is that of \emph{bikernelization}.
In it, the condition that the input instance $x$ and output instance $z$ of the compression algorithm have to be of the same problem is relaxed, so now it might be that $x$ is an instance of the decision problem $\Pi^\kappa_\enum$ and $z$ of another problem $\Gamma^\rho_\enum$.
Finally, one can observe that, if $\Pi^\kappa \in \NPH$ and $\Gamma^\rho \in \NP$, then there is a reduction that maps $z$ into $y$, which is an instance of $\Pi^\kappa$, thus (indirectly) obtaining a polynomial kernel if $|z|, \rho(z) \leq \poly(\kappa(x))$.

In this strategy, it is quite common to first reduce from $\Pi^\kappa$ to $\Gamma^\rho$, kernelize the resulting instance, and then reduce from the compressed instance of $\Gamma^\rho$ back to $\Pi^\kappa$.
When attempting to adapt this strategy to the enumeration context, however, it becomes clear that the reductions from $\Pi^\kappa_\enum$ to $\Gamma^\rho_\enum$ and from $\Gamma^\rho_\enum$ to $\Pi^\kappa_\enum$ are severely constrained if we restrict ourselves to strong \pdkernels; in particular, we would not be allowed to use the widespread technique of only describing how to translate a \emph{good} solution of $y$ into a solution of $z$, as a strong \pdkernel{} explicitly forbids this.
Moving to our new kernelization model does seem to overcome some of these barriers, a property we exploit in \cref{sec:istd} to design a \pdkernel{} of polynomial size for \pname{Enum Independent Set} parameterized by the size of a modulator to graphs of treedepth at most $c$.
Intuitively, we want to allow ``mistakes'' to be made when going from one problem to the other but, unlike in the \pdkernel{} definition, we need some constraints on how many ``mistakes'' are made to obtain our polynomial bound on the delay.
Interestingly, these requirements are similar, although stronger, to the properties offered by Creignou et al.'s \emph{$e$-reductions}~\cite{creignou_hard}; in particular $e$-reductions allow for a polynomial amount of overlaps in the set of lifted solutions, but this clashes with our kernels' requirements.
We formally define our reductions as follows.

\begin{definition}[$\lambda$-permissive enumerative polynomial parameter transformation ($\lambda$-\ePPT)]
    Let $\Pi^\kappa_\enum$ and $\Gamma^\rho_\enum$ be  parameterized enumeration problems, and $\lambda: \mathbb{N} \mapsto \mathbb{N} \cup \{0, \infty\}$.
    A pair of algorithms $(\A_1, \A_2)$ is a $\lambda$-$\ePPT$ from $\Pi^\kappa_\enum$ to $\Gamma^\rho_\enum$ if, given an instance $x$ of the former and $\kappa(x)$:
    \begin{enumerate}[i.]
        \item $\A_1$ runs in polynomial time and outputs an equivalent instance $y$ of $\Gamma^\rho_\enum$, i.e., $\Sol(x) \neq \emptyset$ if and only if $\Sol(y) \neq \emptyset$, such that $\rho(y) \leq \poly(\kappa(x))$.
        \item When additionally given $(y,\rho(y))$ and $Y \in \Sol(y)$, $\A_2$ either outputs the empty set, or, in \emph{polynomial time} on its input, outputs $S_Y \subseteq \Sol(x)$.
        Moreover, the sets $S_Y$ generated in this manner form a partition of $\Sol(x)$ and at most $\lambda(|x|)$ elements of $\Sol(y)$ output the empty set.
    \end{enumerate}
    We say that a $\lambda$-$\ePPT$ is an $x$-$\ePPT$ for $x \in \mathbb{N} \cup \{0,\infty\}$ if $\lambda(n) = x$ for every $n \in \mathbb{N}$; here $\infty$ is just a symbol to denote unboundedness. 
    $\lambda$-$\ePPT$ is a $\poly$-$\ePPT$ if $\lambda(n) \in \bigO{n^c}$ for some constant $c$.
\end{definition}

\begin{theorem}
    \label{thm:bikernel_piping}
    Let $\Pi^\kappa_\enum$ and $\Gamma^\rho_\enum$ be parameterized enumeration problems. If $\Gamma^\rho_\enum$ admits a \pdkernel{} $(\A_1, \A_2)$ of polynomial size, there exists a $\poly$-$\ePPT$ $(\P_1, \P_2)$ from $\Pi^\kappa_\enum$ to $\Gamma^\rho_\enum$, and an $\infty$-$\ePPT$ $(\G_1, \G_2)$ from $\Gamma^\rho_\enum$ to $\Pi^\kappa_\enum$, then $\Pi^\kappa_\enum$ admits a \pdkernel{} of polynomial size.
\end{theorem}

\begin{proof} 
    Let us show how to build a \pdkernel{} $(\X_1, \X_2)$ for $\Pi^\kappa_\enum$ of polynomial size.
    Take $x$ as an instance of $\Pi^\kappa_\enum$.
    We illustrate the many objects involved in this proof in \cref{fig:bikernel_piping}.
    To obtain $\X_1$, simply observe that $(y, \kappa(y)) = \G_1(\A_1(\P_1(x, \kappa(x))))$ satisfies the polynomial time and size bounds: $(p, \rho(p)) = \P_1(x, \kappa(x))$ can be computed in $\poly(|x| + \kappa(x))$-time; $(a, \rho(a)) = \A_1(p, \rho(p))$ has by definition size bounded by $\poly(\rho(p)) \leq \poly(\kappa(x))$ and can be computed in time bounded by $\poly(|p| + \rho(p)) \leq \poly(|x| + \kappa(x))$; and, finally, $(y, \kappa(y)) = \G_1(a, \rho(a))$ is again by definition computable in time bounded by $\poly(|a| + \rho(a)) \leq \poly(|x| + \kappa(x))$, while $|y| \leq \poly(|a|) \leq \poly(\rho(p)) \leq \poly(\kappa(x))$, as desired.
    The equivalences between the instances follow immediately; thus, setting $\X_1(\cdot) = \G_1(\A_1(\P_1(\cdot))$ is sufficient.

    Now, let us show that the lifting algorithms compose with each other to form $\X_2$.
    To this end, let $Y \in \Sol(y)$.
    If $G_Y = \G_2(a,y,\rho(a),\kappa(y),Y)$ is empty, we set $\X_2(x,y,\kappa(x),\kappa(y),Y)$ as the empty set, so suppose that this is not the case.
    Observe that as $G_Y \neq \emptyset$ and, by the definition of $\G_2$, $G_Y$ can be computed in  $\poly(|a| + |y| + \rho(a) + \kappa(y))$-time, so we may run $\G_2$ until it terminates.
    If, for every $Z \in G_Y$, we have that $\A_2(p,a,\rho(p),\rho(a),Z)$ is the empty set, then $\X_2(x,y,\kappa(x),\kappa(y),Y)$ also outputs the empty set.
    Supposing that this is not the case, let $A_Z = \A_2(p,a,\rho(p),\rho(a),Z)$ be non-empty; as $\A_2$ is a polynomial-\emph{delay} algorithm, we will process $A_Z$ online, i.e., once it generates an output, we will run $\P_2$ on it.
    Recall that $A_Z \subseteq \Sol(p)$ and note that it could still be the case that some elements of $A_Z$, or even all of them, fail to lift back to solutions of $x$.
    If the latter is true, then observe that $|A_Z|$ must be bounded by $\poly(|x| + |p| + \kappa(x) + \rho(p))$ as this is the maximum number of solutions of $p$ that may fail to be lifted back to $x$; if this the case, then observe that we have spent at most $\poly(|x| + |y| + \kappa(x) + \kappa(y))$-time to detect that $Y$ lifts to no solution of $x$, and so we can detect that $\X_2(x,y,\kappa(x),\kappa(y),Y)$ should output the empty set in the allowed time.
    If $A_Z$ only has a few liftable solutions by $\P_2$, then observe that, between any two of them, we may fail to lift at most $\poly(|x| + |p| + \kappa(x) + \rho(p))$ times, so the delay from lifting $Z \in \Sol(a)$ to a set of solutions of $p$ is multiplied by a factor bounded by $\poly(|x| + |p| + \kappa(x) + \rho(p))$ when lifting to solutions of $x$, which is allowed by our time bounds.
    Finally, note that we may argue in a very similar manner when considering failures between two solutions of $A_Z$, and so we have that $\X_2(x,y,\kappa(x),\kappa(y),Y)$ either detects that $Y \in \Sol(y)$ should not be lifted in $\poly(|x| + |y| + \kappa(x) + \kappa(y))$-time, or enumerates a non-empty set of solutions $S_Y \subseteq \Sol(x)$ in $\poly(|x| + |y| + \kappa(x) + \kappa(y))$-delay.

    We point out that we have not argued that the $S_Y$'s form a partition of $\Sol(x)$, but this follows immediately from the guarantees given by the lifting algorithms $\A_2,\P_2$, and $\G_2$.
\end{proof}

  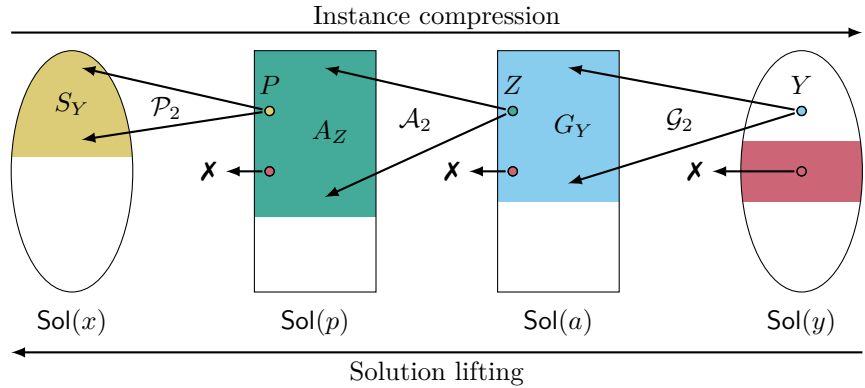
\begin{figure}[!htb]
    \centering
    \begin{tikzpicture}[scale=0.8]
      \GraphInit[unit=3,vstyle=Normal]
      \SetVertexNormal[Shape=circle, FillColor=black, MinSize=2pt]
      \tikzset{VertexStyle/.append style = {inner sep = \inners, outer sep = \outers}}
      \SetVertexLabelOut

      \newcommand{\goodcolora}{goodteal}
      \newcommand{\goodcolorx}{goodyellow}
      \newcommand{\goodcolorc}{goodcyan}
      \newcommand{\badcolora}{goodred}
      
      \draw[thick,-latex] (-7, 2.3) -- (7, 2.3);
      \node at (0, 2.55) {Instance compression};
      
      \draw[thick, -latex] (7, -3) -- (-7, -3);
      \node at (0, -3.35) {Solution lifting};
      \begin{scope}[xshift=-6cm]
        \begin{scope}
            \clip ellipse (1cm and 2cm) (0, 0);
            \begin{scope}
                \clip (-1.2, 2) -- (1.2, 2) -- (1.2, 0.25) -- (-1.2, 0.25) -- cycle;
                \fill[\goodcolorx] (-1, -2) rectangle (1, 2);
            \end{scope}
        \end{scope}
        \draw ellipse (1cm and 2cm) (0, 0);
            
        \node (s0x) at (0, 0.5) {};
        \node (s01x) at (0, 1.75) {};
        \node at (0, 1.125) {$S_Y$};
        \node at (1.55, 1.05) {$\P_2$};
        
        \node at (0, -2.5) {$\Sol(x)$};
      \end{scope}

      \begin{scope}[xshift=-2cm]
        \fill[\goodcolora] (-1, -0.75) rectangle (1, 2);
        \Vertex[x=-0.75, y=1, Lpos=90, Math, L={P}]{s}
        \Vertex[x=-0.75, y=0, Lpos=90, Math, NoLabel]{sbad}
        \node (s0) at (0, -0.5) {};
        \node (s01) at (0, 1.75) {};
        
        \node at (0.25, 0.65) {$A_Z$};
        \node at (1.65, 0.85) {$\A_2$};

        \node (m3) at (-1.75, 0) {\xmark};

        \draw (-1, -2) rectangle (1, 2);
        \node at (0, -2.5) {$\Sol(p)$};
      \end{scope}
      
      \begin{scope}[xshift=2cm]
        \fill[\goodcolorc] (-1, -0.5) rectangle (1, 2);
        
        \node (s0a) at (0, -0.25) {};
        \node (s01a) at (0, 1.75) {};
        \node at (0.25, 0.75) {$G_Y$};
        \node at (2, 0.85) {$\G_2$};

        \draw (-1, -2) rectangle (1, 2);
        \Vertex[x=-0.75, y=1, Lpos=90, Math, L={Z}]{s1}
        \Vertex[x=-0.75, y=0, Lpos=90, Math, NoLabel]{s1m}
        \node (m2) at (-1.75, 0) {\xmark};
        \node at (0, -2.5) {$\Sol(a)$};
      \end{scope}

      \begin{scope}[xshift=6cm]
        \begin{scope}
            \clip ellipse (1cm and 2cm) (0, 0);
            \begin{scope}
                \clip (-1.2, 0.5) -- (1.2, 0.5) -- (1.2, -0.5) -- (-1.2, -0.5) -- cycle;
                \fill[\badcolora] (-1, -2) rectangle (1, 2);
            \end{scope}
        \end{scope}
        
        \Vertex[x=0, y=1, Lpos=90, Math, L={Y}]{sp}
        \Vertex[x=0, y=0, NoLabel]{mp} 
        \draw ellipse (1cm and 2cm) (0, 0);
        
        \node (m1) at (-1.75, 0) {\xmark};
        \node at (0, -2.5) {$\Sol(y)$};
      \end{scope}

      \AddVertexColor{\goodcolora}{s1}
      \AddVertexColor{\goodcolorc}{sp}
      \AddVertexColor{\goodcolorx}{s}
      \Edges[style={-latex}](s1, s0)
      \Edges[style={-latex}](s1, s01)
      \Edges[style={-latex}](s, s01x)
      \Edges[style={-latex}](s, s0x)
      \Edges[style={-latex}](sp, s0a)
      \Edges[style={-latex}](sp, s01a)
      \Edges[style={-latex}](s1m, m2)
      
      \AddVertexColor{\badcolora}{mp,s1m,sbad}
      \Edges[style={-latex}](mp, m1)
      \Edges[style={-latex}](sbad, m3)
    \end{tikzpicture}
    \caption{Example of what the lifting procedure in the proof of \cref{thm:bikernel_piping} does to a solution $Y$ of the compressed instance $y$ of $\Pi^\kappa_\enum$.
    Ellipses correspond to solutions to instances of $\Pi^\kappa_\enum$, while rectangles to solutions to instances of $\Gamma^\rho_\enum$.
    The \pdkernel{} is assumed to be defined between the two rectangles, while between a rectangle and its adjacent ellipse we have a $\poly$-$\ePPT$ (left), or an $\infty$-$\ePPT$ (right).
    The red shaded region in the right-most ellipse represents all solutions of $y$ that were rejected for lifting by $\G_2$, while the shaded regions inside the other shapes correspond to sets of solutions. Vertices and shaded regions are colored associatively: red is used to denote unliftable solutions, while a vertex and a region have a same color if the vertex is lifted to that set of solutions.\label{fig:bikernel_piping}}
    \end{figure}

\begin{corollary}
    \label{cor:bikernel_piping}
    Let $\Pi^\kappa_\enum$ and $\Gamma^\rho_\enum$ be parameterized enumeration problems. If $\Gamma^\rho_\enum$ admits a strong \pdkernel{} of polynomial size, there exists a $0$-$\ePPT$ from $\Pi^\kappa_\enum$ to $\Gamma^\rho_\enum$, and a $0$-$\ePPT$ from $\Gamma^\rho_\enum$ to $\Pi^\kappa_\enum$, then $\Pi^\kappa_\enum$ admits a strong \pdkernel{} of polynomial size.
\end{corollary}

We remark that, in \cref{thm:bikernel_piping,cor:bikernel_piping}, it is possible to ask for a transformation from $\Pi^\kappa_\enum$ to $\Gamma^\rho_\enum$ that allows lifting to be performed with polynomial delay, but no such modification seems to be possible in the transformation in the opposite direction.
Our choice of presentation here is due to the similarity to decision kernelization offered by polynomial-time constraints on all transformation algorithms, restricting the super-polynomial steps only to the kernel itself.
We also point out that, typically, reductions used in decision problems heavily depend on well-behaved solutions that can be easily used to prove the equivalence between the input and output instances, and the nearly-unbounded number of solutions that map to the empty set allowed by $\infty$-$\ePPT$ aims at reflecting this behavior.

To conclude our discussion on compositionality, we point out that, unlike in decision kernelization, \pdkernels{} \emph{do not} compose with each other while guaranteeing the individual size bounds; instead a term that relies on the number of solutions of the intermediate instance emerges, which will typically be of \FPT\ size.
Note that this is essentially the same reason why, in \cref{thm:bikernel_piping}, we required the transformation from $\Pi^\kappa_\enum$ to $\Gamma^\rho_\enum$ to be a $\poly$-$\ePPT$: the bound on the number of unliftable solutions guarantees that no term functionally dependent on the number of solutions of the $\Gamma^\rho_\enum$ instance (the intermediate instance), would appear in the kernel size for $\Pi^\kappa_\enum$.
Finally, we remark that \emph{strong} \pdkernels{} do compose, and this phenomenon arises \emph{precisely} because no bad solutions are allowed to exist in any of the instances.

\subsection{Designing the kernel itself}
\label{sec:istd}

The paper by Bougeret and Sau~\cite{bougeret_is_td} uses a bikernelization strategy to show that \pname{Independent Set} parameterized by the size of a modulator to a graph of treedepth at most $c$, for some fixed integer $c$, admits a kernel of polynomial size.
We use the same strategy, but adapted to the context of enumerative kernelization.
Essentially, we will apply the methodology described in \cref{thm:bikernel_piping}.

\subsubsection{\pname{Annotated Independent Set}}
\label{sec:istd_annotated}
Let us define an annotated version of \pname{Enum Independent Set}, adapted from \pname{Annotated Independent Set}, introduced in~\cite{bougeret_is_td}.

\problenum{Enum Annotated Independent Set}{A graph $G = (X \cup R, E)$, and a set $\Hc$ of hyperedges with ground set $X$, and an integer $t$.}{All independent sets of $G$ of size at least $t$ that contain no $h \in \Hc$.}

In the above definition, we are interested in the case where $X$ is a given modulator of size $k$ to $c$-treedepth; this hypothesis is not a limiting one as we may compute a $2^c$-approximation for a minimum-sized $X$ in polynomial time~\cite{td_modulator_approx}, which is enough for our purposes.
Our goal is to obtain a polynomial kernel in $|X| + |\Hc|$.
First, much as in the feedback vertex set parameterization, we distinguish between what should be the easy cases for decision and show that they are the same for enumerative kernelization.

\begin{istdrule}
    \label{rrule:istd_easy}
    If $\alpha(G[R]) \geq t$, then stop and output $(G[X], \Hc, 0)$.
\end{istdrule}

\begin{sproof}{\cref{rrule:istd_easy}}
    Let $S' \in \Sol(G[X], X, \Hc, 0)$, $R' = R \setminus N_G(S')$, and note that $G[R']$ has constant treedepth. As such, by \cref{thm:generic_lifting}, we can enumerate all independent sets $I \subseteq R'$ of size at least $t - |S'|$ with polynomial-delay and, for each such set, $(I \cup S') \in \Sol(G, X, \Hc, t)$.
\end{sproof}

For the remainder of this section, we suppose that \cref{rrule:istd_easy} is not applicable.
Now, let us generalize \cref{def:conflicts,def:chunks} in a straightforward manner.

\begin{definition}
    Given $X' \subset X$ and $R' \subseteq R$, the number of \emph{conflicts} of $X'$ on $R'$ is defined as $\conf{R'}{X'} \coloneq \alpha(R') - \alpha(R' \setminus N(X'))$.
\end{definition}

\begin{definition}
    A \textit{chunk} of $X$ is an independent subset $Y \subseteq X$ of size at most $2^c$ that contains no element of $\Hc$. The set of all chunks of $X$ is denoted by $\X$.
\end{definition}

We note that the more technically challenging result of~\cite{bougeret_is_td}, namely \cite[Lemma 1]{bougeret_is_td}, is what allows us to only worry about chunks that are subsets of $X$ of size at most $2^c$; we restate it below for completeness.

\begin{lemma}{\cite[Lemma 1]{bougeret_is_td}}
    If $X' \subseteq X$ has $\conf{R'}{X'} > 0$, there there is some chunk $Y \subseteq X'$ with $\conf{R'}{Y} > 0$, regardless of the choice of $R' \subseteq R$.
\end{lemma}

With these guarantees in hand and as \cref{rrule:istd_easy} is not applicable, we are also able to design the following rule.

\begin{istdrule}
    \label{rrule:istd_bad_chunk}
    If there exists $Y \in \X$ with $\conf{R}{Y} > |X|$, then add $Y$ to $\Hc$.
\end{istdrule}

\begin{sproof}{\cref{rrule:istd_bad_chunk}}
    Define $\Hc' = \Hc \cup \{Y\}$ and let us show that $Y$ is not contained in any $S \in \Sol(G, X, \Hc, t)$; towards a contradiction, suppose $Y \subseteq S$ and observe that

    \begin{equation*}
        t \leq |S| = |S \cap X| + |S \cap R| < |X| + (\alpha(G[R]) - |X|) = \alpha(G[R]).
    \end{equation*}

    Where the inequality follows from the fact that $|Y| \leq |S \cap X|$, the fact that $S \cap R \subseteq R \setminus N(Y)$, and the hypothesis that $\conf{R}{Y} = \alpha(G[R]) - \alpha(G[R \setminus N(Y)] > |X|$, which in turn implies $|S \cap R| \leq \alpha(G[R \setminus N(Y)]) < \alpha(G[R]) - |X|$. This is, however, a contradiction to the fact that \cref{rrule:istd_easy} is not applicable.
\end{sproof}

We point out that \cref{rrule:istd_bad_chunk} has the \textit{exact same statement} as Rule~2 of \cite{bougeret_is_td}, but we are able to derive a stronger conclusion due to the previous application of \cref{rrule:istd_easy}. This trend will repeat itself across all of our reduction rules. Besides \cref{rrule:istd_easy}, which has no counterpart in~\cite{bougeret_is_td}, all of our rules have the same number as the corresponding rule of~\cite{bougeret_is_td}.

\begin{istdrule}
    \label{rrule:istd_good_component}
    If there is a connected component $R'$ of $G[R]$ with $\conf{R'}{Y} = 0$ for every $Y \in \X$, then remove $R'$ and set $t \gets t - \alpha(G[R'])$.
\end{istdrule}

\begin{sproof}{\cref{rrule:istd_good_component}}
    Let $(G', X', \Hc, t')$ be the output of the rule; observe first that $S \in \Sol(G, X', \Hc, t)$ implies that $|S \setminus R'| \geq t - \alpha(G[R']) = t'$.
    Now, let $S' \in \Sol(G', X', \Hc, t')$; again by \cref{thm:generic_lifting}, we can generate all independent sets of $G[R']$ of size at least $t - |S'| \leq \alpha(G[R'])$ in polynomial-delay; 
    as shown by Bougeret and Sau (\cite[Lemma 4]{bougeret_is_td}), if every chunk is conflict-free with a component $R'$ of $G[R]$, then every subset of $X$ is also conflict-free towards $R'$, so we are guaranteed to be able to find at least one (maximum) independent set $T$ of $G[R']$ so that $S' \cup T$ is independent and of size at least $t$; as $\Hc$ is unchanged and $\Hc \subseteq 2^{X'}$, $S' \cup T$ is in fact a solution to the original instance.
\end{sproof}

The key combinatorial observation is that, if none of the above rules are applicable, then, for some computable function $f$, $G[R]$ has at most $\bigO{|X|^{f(c) + 2}}$ connected components which, for constant $c$, is a polynomial upper bound on $|X|$. As such, our compression algorithm (\cref{alg:istd_compression}) follows the same outline as the one given by Bougeret and Sau~\cite{bougeret_is_td}.


\begin{algorithm}[!htb]
    \caption{Compression algorithm for \pname{Enum Annotated Independent Set} parameterized by vertex-deletion distance to $c$-treedepth.}\label{alg:istd_compression}
    \begin{algorithmic}[1]
        \Require An instance $(G, X, \Hc', t)$ of \pname{Enum Annotated Independent Set} param. by $|X|$, with $X$ being the modulator to $c$-treedepth, and fixed $c$.
        \Ensure An equivalent instance $(G', X', \Hc', t')$ of with $|V(G')| \leq \poly(|X|)$.
        
        \If{$c = 0$}
            \State \Return $(G, X, t)$.
        \EndIf
        
        \State If \cref{rrule:istd_easy} is applicable to $(G,X,t)$, apply it and return $(G[X], X, 0)$. \Comment{Line 4 can safely be applied only in the first call of the algorithm; cf. \cref{thm:istd_kernel}}
        
        \State Apply \cref{rrule:istd_bad_chunk} exhaustively.

        \State Apply \cref{rrule:istd_good_component}  exhaustively.
        
        \State Set $\Gamma \gets \emptyset$.
        \ForEach{connected component $R'$ of $R$}
            \State Compute an optimal treedepth decomposition $T_{R'}$ of $R'$.
            \State $r_{R'} \gets \roott(T_{R'})$.
            \State $\Gamma \gets \Gamma \cup \{r_{R'}\}$.
        \EndFor
        \State $Z \gets \{uv \in E(G) \mid u \in X, v \in \Gamma\}$.
        \State $G' \gets (V(G), E(G) \setminus Z)$.
        \State \Return \cref{alg:istd_compression}$(G', X \cup \Gamma, \Hc \cup Z, t)$.
        \State \Comment{The returned instance has treedepth at most $c-1$.}
    \end{algorithmic}
\end{algorithm}

\begin{theorem}
    \label{thm:istd_annotated}
    For every fixed integer $c \geq 0$, there is a polynomial-sized \pdkernel{} for \pname{Enum Annotated Independent Set} parameterized by the number of hyperedges $|\Hc|$ and size of a modulator $X$ to $c$-treedepth graphs. In particular, the kernel has $\bigO{|X|^{2^{(c+1)(c+2)/2}}}$ vertices and $|\Hc| + \bigO{|X|^{2^{(c+1)(c+2)/2}}}$ hyperedges.
\end{theorem}

\begin{proof}
    Let $I = (G, X, \Hc, t)$ be our input instance.
    The proof on the size bounds of the kernel is the exact same as in~\cite{bougeret_is_td}, so we omit it for brevity.
    Now, let us concern ourselves with the enumeration questions by first showing that if \cref{rrule:istd_easy} is not applicable at the very first step, then it is not applicable at any point.
    To this end, let $I_1 = (G_1, X_1, \Hc_1, t_1)$ with $R_1 = V(G_1) \setminus X_1$ be an instance and $I_2$, with components defined analogously, be obtained after a single application of one of our reduction rules or the transposition of the roots $\Gamma_1$ to the modulator to $I_1$. If $I_2$ was obtained by an application of \cref{rrule:istd_bad_chunk} to $I_1$, then $t_2 = t_1$, and $G_1[R_1 \setminus X_1] = G_2[R_2 \setminus X_2]$, so $\alpha(G_2[R_2 \setminus X_2]) < t_2$; if $I_2$ was obtained by moving $\Gamma_1$ to $X_1$, then $\alpha(G_2[R_2]) \leq \alpha(G_1[R_1]) < t_1 = t_2$; finally, if $I_2$ is obtained after the application of \cref{rrule:istd_good_component}, with the removal of $R' \subseteq R_1$, then we have that $\alpha(G_2[R_2]) = \alpha(G_1[R_1]) - \alpha(G_1[R']) = t_1 - \alpha(G_1[R']) = t_2$.
    As such, if \cref{rrule:istd_easy} is applicable, then its proof of safeness guarantees that, for every solution $S' \in \Sol(G[X], X, \Hc, 0)$, we can either decide if $S'$ is inextensible to a solution $S$ of $(G, X, \Hc, t)$ in polynomial time, or we can generate all solutions $S \in \Sol(G, X, \Hc, t)$ that have $S \cap X = S'$ in polynomial-delay.

    Now, suppose that \cref{rrule:istd_easy} is not applicable and let $I_{i-1} = (G_{i-1}, X_{i-1}, \Hc_{i-1}, t_{i-1})$ be the instance obtained during the execution of the compression algorithm by an application of \cref{rrule:istd_good_component} to instance $I_i = (G_i, X_i, \Hc_i, t_i)$ and $R'$ the removed component; for convenience, we shorthand $\Sol(G_i, X_i, \Hc_i, t_i)$ by $\Sol(I_i)$.
    Our lifting algorithm is as follows: let $S_{i-1} \in \Sol(I_{i-1})$. By the proof of safeness of \cref{rrule:istd_good_component}, we are guaranteed to find at least one feasible $T_i \subseteq R'$ such that $S_{i-1} \cup T_i$ is independent in $G_i$, contains no hyperedge of $\Hc_i$, and is of size at least $t_i$ by enumerating all independent sets of $R' \setminus N(S_{i-1})$ of size at least $t_i - |S_{i-1}|$.
    By the safeness proof of \cref{rrule:istd_bad_chunk}, a large enough set is independent in the output if and only if it is independent in the input, so no additional operations need to be performed.
    
    Let us show that this lifting algorithm adheres to our requirements. Let $\mathcal{I} = \angled{I_1, \dots, I_r}$ be the sequence of instances produced by our compression algorithm, with $I_r = I$ and $I_1$ being the final output of the compression algorithm; we opt for this ordering of the instances as this is the order upon which our lifting algorithm is applied.
    We prove by induction that every solution of $I_i$ is generated by our algorithm and no two solutions of $I_{i-1}$ generate the same solution of $I_i$.
    By definition, $\Sol(I_1)$ is correctly generated as no lifting steps are performed.
    Now, take $I_i$ with $i > 1$.
    If $I_{i-1}$ is obtained by an application of \cref{rrule:istd_bad_chunk} to $I_i$ or by moving a set of roots $\Gamma$ from $R_{i} = V(G_{i}) \setminus X_{i}$ to $X_{i}$, we are done as $\Sol(I_{i-1}) = \Sol(I_i)$ and the former is correctly computed by the induction hypothesis.
    So suppose that $I_{i-1}$ is obtained by an application of \cref{rrule:istd_good_component} to $I_i$, let $S_i \in \Sol(I_i)$, $R' = R_{i} \setminus R_{i-1}$ be the removed component, and $T = S_i \cap R'$; by the proof of safeness, $S_{i-1} = S_i \setminus T$ is in $\Sol(I_{i-1})$ as $S_{i-1}$ is independent and $|S_{i-1}| = |S_{i}| - |T| \geq t_i - \alpha(G_i[R_i]) = t_{i-1}$.
    Moreover, $|T| + |S_{i-1}| \geq t_{i-1}$, implying that $|T| \geq t_{i-1} - |S_{i-1}|$ and, since $T$ is independent, at some point $T$ will be considered in the algorithm of \cref{thm:generic_lifting} when $S_{i-1}$ is the candidate for extension, which will happen as $\Sol(I_{i-1})$ is correctly computed by induction.
    To complete the proof, note that no two solutions $A, B \in \Sol(I_{i-1})$ generate the same solution $S$ of $I_i$, as we only consider vertices in $V(G_{i}) \setminus V(G_{i-1})$ to add to a solution of $I_{i-1}$, so $A \neq B$ implies that their extensions are also different.
\end{proof}

We note that, aside from \cref{rrule:istd_easy}, our rules actually describe a strong kernel for \pname{Enum Annotated Independent Set}, and so we have the following corollary.

\begin{corollary}
    \label{cor:istd_annotated_strong}
    For every fixed integer $c \geq 0$, there is a strong polynomial-sized \pdkernel{} for \pname{Enum Annotated Independent Set} parameterized by the number of hyperedges $|\Hc|$ and size of a modulator $X$ to $c$-treedepth graphs if $\alpha(G \setminus X)$ is smaller than the target desired independent set size. In particular, the kernel has $\bigO{|X|^{2^{(c+1)(c+2)/2}}}$ vertices and $|\Hc| + \bigO{|X|^{2^{(c+1)(c+2)/2}}}$ hyperedges.
\end{corollary}

\subsubsection{Applying bikernelization to \pname{Enum Independent Set}}
\label{sec:istd_bikernel}

Our goal in this section is to leverage \cref{thm:bikernel_piping} to obtain our \pdkernel{} for \pname{Enum Independent Set}.

\begin{lemma}
    \label{lem:istd_eis_to_eais}
    There is a $0$-$\ePPT$ from \pname{Enum Independent Set} to \pname{Enum Annotated Independent Set}.
\end{lemma}

\begin{proof}
    Let $(G, X, t)$ be an \pname{Enum Independent Set} instance.
    We obtain the instance $(G', X, \Hc, t)$ of \pname{Annotated Independent Set} as follows: $G'$ is a copy of $G$ , but each edge $uv$ with $u,v \in X$ is replaced with the hyperedge $\{u,v\}$, which in turn is added to $\Hc$; we opt for this approach as it is the one used in \cite{bougeret_is_td}.
    The bijection between the solutions of the two instances is immediate: $S \in \Sol(G, X, t)$ if and only if $S \in \Sol(G', X, \Hc, t)$.
\end{proof}

We are now tasked with designing the reverse reduction, with the challenging part being, as expected, how to encode the hyperedges of size greater than two; we currently do not know how to design a $0$-$\ePPT$ in this direction or if it is even possible.

\begin{lemma}
    \label{lem:istd_eais_to_eis}
    There is an $\infty$-$\ePPT$ from \pname{Enum Annotated Independent Set} to \pname{Enum Independent Set}.
\end{lemma}

\begin{proof}    
    We adapt the reduction of~\cite{bougeret_is_td} from \pname{Annotated Independent Set} to \pname{Independent Set}, to an $\infty$-$\ePPT$ from \pname{Enum Annotated Independent Set} to \pname{Enum Independent Set}.
    We start with the instance $I_1 = (G, X, \Hc, t)$ of the former and output the instance $O = (G', X', t')$ of the latter; moreover, we have that $V(G) = X$, as these are the instances produced by our rules.
    First, for each $v_i \in X$, we add to $G'$ a path on three vertices $\angled{a_i, z_i, b_i}$; intuitively, we want to equate having $v_i$ in a solution to $I_1$ with having both $a_i$ and $b_i$ in the solution to $O$.
    Now, for each hyperedge $H \in \Hc$, we add a complete $|H|$-partite graph $W_H$, with part $W_H^i$ corresponding to $v_i \in H$ and having $|X|$ vertices.
    We equate the index of $v_i$ with the superscript of $W_H^i$; for example, if $H = \{v_1, v_4, v_7\}$, then $W_H$ has three parts of size $|X|$: $W_H^1$, $W_H^4$, and $W_H^7$.
    To complete the construction of $G'$, for each $H \in \Hc$ and each $v_i \in H$, we add all edges from $W_H^i$ to both $a_i$ and $b_i$.
    Finally, we set $t' = |X| + t + |X| \cdot |\Hc|$, completing the construction of $O$.

    Instead of arguing that the instances are equivalent, we now show: (i)  a mapping $\sigma: \Sol(I_1) \mapsto \Sol(O)$, and (ii) how to compute $\sigma^{-1}(S')$ in polynomial time for each $S' \in \Sol(O)$, or decide that it is undefined.
    For the first point, given $S \in \Sol(I_1)$, we begin by adding $\{a_i, b_i\}$ to $\sigma(S)$ if $v_i \in S$, otherwise we add $z_i$; now, for each $H \in \Hc$ note that there is some $v_i \in h \setminus S$; among all of these, take $v_i$ with the smallest index and add $W_H^i$ to $\sigma(S)$.
    We claim that $\sigma(S)$ is independent: no edge in a vertex gadget has both endpoints in $\sigma(S)$, for each $H \in \Hc$ there is at most one $i$ for which $W_H^i \subseteq \sigma(S)$ and, if $W_H^i \subseteq \sigma(S)$, then $z_i \in \sigma(S)$.
    Moreover, $\sigma(S)$ is maximal: each multipartite $W_H$ has one of its parts entirely in $S'$ and each vertex gadget has either both endpoints of the path or the central path vertex in $S'$.
    Towards point (ii), let $S' \in \Sol(O)$; if $S'$ is not maximal, we know that it is not the image of any $S \in \Sol(I_1)$.
    Technically, it must be proved that from $S'$ we can derive another solution $S''$ from which we can read some solution of $I_1$; we omit this for brevity as it has already been done in~\cite{bougeret_is_td}.
    Now, suppose $S'$ is maximal, i.e., at least one vertex of each vertex gadget has been picked and each hyperedge gadget $W_H$ has exactly one of its parts contained in $S'$; to see this latter point, note that if $W_H \cap S' = \emptyset$, then $|S'| \leq |X| \cdot (|\Hc| - 1) + 2|X| < t'$.
    Let $W_H^\ell$ be the part of $W_H$ contained in $S'$; if $S' \setminus W_H^\ell$ is \emph{not} maximal in $G' \setminus W_H^\ell$, then there is some part $W_H^j$ that could replace $W_H^\ell$ in $S'$; if $\ell > j$, then we have that $S'$ is not in the image of $\sigma$, and we reject it.
    If this check fails for every $W_H$, either because $S' \setminus W_H^\ell$ is maximal in $G' \setminus W_H^\ell$ or because $\ell < j$, then we have that $S'$ is in fact in the image of $\sigma$; specifically, for each $i \in [|X|]$, we add $v_i$ to $S = \sigma^{-1}(S')$ if $a_i, b_i \in S'$.
    To see that $|S| \geq t$, note that $|S' \cap \bigcup_{h \in \Hc} W_H| = |X| \cdot |\Hc|$ so at least $|X| + t$ vertices of the vertex gadgets must be picked, and as there are only $|X|$ such gadgets, it follows that at least $t$ of them must have two vertices in them, which is only possible by picking the paths' endpoints.
\end{proof}

\begin{theorem}
    \label{thm:istd_kernel}
    For every fixed integer $c \geq 0$, there is a polynomial-sized \pdkernel{} for \pname{Enum Independent Set} parameterized by the size of a modulator $X$ to $c$-treedepth graphs. In particular, the kernel has $\bigO{|X|^{2^{(c+1)(c+2)/2 + 1}}}$ vertices.
\end{theorem}

\begin{proof}
    The proof follows immediately from the fact that \cref{lem:istd_eis_to_eais,thm:istd_annotated,lem:istd_eais_to_eis} together satisfy the hypotheses of \cref{thm:bikernel_piping}.
\end{proof}


\section{Vertex-deletion distance to \titlemath{c}-bridgedepth - proving \cref{thm:dichotomy} }
\label{sec:isbd}

Bridgedepth was introduced by Bougeret, Jansen, and Sau~\cite{bridgedepth} in order to characterize, for which minor-closed families $\mathcal{F}$, \pname{Vertex Cover} admits a polynomial kernel when parameterized by vertex-deletion distance to $\mathcal{F}$, under standard complexity-theoretic hypotheses.
We begin this section by formally defining this parameter.

\begin{definition}[Bridgedepth]
    \label{def:bridgedepth}
    Let $G$ be a graph and $\cb{G}$ be the graph obtained from $G$ by contracting all bridges.
    We say that an induced connected subgraph $T$ of $G$ is a \emph{tree-of-bridges} if every edge of $T$ is a bridge of $G$.
    The \emph{bridgedepth} of $G$, denoted by $\bd{G}$, is recursively defined as follows:
    \begin{itemize}
        \item If $G$ is the graph with no vertices, then $\bd{G} = 0$.
        \item If $G$ has $\ell > 1$ connected components $\{G_1, \dots, G_\ell\}$, then $\bd{G} = \max\{\bd{G_i}\}$.
        \item If $G$ is connected, then $\bd{G} = 1 + \min_{v \in V(\cb{G})}\{\bd{\cb{G} \setminus \{v\}}\}$.
    \end{itemize}
    
\end{definition}

\begin{lemma}[\cite{bridgedepth}, Proposition 3.2]
    \label{lem:bd_properties}
    For any graph $G$, it holds that:
    \begin{enumerate}[i.]
        \item $\bd{G} = 1$ if and only if $G$ is a forest with at least one vertex;
        \item $\bd{\cb{G}} = \bd{G}$; 
        \item Bridgedepth is \emph{minor-closed}: for every minor $G'$ of $G$, we have $\bd{G'} \leq \bd{G}$;
        \item $\bd{G} = 1 + \min_T\{\bd{G \setminus V(T)}\}$, where the minimum runs over all tree-of-bridges of $G$; if $T$ is a tree-of-bridges and $\bd{G} = 1 + \bd{G \setminus V(T)}$, then we say that $T$ is a \emph{lowering tree} of $G$.
        \item For every $X \subseteq V(G)$, $\bd{G} \leq |X| + \bd{G \setminus X}$.
        \item The treewidth of $G$ is upper-bounded by its bridgedepth.
    \end{enumerate} 
\end{lemma}

\subsection{Verifying \cref{cond:good_decision_kernel}}
Let $(G, X, t)$ be our input instance to \pname{Independent Set}, where $G$ is our graph, $X$ is a modulator to $c$-bridgedepth, and $t$ is the minimum size of a solution.
We will tread a very similar path to the one employed for the feedback vertex set parameterization. That is, we will apply the same rules used in~\cite{bridgedepth}, and are able to establish $(X, \id)$ as the core of our kernel.
For the first part, we begin by designing rules to limit the number of connected components of $R$, then present another set of rules and use it to bound the size of a tree-of-bridges of a component.
With this total bound on the size of a tree-of-bridges $T$, we move $T$ to $X$, and proceed by recursively augmenting $X$ until $V(R) = \emptyset$; while this does increase (polynomially) the size of $X$, we guarantee that the algorithm is progressing by the fact that the resulting graph has smaller bridgedepth.
We remark that, while we will show that \cref{cond:good_decision_kernel} is satisfied even for these augmentations of $X$, our core consists only of the initial modulator (minus the irrelevant vertices identified in \cref{rrule:isbd_many_confs}).
The necessary bounds are guaranteed by the results in~\cite{bridgedepth}, and we will omit the proofs of the corresponding lemmas for brevity.

Let us generalize \cref{def:chunks,def:conflicts} one final time.

\begin{definition}
    \label{def:confs_bd}
    Given $X' \subseteq X$ and $R' \subseteq R$, the \emph{number of conflicts} induced by $X'$ on $R'$ on $G$ is defined as:
    \begin{equation*}
        \confg{R'}{X'}{G} = \alpha(R') - \alpha(G[V(R') \setminus N(X')]).
    \end{equation*}
    When $G$ is clear from the context, we opt for the notation $\conf{R'}{X'}$ instead.
\end{definition}

\begin{definition}
    \label{def:chunks_bd}
    Let $(G,X,t)$ be our input instance:
    \begin{itemize}
        \item A \emph{chunk} of $X$ is an independent set of size at most $2^c$, and $\X$ is the set of all chunks of $X$.
        \item The \emph{degree} of a chunk $Y \in \X$, denoted by $\degR(Y)$, is defined as the number of connected components $R'$ of $R$ for which $\conf{R'}{Y} \neq 0$.
        \item A set $Z \subseteq V(R)$ is \emph{free} if $\conf{Z}{Y} = 0$ for every $Y \in \X$.
        \item A set $Z \subseteq V(R)$ is \emph{$y$-almost-free} for $y \in \mathbb{N}$ if, for every chunk $Y$ with $\conf{Z}{Y} > 0$, we have $\conf{R}{Y} \geq y$. Note that the last $\Conf$ has $R$ and not $Z$ in the subscript.
    \end{itemize}
    From \cite[Lemma 6.3]{bridgedepth}, $\Conf$ and all the above properties are computable in polynomial-time.
\end{definition}

\begin{lemma}[\cite{bridgedepth}, Lemma 6.4]
    \label{lem:isbd_contained_chunk}
    Let $(G, X, t)$ be an instance and $R'$ be a connected component of $R = G \setminus X$.
    For every independent set $S_X \subseteq X$ with $\conf{R'}{S_X} > 0$ there exists a chunk $Y \subseteq S_X$ with $\conf{R'}{Y} > 0$.
    Moreover, a set $Z \subseteq V(R)$ is \emph{free} if and only if every independent set $S_X$ of $X$ satisfies $\conf{Z}{S_X} = 0$.
\end{lemma}

As in the treedepth and feedback vertex set parameterizations, we  distinguish the case where $\alpha(R) \geq t$; note that it can be easily checked in polynomial time as $G[R]$ has constant bridgedepth and, consequently, constant treewidth by \cref{lem:bd_properties}.

\begin{isbdrule}
    \label{rrule:isbd_easy}
    If $\alpha(G[R]) \geq t$, then stop and output $(G[X], X, 0)$.
\end{isbdrule}

\begin{sproof}{\cref{rrule:isbd_easy}}
    Let $S' \in \Sol(G[X], X, 0)$, $R' = R \setminus N_G(S')$, and note that $G[R']$ has constant bridgedepth and, consequently, constant treewidth. As such, by \cref{thm:generic_lifting}, we can enumerate all independent sets $I \subseteq R'$ of size at least $t - |S'|$ with polynomial delay and, for each such set, $I \cup S' \in \Sol(G, X, 0)$.
\end{sproof}

From this point forward, we assume that \cref{rrule:isbd_easy} is not applicable to $(G, X, t)$, so $\alpha(G[R]) < t$.
Beyond our conditions on \cref{sec:framework}, we also need to  enforce the following properties.

\begin{invariant}
    \label{inv:isbd_no_lonely_r_is}
    If $\alpha(G \setminus X) < t$, then, after applying a reduction rule and obtaining $(G', X', t')$, it holds that $\alpha(G' \setminus X') < t'$.
\end{invariant}


\begin{invariant}
    \label{inv:isbd_controlled_bd}
    Let $(G', X, t')$ be obtained from $(G, X, t')$ by an application of a reduction rule. It holds that $\bd{G' \setminus X} \leq \bd{G \setminus X}$.
\end{invariant}

Interestingly, the assumption that $\alpha(R) < t$ is used to both simplify the arguments of some of our reduction rules as well to ensure that we have \cref{cond:good_decision_kernel}.
Moreover, by guaranteeing \cref{inv:isbd_no_lonely_r_is}, we also ensure that \cref{rrule:isbd_easy} can only be applied at the very beginning of the kernelization algorithm.
\cref{inv:isbd_controlled_bd} is primarily useful for our lifting algorithm, but it must be enforced throughout the compression phase.
The following is a helpful lemma that is extensively used in our safeness proofs.

\begin{lemma}
    \label{lem:isbd_large_alpha_r}
    If there exists a solution $S$ of $(G, X, t)$ with $S_X = S \cap X$ such that $\conf{R}{S_X} \geq |S_X|$, then $\alpha(R) \geq t$.
\end{lemma}

\begin{proof}
    Note that:
    \begin{align*}
        \confg{R}{S_X}{G} &= \alpha(R) - \alpha(R \setminus N_G(S_X)) \geq |S_X|\\
        \alpha(R) &\geq \alpha(R \setminus N_G(S_X)) + |S_X| \geq |S| \geq t. \qedhere
    \end{align*}
\end{proof}

\subsubsection{Bounding the number of connected components of $R$}

\begin{isbdrule}
    \label{rrule:isbd_comp_removal}
    If $R' \subseteq R$ is a free connected component of $G \setminus X$, remove $R'$ and set $t \gets t - \alpha(R')$.
\end{isbdrule}

\begin{invproof}{\cref{rrule:isbd_comp_removal}}
    Note that for any $S \in \Sol(G, X, t)$, $S' = S \setminus V(R')$ is an independent set of $G' = G \setminus R'$ of size at least $t' = t - \alpha(R')$ as $|S \cap V(R')| \leq \alpha(R')$ and $|S| \geq t$.
    Moreover, observe that $S \cap X = S' \cap X$, so \cref{cond:good_decision_kernel} is preserved.
    For \cref{inv:isbd_controlled_bd}, it suffices to recall that bridgedepth is minor-closed and observe that $G \setminus R'$ is a minor of $G$.
    Note that, if $\alpha(G' \setminus X) \geq t' = t - \alpha(R')$, then $G' \cup R' \setminus X = G \setminus X = R$, which implies that $\alpha(R) \geq t' + \alpha(R') \geq t$, so $\alpha(G' \setminus X) < t'$ as \cref{rrule:isbd_easy} is not applicable and we preserve \cref{inv:isbd_no_lonely_r_is}.    
\end{invproof}

The following rule investigates the case where a chunk of the modulator has a large number of conflicts.

\begin{isbdrule}
    \label{rrule:isbd_many_confs}
    If there is a connected component $R'$ of $R$ such that, for every $Y \in \X$ with $\conf{R'}{Y} > 0$ we have $\degR(Y) \geq |X|+1$, then remove all edges between $X$ and $R'$.
\end{isbdrule}

\begin{invproof}{\cref{rrule:isbd_many_confs}}
    Take $S \in \Sol(G, X, t)$ is also in $\Sol(G', X, t)$ as $G' = G \setminus E(X, R')$ is a subgraph of $G$ with as many vertices; consequently, if $S_X \subseteq X$ is in a solution to $(G, X, t)$, then it is in a solution to $(G', X, t)$ as well, satisfying \cref{cond:good_decision_kernel}, while \cref{inv:isbd_controlled_bd} follows immediately as bridgedepth is minor closed and $G'[R] = G[R]$.
    Towards proving that \cref{inv:isbd_no_lonely_r_is} holds, suppose there is some $S' \in \Sol(G', X, t)$, let $S'_X = S' \cap X$, and $S'_{R'} = S \cap R'$, and further assume that $\conf{R'}{S'_X} > 0$.
    Then \cref{lem:isbd_contained_chunk} guarantees that there is some chunk $Y \subseteq S'_X$ and, as the rule is applicable, we have that:

    \begin{equation*}
        \conf{R' \setminus R}{S'_X} \geq \conf{R \setminus R'}{Y} \geq \degR(Y) - 1 \geq |X|.
    \end{equation*}

    By \cref{lem:isbd_large_alpha_r}, $\alpha(G[R]) \geq t$, which contradicts the inapplicability of \cref{rrule:isbd_easy}, so it follows that $S'$ cannot exist with $\conf{R'}{S'_X} > 0$ and \cref{inv:isbd_no_lonely_r_is} is preserved.
\end{invproof}

\begin{lemma}[\cite{bridgedepth}, Lemma 6.6]
    \label{lem:isbd_bounded_comp_count}
    Let $(G, X, t)$ be an instance of \pname{Enum Independent Set} and $R = G \setminus X$.
    If \cref{rrule:isbd_comp_removal,rrule:isbd_many_confs} are not applicable to the instance, then $R$ has at most $|\X|\cdot|X|$ connected components.
\end{lemma}

\subsubsection{Bounding the size of the lowering trees}

The last set of reduction rules is devoted to limiting the size of a lowering tree of a connected component.
This is done in three steps: first, we bound the diameter of the tree, then bound its maximum degree, and finally its number of leaves.
Intuitively, we look for simple objects, called \emph{$T$-structures}, inside a lowering tree $T$ of a connected component $R'$ of $R$ and locally alter them without changing too much how they behave towards the modulator $X$.
Once this is accomplished, we move $T$ to $X$ and recursively kernelize the resulting instance.
We rely on the following definitions, first given in~\cite{bridgedepth}, for the remainder of our rules.

\begin{definition}
    \label{def:isbd_vertex_typing}
    Let $R' \subseteq R$ be a connected component of $G \setminus X$, $T \subseteq R$ be a tree-of-bridges of $R'$, and $R^\star = R' \setminus E(T)$.
    The \emph{pending component} of $v \in V(R')$, denoted by $H^T_v$ or $H_v$ when $T$ is clear from the context, is the connected component of $R^\star$ that contains $v$.
    We say that $v$ is the \emph{root} of $H_v$ and  that:
    \begin{itemize}
        \item vertex $v$ has \emph{type $A$} in $T$ if $\alpha(R'[V(H_v)]) = \alpha(R'[V(H_v) \setminus \{v\}])$;
        \item vertex $v$ has \emph{type $B$} in $T$ if $\alpha(R'[V(H_v)]) = \alpha(R'[V(H_v) \setminus \{v\}]) + 1$.
    \end{itemize}
    To simplify our notation, we use $H_v$ for both the induced subgraph and for the vertex set, when no ambiguity is present.
\end{definition}

\begin{definition}
    \label{def:isbd_struct_typing}
    Let $R' \subseteq R$ be a connected component, and $T \subseteq R$ be a tree-of-bridges of $R'$.
    The set of vertices $C \subseteq V(R')$ is a \emph{$T$-conflict structure}:
    \begin{itemize}
        \item of \emph{type 1} if $C = H_{v_1} \cup H_{v_2}$ where $v_1v_2 \in E(T)$ and at least one of $v_1,v_2$ is of type $A$.
        \item of \emph{type 2} if $C = H_{v_1} \cup H_{v_2}$ and:
            \begin{itemize}
                \item there exists $u_1,u_2 \in V(T)$ such that $\angled{u_2,v_1,v_2,u_1}$ is a path in $T$;
                \item $\deg_T(v_1) = \deg_T(v_2) = 2$; and
                \item both $v_1,v_2$ are of type $B$ in $T$.
            \end{itemize}
        \item of \emph{type 3} if $C = H_u$ and $u$ is a leaf of $T$ of type $B$.
        \item of \emph{type 4} if $C = H_{v_1} \cup H_{v_2}$ and:
            \begin{itemize}
                \item there exists $u \in V(T)$ such that $\angled{v_1,v_2,u}$ is a path in $T$;
                \item $v_1$ is a leaf in $T$, $\deg_T(v_2) = 2$; and
                \item both $v_1,v_2$ have type $B$ in $T$.
            \end{itemize}
    \end{itemize}
\end{definition}

Note that all the properties listed in \cref{def:isbd_vertex_typing,def:isbd_struct_typing} can be computed in polynomial time as $R$ has constant bridgedepth; by \cref{lem:bd_properties}, $R$ has constant treewidth, and an \FPT\ algoritm for \pname{Maximum Independent Set} parameterized by treewidth is a textbook example.
We opt to not state the more technical lemmas of~\cite{bridgedepth} as they are used in the formal proofs of the kernelization size bounds, which we use as a black box; intuitively, they use the fact that if a chunk has a lot of conflicts, then it must interact with $T$-conflict structures, and so the latter must be $y$-almost-free for $y \geq |X| + \Theta(1)$, and if no such chunk exists then $T$ is of bounded size.
The next two rules are used to bound the diameter of a tree-of-bridges.

\begin{isbdrule}
    \label{rrule:isbd_type1}
    Let $T$ be a tree-of-bridges of a connected component $R' \subseteq R$.
    If $T$ contains a $T$-conflict structure $C = H_{v_1} \cup H_{v_2}$ of type 1 that is $(|X| + 2)$-almost-free, then remove edge ${v_1}{v_2}$ from $G$.
\end{isbdrule}

\begin{invproof}{\cref{rrule:isbd_type1}}
    It is immediate to verify that, for every $S \in \Sol(G, X, t)$ is also in $\Sol(G', X, t)$ as $G' = G \setminus \{{v_1}{v_2}\}$ is a subgraph of $G$ on the same vertex set; consequently, if $S_X 
    = S \cap X$ is in a solution to $(G, X, t)$, then it is in a solution to $(G', X, t)$, safeguarding \cref{cond:good_decision_kernel}.
    For \cref{inv:isbd_controlled_bd}, note that $G[R] \setminus \{v_1v_2\}$ is a minor of $G[R]$.
    Towards proving that \cref{inv:isbd_no_lonely_r_is} holds, suppose that this is not that case, so take $S' \in \Sol(G', X, t)$ with $S' \subseteq V(R)$.
    If at most one of ${v_1},{v_2} \in S'$, then $S' \in \Sol(G, X, t)$, so assume now that both are in $S'$ and moreover, $v_1$ is of type $A$.
    That is, there is some independent set $Z$ of $H_{v_1}$ of size $\alpha(G[H_{v_1}]$, and we can set $S = (S' \setminus H_{v_1}) \cup Z$; observe that $S$ is independent as $H_{v_1}$ is a pending component, and $|S| \geq |S'| > t$, which contradicts the hypothesis that \cref{inv:isbd_no_lonely_r_is} held for $(G, X, t)$.
\end{invproof}

\begin{isbdrule}
    \label{rrule:isbd_type2}
    Let $T$ be a tree-of-bridges of a connected component $R' \subseteq R$.
    If $T$ contains a $T$-conflict structure $C = H_{v_1} \cup H_{v_2}$ of type 2 where $\angled{u_2,v_1,v_2,u_1}$ is the corresponding path and $C$ is $(|X| + 1)$-almost-free, then identify $v_i$ into $u_i$, for $i \in \{1,2\}$, and set $t \gets t - 1$.
\end{isbdrule}

\begin{invproof}{\cref{rrule:isbd_type2}}
    Let $(G', X, t-1)$ be the instance obtained after applying the rule and define $U = \{u_1,u_2\}$.
    To show that \cref{cond:good_decision_kernel} holds, take $S \in \Sol(G, X, t)$; we denote by $S_W = S \cap W$ for $W \subseteq V(G)$.
    We branch the analysis as follows:
    \begin{enumerate}
        \item If $S \cap U = \emptyset$, $|S \setminus \{v_1,v_2\}| \geq t - 1$ and we are done.
        \item Suppose that $S \cap U = \{u_i\}$. If $v_i \in S$, then $S' = S \setminus \{v_i\}$ is a solution to $(G', X, t-1)$; if $v_i \notin S$, then $S' = S \setminus \{u_i\}$ is again a solution to the output instance.
        \item Otherwise, we have $U \subseteq S$. We once again branch the analysis:
        \begin{enumerate}
            \item If $\conf{C}{S_X} = 0$, then there is some independent set $Z \subseteq C$ of size $\alpha(G[C])$ and $Z \cup S_X$ is independent in $G$.
            As both $v_i$'s have type $B$, we have that $|Z| = \alpha(G[H_{v_1}]) + \alpha(G[H_{v_2}]) - 1$, as only one of $v_1,v_2$, say $v_1$, is in $Z$.
            As $v_1$ is the only vertex in $Z$ that has a neighbor outside of $C$, which is precisely $u_2$, it follows that $S^\star = Z \cup S_X \cup S_{R \setminus R'} \cup  (S_{R' \setminus C} \setminus \{u_2\})$ is independent and of size at least $t$.
            The latter follows from the fact that $U \subseteq S$ implies $v_1,v_2 \notin S$, which in turn means that $|S_C| < \alpha(G[C]) = |Z|$ since $C$ is a type 2 $T$-structure; consequently, by replacing $S_C \cup \{u_2\}$ with $Z$ to obtain $S^\star$, we essentially at least replaced $u_2$ with $v_1$.
            Finally, $S' = S^\star \setminus \{v_1\}$ is a solution to $(G', X, t-1)$.
            A symmetric analysis can be made if $v_2 \in Z$.
            \item If $\conf{C}{S_X} > 0$, then there is some chunk $Y \subseteq S_X$ and, since the rule is applicable, we have $\conf{C}{Y} \geq |X|+1 \geq |Y|$. By \cref{lem:isbd_large_alpha_r}, this implies that $\alpha(R) \geq t$, which contradicts the fact that \cref{rrule:isbd_easy} is not applicable.
        \end{enumerate}
    \end{enumerate}
    As the above analysis never modified $S_X$, \cref{cond:good_decision_kernel} is satisfied.
    Towards proving that \cref{inv:isbd_controlled_bd} holds, first observe that there is a sequence of lowering trees $\T = \angled{T_1, \dots, T_\ell}$ that witnesses that $R'$ has bridgedepth $\ell \leq c$; moreover, every vertex of $R'$ must be in exactly one $T_i$ and, w.l.o.g., we may assume that every $T_i$ is inclusion-wise maximal (cf. \cite[proof of Proposition 3.2]{bridgedepth}).
    Thus, as $T$ is a tree-of-bridges of $R'$, there exists $T_i \in \T$ that contains $T$; let $R'_i = R' \setminus \bigcup_{j = 1}^{i - 1} V(T_i)$, $T''_i$ be obtained from $T_i$ by applying the operations given in the statement of the rule to $\angled{u_2,v_1,v_2,u_1}$, and $R''$, $R''_i$ be defined analogously.
    Now, as $R'_i \setminus V(T_i) = R''_i \setminus V(T''_i)$ and $T_i$ is a lowering tree of $R'_i$, it follows that $\bd{R''_i} \leq 1 + \bd{R''_i \setminus T''_i} = 1 + \bd{R'_i \setminus T_i} = \bd{R'_i}$.
    Thus, by replacing $T_i$ with $T''_i$ in $\T$, we obtain a sequence of lowering trees that witnesses that $\bd{G'} \leq \bd{G}$, as $\angled{T_1, \dots, T_{i-1}}$ is also a sequence of lowering trees of $G'$ and $R''_j = R'_j$ for every $j > i$.
    This concludes the proof that \cref{inv:isbd_controlled_bd} is satisfied.

    Let us prove that \cref{inv:isbd_no_lonely_r_is} holds along with the converse direction.
    To this end, take $S' \in \Sol(G', X, t-1)$ and recall that $u_1,u_2$ are adjacent in $G'$.
    As such, if $S' \cap U = \{u_1\}$, then $S = S' \cup \{v_1\}$ is a solution to $(G, X, t)$: no vertex of $N_G(v_1)$ is in $S$ as $N_G(v_1) \setminus \{v_2\} \subseteq N_{G'}(v_1)$, $v_2 \notin S$, and $|S| = |S'| + 1 \geq t$; note that, $S' \cap X \neq \emptyset$, otherwise we would have $\alpha(R) \geq t$ and contradict \cref{rrule:isbd_easy}.

    For the remainder, we have $S' \cap U = \emptyset$ and let us define $H'_{v_i} = H_{v_i}  \setminus \{v_1\}$, $V' = R' \setminus (C \cup U)$, i.e., $N_G(V') \cap C = \emptyset$, and partition $S' = S'_X \cup S'_{H'_{v_1}} \cup S'_{H'_{v_2}} \cup S'_{V'} \cup S'_{R \setminus R'}$.
    As both $v_i$'s are of type $B$ and $S' \cap U = \emptyset$, it follows that $|S'_{H'_{v_1}} \cup S'_{H'_{v_2}}| \leq \alpha(G[C]) - 1$.
    We again have two cases:
    \begin{enumerate}
        \item If $\confg{C}{S'_X}{G} = 0$, then we can take an independent set $Z$ of $G[C]$ with $|Z| = \alpha(G[C]) > |S'_{H'_{v_1}} \cup S'_{H'_{v_2}}|$, which implies that $S = S'_X \cup Z \cup S_{V'} \cup S_{R \setminus R'} = (S' \setminus S'_{H'_{v_1}} \cup S'_{H'_{v_2}}) \cup Z$ is independent and so is a solution to $(G, X, t)$. Once again, $S'_X \neq \emptyset$, as otherwise we would have $S_X = \emptyset$ and $\alpha(G[R]) \geq t$, so \cref{inv:isbd_no_lonely_r_is} holds in this case.
        \item If $\confg{C}{S'_X}{G} \neq 0$, then by \cref{lem:isbd_contained_chunk} there is a chunk $Y \subseteq S'_X$ and, by the applicability of the rule, $\confg{R}{Y}{G} \geq |X|+1$, which implies that the condition of \cref{lem:isbd_large_alpha_r} is satisfied, in turn contradicting that \cref{rrule:isbd_easy} is not applicable, so it follows that \cref{inv:isbd_no_lonely_r_is} is preserved. \qedhere
    \end{enumerate}
\end{invproof}

\begin{lemma}{\cite[Lemma 6.13]{bridgedepth}}
    \label{lem:isbd_bounded_diam}
    Let $R'$ be a connected component of $R$, $T$ be a lowering tree of $R'$, and $T'$ be a longest path of $T$.
    If \cref{rrule:isbd_type1,rrule:isbd_type2} are not applicable to $(G,X,t)$ when taking $T'$ as the tree-of-bridges, then $V(T') \leq \bigO{|\X| \cdot |X|}$ and, consequently, the diameter of $T$ has the same bound.
\end{lemma}

We now move on to rules that will allow us to (polynomially) bound the maximum degree and number of leaves of lowering trees.
We accomplish the first with the following rule.

\begin{isbdrule}
    \label{rrule:isbd_back-n-forth}    
    Let $T$ be a tree-of-bridges of $R' \subseteq R$.
    If $T$ contains a vertex $v$ with $\deg_T(v) > 3|\X| \cdot |X|$ then apply \cref{rrule:isbd_comp_removal,rrule:isbd_many_confs} exhaustively to $(G, X \cup \{v\}, t)$, obtaining $(G', X \cup \{v\}, t')$, and return $(G', X, t')$.
\end{isbdrule}

\begin{invproof}{\cref{rrule:isbd_back-n-forth}}
    The safeness of the rule follows immediately from the safeness of \cref{rrule:isbd_comp_removal,rrule:isbd_many_confs}; for more details, we refer to \cite[Lemma 6.14]{bridgedepth}.
    What is missing from the original proof are the guarantees that \cref{cond:good_decision_kernel,inv:isbd_no_lonely_r_is,inv:isbd_controlled_bd} are maintained.
    Define $X^\star = X \cup \{v\}$.
    By the proofs of \cref{rrule:isbd_comp_removal,rrule:istd_bad_chunk}, we know that for every $S \in \Sol(G, X, t)$, $S_{X^\star} = S \cap X^\star$ is also in at least one solution $S'$ of $(G', X^\star, t')$, implying that the same is true for  $S_X = S_{X^\star} \setminus \{v\}$ with respect to the instance $(G', X, t')$, so \cref{cond:good_decision_kernel} is preserved.
    \cref{inv:isbd_controlled_bd} also holds: as $G'$ is a minor of $G$ and bridgedepth is minor-closed, $\bd{G' \setminus X} \leq \bd{G \setminus X}$.
    
    Towards a contradiction, let $R' = G' \setminus X$, and suppose that $\alpha(R') \geq t'$.
    If $\alpha(R' \setminus \{v\}) \geq t'$, then by  the proofs of \cref{rrule:isbd_comp_removal,rrule:istd_bad_chunk}, we would also have that $\alpha(R) \geq t$, a contradiction to the inapplicability of \cref{rrule:isbd_easy}.
    If, however, $\alpha(R' \setminus \{v\}) \leq t' - 1$ but $\alpha(R') \geq t'$, then it must be the case that $\alpha(R') = t'$, i.e., there is some solution $S' \in \Sol(G', X, t)$ such that $S' \cap X^\star = \{v\}$; by the proofs of \cref{rrule:isbd_comp_removal,rrule:istd_bad_chunk}, there is some solution $S$ of $(G, X, t)$ that has $S \cap X = S' \cap X$, but $S' \cap X = \emptyset$, so $S \subseteq R$ and so $\alpha(R) \geq t$, and we have our contradiction, implying that \cref{inv:isbd_no_lonely_r_is} is maintained.
\end{invproof}

\begin{observation}
    \label{obs:isbd_t_deg_bound}
    If $T$ is a tree-of-bridges to which \cref{rrule:isbd_back-n-forth} is not applicable, then $T$ has maximum degree at most $3|\X|\cdot|X|$.
\end{observation}

In the following, we assume that all involved trees have at least three vertices and, consequently, at least one non-leaf vertex, as otherwise they already have bounded size. Moreover, we assume they are rooted at a non-leaf vertex; this can be done arbitrarily.

\begin{isbdrule}
    \label{rrule:isbd_type3}
    Let $T$ be a tree-of-bridges of a connected component $R' \subseteq R$.
    If $T$ contains a $T$-conflict structure $C = H_u$ of type 3, i.e., $u$ is a leaf of $T$ and of type $B$, and $C$ is $(|X| + 2)$-almost-free, then remove edge $uv$, remove its parent $v$ in $T$ from $G$, and set $t \gets t - 1$.
\end{isbdrule}

\begin{invproof}{\cref{rrule:isbd_type3}}
    Let $(G', X, t-1)$ be the instance obtained after applying the rule.
    Towards proving \cref{cond:good_decision_kernel}, take any $S \in \Sol(G, X, t)$ and note that it satisfies $|S \setminus \{u,v\}| \geq |S| - 1$ as $uv \in E(G)$; as $S \cap X$ remains unaltered, we are done.
    Furthermore, $G'$ is a minor of $G$, so \cref{inv:isbd_controlled_bd} holds.
    To prove that \cref{inv:isbd_no_lonely_r_is} holds, take $S' \in \Sol(G', X, t-1)$, assume that $S' \cap X = \emptyset$ and define $S'_W = S' \cap W$ for $W \subseteq V(G)$.
    As $S'_X = \emptyset$ and $\confg{C}{S'_X}{G} = 0$, we may take an independent set $Z \subseteq H_u$ of size $\alpha(G[C])$, which satisfies $u \in Z$ and $|Z| > |S'_C|$ as $u$ has type $B$.
    As $v \notin S'$ is the sole neighbor of $H_u$ in $R$, $Z$ is independent in $G$ and, by our previous observation, $|Z| \geq |S'| + 1 \geq t$, implying that $\alpha(R) \geq t$ and contradicting the hypothesis that \cref{inv:isbd_no_lonely_r_is} holds for the input instance.
\end{invproof}

\begin{isbdrule}
    \label{rrule:isbd_type4}
    Let $T$ be a tree-of-bridges of a connected component $R' \subseteq R$.
    If $T$ contains a $T$-conflict structure $C = H_{v_1} \cup H_{v_2}$ of type 4, i.e., $v_1$ is a leaf of $T$, $v_2$ is its parent, both are of type $B$, $N_T(v_2) = \{v_1, u\}$, and $C$ is $(|X| + 1)$-almost-free, then remove $v_2$, identify $v_1$ into $u$, and set $t \gets t - 1$.
\end{isbdrule}

\begin{invproof}{\cref{rrule:isbd_type4}}
    Let $(G', X, t-1)$ be the instance resulting from the application of the rule.
    As in \cite[Meta-Rule 5]{bridgedepth}, we will use the rule for type 2 $T$-conflict structures in our proof; we refer to \cref{fig:isbd_type4} for a visual aid on the graphs and operations involved in the proof.
    To this end, let $x,y$ be new vertices, and $G_1$ be such that $V(G_1) = V(G) \cup \{x,y\}$, $E(G_1) = E(G) \cup \{xy, yv_1\}$, and $R'_1 = R' \cup \{x,y\}$.
    Note that $xy,yv_1$ are bridges of $R_1$, so $T_1 = T \cup \{xy,yv_1\}$ is a tree-of-bridges of $R_1$.
    Now, observe that $v_1, v_2$ still have type $B$ since every connected component of $R \setminus E(T)$ is also a connected component of $R_1 \setminus E(T_1)$.
    Moreover, $\deg_{T_1}(v_1) = \deg_{T_1}(v_2)$, and so it follows that $C$ is a type 2 $T_1$-conflict structure and remains $(|X|+1)$-almost-free.
    As such, we can apply \cref{rrule:isbd_type2} to $(G_1, X, t+1)$ using $T_1$ as our tree-of-bridges to obtain $(G_2,X,t)$; in particular, we identified $v_2$ into $y$ and $v_1$ into $u$ in $G_1$ to obtain $G_2$.
    By the soundness of \cref{rrule:isbd_type2}, the subsets of $X$ that are contained in solutions of $(G_1, X, t+1)$ are preserved in $(G_2, X, t)$, and $\alpha(G_2 \setminus X) < t$.
    Now, we remove $N_{G_2}[x] = \{x,y\}$ from $G_2$ to obtain the instance $(G_3, X, k-1)$, which is equivalent to $(G_2, X, t)$.
    Observe that these operations preserve \cref{cond:good_decision_kernel,inv:isbd_no_lonely_r_is,inv:isbd_controlled_bd}, the latter following by the fact that $G_3$ is a minor of $G_2$.
    Finally, observe that $G_3$ is isomorphic to $G'$ and the proof is complete.
\end{invproof}

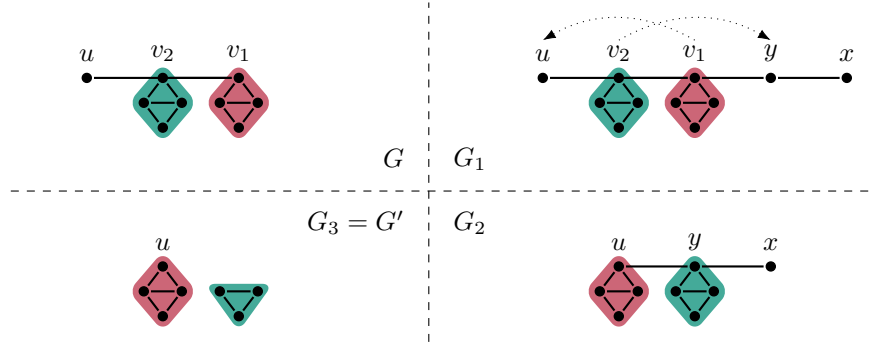
\begin{figure}[!htb]
    \centering
    \begin{tikzpicture}[scale=1]
      \GraphInit[unit=3,vstyle=Normal]
      \SetVertexNormal[Shape=circle, FillColor=black, MinSize=2pt]
      \tikzset{VertexStyle/.append style = {inner sep = \inners, outer sep = \outers}}
      \SetVertexLabelOut
      \newcommand{\gname}{g}

      \begin{scope}
          \Vertex[x=0,y=0,Math,Lpos=90,L={u}]{u_\gname}
          \renewcommand{\gname}{g}
          
          \foreach \i in {1,2} {
          \pgfmathsetmacro{\xv}{3-\i}
          \Vertex[x=\xv,y=0,Math,Lpos=90,L={v_\i}]{v\i_\gname}
            \foreach \j in {1,2,3} {
                \pgfmathsetmacro{\x}{\xv+(\j-2)*0.25}
                \pgfmathmod{\j}{2}
                \pgfmathsetmacro{\y}{-0.66+\pgfmathresult*0.33}
                \Vertex[x=\x,y=\y,NoLabel]{b\j_v\i_\gname}
            }
            \begin{scope}[on background layer]
                \pgfmathmod{\i}{2}
                \pgfmathsetmacro{\ccol}{ifthenelse(\pgfmathresult==1,"goodred","goodteal")}
                \fill[\ccol,rounded corners] ($(v\i_\gname) + (0, 0.2)$) -- ($(b1_v\i_\gname) + (-0.2,0)$) -- ($(b2_v\i_\gname) + (0, -0.2)$)  -- ($(b3_v\i_\gname) + (0.2,0)$) -- cycle;
            \end{scope}
            
            \Edges(b3_v\i_\gname,v\i_\gname,b1_v\i_\gname,b2_v\i_\gname,b3_v\i_\gname,b1_v\i_\gname)
          }
          \Edges(u_\gname,v1_\gname,v2_\gname)
      \end{scope}

      \begin{scope}[xshift=6cm]
          \renewcommand{\gname}{g1}
          \Vertex[x=0,y=0,Math,Lpos=90,L={u}]{u_\gname}
          
          \foreach \i in {1,2} {
          \pgfmathsetmacro{\xv}{3-\i}
          \Vertex[x=\xv,y=0,Math,Lpos=90,L={v_\i}]{v\i_\gname}
            \foreach \j in {1,2,3} {
                \pgfmathsetmacro{\x}{\xv+(\j-2)*0.25}
                \pgfmathmod{\j}{2}
                \pgfmathsetmacro{\y}{-0.66+\pgfmathresult*0.33}
                \Vertex[x=\x,y=\y,NoLabel]{b\j_v\i_\gname}
            }
            \begin{scope}[on background layer]
                \pgfmathmod{\i}{2}
                \pgfmathsetmacro{\ccol}{ifthenelse(\pgfmathresult==1,"goodred","goodteal")}
                \fill[\ccol,rounded corners] ($(v\i_\gname) + (0, 0.2)$) -- ($(b1_v\i_\gname) + (-0.2,0)$) -- ($(b2_v\i_\gname) + (0, -0.2)$)  -- ($(b3_v\i_\gname) + (0.2,0)$) -- cycle;
            \end{scope}
            
            \Edges(b3_v\i_\gname,v\i_\gname,b1_v\i_\gname,b2_v\i_\gname,b3_v\i_\gname,b1_v\i_\gname)
          }
          \Vertex[x=3,y=0,Math,Lpos=90,L={y}]{y_\gname}
          \Vertex[x=4,y=0,Math,Lpos=90,L={x}]{x_\gname}
          
          \Edges(u_\gname,v1_\gname,v2_\gname)
          \Edges(v1_\gname,y_\gname,x_\gname)
          \draw[Latex-,dotted] ($(y_\gname) +(0, 0.5)$) to [bend right] ($(v2_\gname) + (0,0.5)$);
          \draw[Latex-,dotted] ($(u_\gname) +(0, 0.5)$) to [bend left] ($(v1_\gname) + (0,0.5)$);
      \end{scope}
      
      \begin{scope}[xshift=6cm, yshift=-2.5cm]
          \renewcommand{\gname}{g2}
          
          \foreach \i in {1,2} {
          \pgfmathsetmacro{\xv}{3-\i}
          \pgfmathmod{\i}{2}
          \pgfmathsetmacro{\vname}{ifthenelse(\pgfmathresult==0,"u","y")}
          \Vertex[x=\xv,y=0,Math,Lpos=90,L={\vname}]{\vname_\gname}
            \foreach \j in {1,2,3} {
                \pgfmathsetmacro{\x}{\xv+(\j-2)*0.25}
                \pgfmathmod{\j}{2}
                \pgfmathsetmacro{\y}{-0.66+\pgfmathresult*0.33}
                \Vertex[x=\x,y=\y,NoLabel]{b\j_v\i_\gname}
            }
            \begin{scope}[on background layer]
                \pgfmathmod{\i}{2}
                \pgfmathsetmacro{\ccol}{ifthenelse(\pgfmathresult==0,"goodred","goodteal")}
                \fill[\ccol,rounded corners] ($(\vname_\gname) + (0, 0.2)$) -- ($(b1_v\i_\gname) + (-0.2,0)$) -- ($(b2_v\i_\gname) + (0, -0.2)$)  -- ($(b3_v\i_\gname) + (0.2,0)$) -- cycle;
            \end{scope}
            
            \Edges(b3_v\i_\gname,\vname_\gname,b1_v\i_\gname,b2_v\i_\gname,b3_v\i_\gname,b1_v\i_\gname)
          }
          \Edges(u_\gname,y_\gname)
          \Vertex[x=3,y=0,Math,Lpos=90,L={x}]{x_\gname}
          
          \Edges(y_\gname,x_\gname)
      \end{scope}
      
      \begin{scope}[yshift=-2.5cm]
          \renewcommand{\gname}{g3}
          
          \foreach \i in {1,2} {
            \pgfmathsetmacro{\xv}{3-\i}
            \foreach \j in {1,2,3} {
                \pgfmathsetmacro{\x}{\xv+(\j-2)*0.25}
                \pgfmathmod{\j}{2}
                \pgfmathsetmacro{\y}{-0.66+\pgfmathresult*0.33}
                \Vertex[x=\x,y=\y,NoLabel]{b\j_v\i_\gname}
            } 
            \Edges(b1_v\i_\gname,b2_v\i_\gname,b3_v\i_\gname,b1_v\i_\gname)
          }
          
          \foreach \i in {2} {
              \pgfmathsetmacro{\xv}{3-\i}
              \pgfmathmod{\i}{2}
              \pgfmathsetmacro{\vname}{ifthenelse(\pgfmathresult==0,"u","y")}
              \Vertex[x=\xv,y=0,Math,Lpos=90,L={\vname}]{\vname_\gname}
              \Edges(b3_v\i_\gname,u_\gname, b1_v\i_\gname)
          }
        \begin{scope}[on background layer]
        
          \pgfmathsetmacro{\i}{1}
          \fill[goodteal,rounded corners] ($(b1_v\i_\gname) + (-0.2,0.1)$) -- ($(b2_v\i_\gname) + (0, -0.2)$)  -- ($(b3_v\i_\gname) + (0.2,0.1)$) -- cycle;
          \pgfmathsetmacro{\i}{2}
          \fill[goodred,rounded corners] ($(u_\gname) + (0, 0.2)$) -- ($(b1_v\i_\gname) + (-0.2,0)$) -- ($(b2_v\i_\gname) + (0, -0.2)$)  -- ($(b3_v\i_\gname) + (0.2,0)$) -- cycle;
        \end{scope}
      \end{scope}

      \draw[dashed] (-1, -1.5) -- (10.5,-1.5);
      \draw[dashed] (4.5, -3.5) -- (4.5,1);

      \begin{scope}[xshift=4.5cm,yshift=-1.5cm]
          \node[anchor=south east] at (-0.2,0.2) {$G$};
          \node[anchor=south west] at (0.2,0.15) {$G_1$};
          \node[anchor=north west] at (0.2,-0.15) {$G_2$};
          \node[anchor=north east] at (-0.2,-0.13) {$G_3 = G'$};
      \end{scope}
    \end{tikzpicture}
    \caption{Example of the graphs employed in the safeness proof of \cref{rrule:isbd_type4}. Each differently shaded set of four vertices represents one $H_{v_i}$; note that each $v_i$ is of type $B$ as $H_{v_i} \setminus \{v_i\}$ is complete but $\alpha(G[H_{v_i}]) = 2$. The dotted arrows represent the identification operation performed by \cref{rrule:isbd_type2}, which is applied as a black box in the proof of \cref{rrule:isbd_type4}.\label{fig:isbd_type4}}
\end{figure}

\begin{lemma}[\cite{bridgedepth}, Lemma 6.21]
    \label{lem:isbd_bounded_tree}
    Let $T$ be a lowering tree of a connected component $R'$ of $R$. If \cref{rrule:isbd_type1,rrule:isbd_type2,rrule:isbd_type3,rrule:isbd_type4,rrule:isbd_back-n-forth} are not applicable to $T$, then $|V(T)| \in \bigO{|\X|^5\cdot|X|^5}$.
\end{lemma}

We compress our input instance with \cref{alg:isbd_compression}, which has been taken directly from~\cite{bridgedepth}.
Note that the modulator is enlarged at each recursive call. This is not a problem for our framework, as our core, the initial modulator $X$, is kept intact throughout the process. Moreover, by the discussion at the end of \cref{sec:framework}, specifically \cref{lem:condition_one_composes,obs:good_core_subset}, by proving that all traces of the enlarged modulator are preserved, immediately prove that all traces of our core are preserved.

\begin{algorithm}[!htb]
    \caption{Compression algorithm for \pname{Independent Set} parameterized by vertex-deletion distance to $c$-bridgedepth.}\label{alg:isbd_compression}
    \begin{algorithmic}[1]
        \Require An instance $(G, X, t)$ of \pname{Enum Independent Set} parameterized by $|X|$, with $X$ the modulator to $c$-bridgedepth, and a fixed integer $c$.
        \Ensure An equivalent instance $(G', X', t')$ of \pname{Enum Independent Set} with $|V(G')| \in \bigO{\poly(|X|)}$.
        
        \If{$c = 0$}
            \State \Return $(G, X, t)$.
        \EndIf
        \State If \cref{rrule:isbd_comp_removal} is applicable to $(G,X,t)$, apply it and go to Line 1.
        \State If \cref{rrule:isbd_many_confs} is applicable to $(G,X,t)$, apply it and go to Line 1.
        \State \Comment{$R$ now has a bounded number of connected components by \cref{lem:isbd_bounded_comp_count}.}
        \State Set $X_1 \gets \emptyset$.
        \ForEach{connected component $R'$ of $R$}
            \State Compute a lowering tree $T$ of $R'$.
            \State Compute a longest path $T'$ of $T$.
            \State If \cref{rrule:isbd_type1} is applicable to $((G,X,t), T')$, apply it and go to Line 1.
            \State If \cref{rrule:isbd_type2} is applicable to $((G,X,t), T')$, apply it and go to Line 1.
            \State \Comment{The diameter of $T$ is now bounded by \cref{lem:isbd_bounded_diam}.}
            \State If \cref{rrule:isbd_back-n-forth} is applicable to $((G,X,t), T)$, apply it and go to Line 1.
            \State \Comment{The maximum degree of $T$ is now bounded by \cref{obs:isbd_t_deg_bound}.}
            \State Obtain $T_1$ by removing from $T$ all type $A$ leaves with a type $B$ parent.
            \State If \cref{rrule:isbd_type1} is applicable to $((G,X,t), T_1)$, apply it and go to Line 1.
            \State If \cref{rrule:isbd_type3} is applicable to $((G,X,t), T_1)$, apply it and go to Line 1.
            \State If \cref{rrule:isbd_type4} is applicable to $((G,X,t), T_1)$, apply it and go to Line 1.
            \State \Comment{$|V(T)|$ is now bounded by \cref{lem:isbd_bounded_tree}.}
            \State $X_1 \gets X_1 \cup V(T)$.
        \EndFor
        \State \Comment{$X_1$ is bounded by \cref{lem:isbd_bounded_comp_count,lem:isbd_bounded_tree}.}
        \State \Return \cref{alg:isbd_compression}$(G, X \cup X_1, t, c-1)$.
    \end{algorithmic}
\end{algorithm}

\begin{lemma}[\cite{bridgedepth}, Theorem 6.22]
    \label{lem:isbd_compressed}
    Let $(G,X, t)$ and $(G', X', t)$ be the input and output instances of \cref{alg:isbd_compression}, respectively, with the input having total size $n$.
    It holds that $(G', X', t')$ can be computed in $(g_1(c)\cdot n^{g_2(c)})$-time and has size bounded by $\bigO{h(c)\cdot|X|^{2^{f(c)}}}$, with $f(c) \in \bigO{c^2}$ and $g_1,g_2,h$ being computable functions.
\end{lemma}

Stringing together the soundness proofs of each of our reduction rules, we immediately have the following lemma.

\begin{lemma}
    \label{lem:isbd_all_x_traces}
    \cref{cond:good_decision_kernel} holds:
    for every good trace $Y \subseteq X$, i.e., such that there is some $S \in \Sol(G, X, t)$ with $S \cap X = Y$, there exists $S' \in \Sol(G', X', t')$ with $S' \cap X = Y$.
\end{lemma}

\subsection{Lifting algorithm}
\label{sec:isbd_lifting}

All that remains now is to devise a lifting algorithm for \pname{Enum Independent Set} that takes us from $\Sol(G', X, t')$ to $\Sol(G', X',  t')$ with polynomial delay.
We proceed much like in \cref{sec:isfvs_lex_lifting}, and so do not repeat the discussions given there.
At first, it seems that we must be more careful as we augmented $X$ in \cref{alg:isbd_compression} to obtain $X'$; however, the following observation tells us that this is not the case.

\begin{observation}
    \label{obs:isbd_bounded_x_bd}
    It holds that $G[X' \setminus X]$ has bridgedepth at most $c$.
\end{observation}

\begin{proof}
    Note that $V(G') = X'$ and that $G'[X' \setminus X]$ is a graph of bridgedepth at most $c$ by \cref{inv:isbd_controlled_bd} and, in each pass of Line 24 of \cref{alg:isbd_compression}, we essentially move to $X$ a lowering tree of each component of $R$.
\end{proof}

\begin{definition}
    \label{def:isbd_base_y}
    Let $\sigma$ be an arbitrary but fixed ordering of $V(G)$.
    Given a good trace $Y \subseteq X$, we say that $B_Y \in \Sol(G', X', t')$ is the \emph{canonical solution for $Y$} if $Y = B_Y \cap X$, and $S'_Y$ is the lexicographically smallest solution of $(G', X', t')$ with respect to $\sigma$ whose intersection with $X$ is $Y$.
\end{definition}

\begin{lemma}
    \label{lem:isbd_base_y_check}
    Our kernel respects \cref{cond:decidable_trace}:
    there is a polynomial-time algorithm that, given $(G, X, t)$, $(G', X', t')$, $Y \subseteq X$, and $\sigma$: (i) decides if $Y$ is a good trace, and (ii) computes the canonical solution $B_Y \in \Sol(G', X', t')$ for $Y$.
\end{lemma}

\begin{proof}
    For the first claim, using \cref{thm:generic_lifting} and the fact that $G \setminus X$ has bounded treewidth, we can check in polynomial time if $(G \setminus (X \cup N_{G}(Y)), t - |Y|)$ has a solution $S$ and, consequently, if $Y$ is viable.
    For the second claim, let us apply \cref{thm:generic_lifting} again, but this time to the instance $(G' \setminus (X \cup N(Y)), t - |Y|)$.
    Let $Z$ be its first output; observe that $B_Y = Y \cup Z$ is the lexicographically smallest solution of $(G', X', t')$ with $B_Y \cap X = Y$, and so it is the canonical solution for $Y$ solution.
\end{proof}

\begin{lemma}
    \label{lem:isbd_lex_lifting}
    \cref{cond:poly_delay} is satisfied:
    there is a polynomial-delay algorithm that, given $(G, X, t)$, $(G', X', t')$, and a good trace $Y \subseteq X$, outputs $\Lift(Y) = \{S \in \Sol(G, X, t) \mid S \cap X = Y\}$.
\end{lemma}

\begin{proof}
    Define $\I_Y$ as the outputs of \cref{thm:generic_lifting} when applied to $(G \setminus (X \cup N_G(Y)), t - |Y|)$ and set $\Lift(Y) = \{Y \cup Z \mid Z \in \I_Y\}$; we can output the element $Z \cup Y$ of $\Lift(Y)$ as soon as \cref{thm:generic_lifting} outputs $Z$, so our algorithm has delay bounded by that of \cref{thm:generic_lifting}.
\end{proof}

\begin{theorem}
    \label{thm:isbd_kernel}
    There is a \pdkernel{} of polynomial size for \pname{Enum Independent Set} parameterized by the size of a given vertex-deletion distance to $c$-bridgedepth.
\end{theorem}

\begin{proof}
    Our result follows immediately from \cref{thm:framework_kernel,lem:isbd_all_x_traces,lem:isbd_base_y_check,lem:isbd_lex_lifting}.
\end{proof}

\begin{theorem}
    \label{thm:vcbd_kernel}
    There is a \pdkernel{} of polynomial size for \pname{Enum Vertex Cover} parameterized by the size of a given vertex-deletion distance to $c$-bridgedepth.
\end{theorem}

\thmdichotomy*

\bibliography{refs}

@article{ABUKHZAM2010524,
title = {{A kernelization algorithm for $d$-Hitting Set}},
journal = {Journal of Computer and System Sciences},
volume = {76},
number = {7},
pages = {524-531},
year = {2010},
issn = {0022-0000},
doi = {https://doi.org/10.1016/j.jcss.2009.09.002},
author = {Faisal N. Abu-Khzam},
}

@book{murty,
    author = {Bondy, J.A. and Murty, U.S.R},
    title = {Graph Theory},
    year = {2008},
    isbn = {1846289696},
    publisher = {Springer Publishing Company, Incorporated},
    edition = {1st}
}

@book{cygan_parameterized,
  author    = {Marek Cygan and
               Fedor V. Fomin and
               Lukasz Kowalik and
               Daniel Lokshtanov and
               D{\'{a}}niel Marx and
               Marcin Pilipczuk and
               Michal Pilipczuk and
               Saket Saurabh},
  title     = {Parameterized Algorithms},
  publisher = {Springer},
  year      = {2015},
  doi       = {10.1007/978-3-319-21275-3}
}

@book{downey_fellows,
  title={Fundamentals of parameterized complexity},
  author={Downey, Rodney G and Fellows, Michael R},
  volume={4},
  year={2013},
  publisher={Springer}
}

@book{courcelle_book,
place={Cambridge},
series={Encyclopedia of Mathematics and its Applications},
title={Graph Structure and Monadic Second-Order Logic: A Language-Theoretic Approach},
DOI={10.1017/CBO9780511977619},
publisher={Cambridge University Press},
author={Courcelle, Bruno and Engelfriet, Joost},
year={2012},
collection={Encyclopedia of Mathematics and its Applications}}

@book{book_kernels,
    place={Cambridge},
    title={Kernelization: Theory of Parameterized Preprocessing},
    DOI={10.1017/9781107415157},
    publisher={Cambridge University Press},
    author={Fomin, Fedor V. and Lokshtanov, Daniel and Saurabh, Saket and Zehavi, Meirav},
    year={2019}
}

@article{golovach2022refined,
  title={Refined notions of parameterized enumeration kernels with applications to matching cut enumeration},
  author={Golovach, Petr A. and Komusiewicz, Christian and Kratsch, Dieter and Le, Van B.},
  journal={Journal of Computer and System Sciences},
  volume={123},
  pages={76--102},
  year={2022},
  publisher={Elsevier},
  doi = {10.1016/j.jcss.2021.07.005}
}

@Article{creignou2017enum,
author="Creignou, Nadia
and Meier, Arne
and M{\"u}ller, Julian-Steffen
and Schmidt, Johannes
and Vollmer, Heribert",
title="Paradigms for Parameterized Enumeration",
journal="Theory of Computing Systems",
year="2017",
month="May",
day="01",
volume="60",
number="4",
pages="737--758",
issn="1433-0490",
doi="10.1007/s00224-016-9702-4",
url="https://doi.org/10.1007/s00224-016-9702-4"
}

@article{bougeret_is_td,
	author = {Bougeret, Marin and Sau, Ignasi},
	date = {2019/10/01},
	doi = {10.1007/s00453-018-0468-8},
	id = {Bougeret2019},
	isbn = {1432-0541},
	journal = {Algorithmica},
	number = {10},
	pages = {4043--4068},
	title = {How Much Does a Treedepth Modulator Help to Obtain Polynomial Kernels Beyond Sparse Graphs?},
	url = {https://doi.org/10.1007/s00453-018-0468-8},
	volume = {81},
	year = {2019},
}

@article{td_modulator_approx,
title = {Kernelization using structural parameters on sparse graph classes},
journal = {Journal of Computer and System Sciences},
volume = {84},
pages = {219-242},
year = {2017},
issn = {0022-0000},
doi = {https://doi.org/10.1016/j.jcss.2016.09.002},
url = {https://www.sciencedirect.com/science/article/pii/S0022000016300812},
author = {Jakub Gajarský and Petr Hliněný and Jan Obdržálek and Sebastian Ordyniak and Felix Reidl and Peter Rossmanith and Fernando {Sánchez Villaamil} and Somnath Sikdar},
}

@article{creignou_hard,
    title = {A complexity theory for hard enumeration problems},
    journal = {Discrete Applied Mathematics},
    volume = {268},
    pages = {191-209},
    year = {2019},
    issn = {0166-218X},
    doi = {https://doi.org/10.1016/j.dam.2019.02.025},
    author = {Nadia Creignou and Markus Kröll and Reinhard Pichler and Sebastian Skritek and Heribert Vollmer},
}

@article{vc_fvs,
	author = {Jansen, Bart M. P. and Bodlaender, Hans L.},
	date = {2013/08/01},
	doi = {10.1007/s00224-012-9393-4},
	id = {Jansen2013},
	isbn = {1433-0490},
	journal = {Theory of Computing Systems},
	number = {2},
	pages = {263--299},
	title = {{Vertex Cover Kernelization Revisited - Upper and Lower Bounds for a Refined Parameter}},
	url = {https://doi.org/10.1007/s00224-012-9393-4},
	volume = {53},
	year = {2013},
}

@article{nemhauser_crown,
title = {Crown reductions for the Minimum Weighted Vertex Cover problem},
journal = {Discrete Applied Mathematics},
volume = {156},
number = {3},
pages = {292-312},
year = {2008},
note = {Combinatorial Optimization 2004},
issn = {0166-218X},
doi = {https://doi.org/10.1016/j.dam.2007.03.026},
url = {https://www.sciencedirect.com/science/article/pii/S0166218X07001308},
author = {Miroslav Chlebík and Janka Chlebíková},
}

@article{bafna_fvs_apx,
author = {Bafna, Vineet and Berman, Piotr and Fujito, Toshihiro},
title = {A 2-Approximation Algorithm for the Undirected Feedback Vertex Set Problem},
journal = {SIAM Journal on Discrete Mathematics},
volume = {12},
number = {3},
pages = {289-297},
year = {1999},
doi = {10.1137/S0895480196305124},
URL = {https://doi.org/10.1137/S0895480196305124},
eprint = {https://doi.org/10.1137/S0895480196305124},
}

@article{bridgedepth,
  author       = {Marin Bougeret and
                  Bart M. P. Jansen and
                  Ignasi Sau},
  title        = {Bridge-Depth Characterizes which Minor-Closed Structural Parameterizations
                  of Vertex Cover Admit a Polynomial Kernel},
  journal      = {{SIAM} Journal on Discrete Mathematics},
  volume       = {36},
  number       = {4},
  pages        = {2737--2773},
  year         = {2022},
  doi          = {10.1137/21M1400766},
}

@misc{enum_long_path,
      title={Polynomial-Size Enumeration Kernelizations for Long Path Enumeration}, 
      author={Christian Komusiewicz and Diptapriyo Majumdar and Frank Sommer},
      year={2025},
      eprint={2502.21164},
      archivePrefix={arXiv},
      primaryClass={cs.DM},
      url={https://arxiv.org/abs/2502.21164}, 
}

@InProceedings{bagan_delay,
    author="Bagan, Guillaume
    and Durand, Arnaud
    and Grandjean, Etienne",
    editor="Duparc, Jacques
    and Henzinger, Thomas A.",
    title="On Acyclic Conjunctive Queries and Constant Delay Enumeration",
    booktitle="Computer Science Logic",
    year="2007",
    publisher="Springer Berlin Heidelberg",
    address="Berlin, Heidelberg",
    pages="208--222",
    isbn="978-3-540-74915-8"
}

@inproceedings{triangle_incremental,
    author = {Carmeli, Nofar and Kenig, Batya and Kimelfeld, Benny},
    title = {Efficiently Enumerating Minimal Triangulations},
    year = {2017},
    isbn = {9781450341981},
    publisher = {Association for Computing Machinery},
    address = {New York, NY, USA},
    url = {https://doi.org/10.1145/3034786.3056109},
    doi = {10.1145/3034786.3056109},
    booktitle = {Proceedings of the 36th ACM SIGMOD-SIGACT-SIGAI Symposium on Principles of Database Systems},
    pages = {273–287},
    numpages = {15},
    location = {Chicago, Illinois, USA},
    series = {PODS '17}
}

@inproceedings{fo_query_delay,
    author = {Durand, Arnaud and Schweikardt, Nicole and Segoufin, Luc},
    title = {Enumerating answers to first-order queries over databases of low degree},
    year = {2014},
    isbn = {9781450323758},
    publisher = {Association for Computing Machinery},
    address = {New York, NY, USA},
    url = {https://doi.org/10.1145/2594538.2594539},
    doi = {10.1145/2594538.2594539},
    booktitle = {Proceedings of the 33rd ACM SIGMOD-SIGACT-SIGART Symposium on Principles of Database Systems},
    pages = {121–131},
    numpages = {11},
    location = {Snowbird, Utah, USA},
    series = {PODS '14}
}

@article{frequent_incremental,
    author = {Kimelfeld, Benny and Kolaitis, Phokion G.},
    title = {The Complexity of Mining Maximal Frequent Subgraphs},
    year = {2015},
    issue_date = {December 2014},
    publisher = {Association for Computing Machinery},
    address = {New York, NY, USA},
    volume = {39},
    number = {4},
    issn = {0362-5915},
    url = {https://doi.org/10.1145/2629550},
    doi = {10.1145/2629550},
    journal = {ACM Trans. Database Syst.},
    month = dec,
    articleno = {32},
    numpages = {33}
}

@PHDTHESIS{strozecki_thesis,
    url = "http://www.theses.fr/2010PA077178",
    title = "Enumeration complexity and matroid decomposition",
    author = "Strozecki, Yann",
    year = "2010",
    school = "Paris 7",
}

@article{maximal_ind_set_delay,
    title = {On generating all maximal independent sets},
    journal = {Information Processing Letters},
    volume = {27},
    number = {3},
    pages = {119-123},
    year = {1988},
    issn = {0020-0190},
    doi = {https://doi.org/10.1016/0020-0190(88)90065-8},
    url = {https://www.sciencedirect.com/science/article/pii/0020019088900658},
    author = {David S. Johnson and Mihalis Yannakakis and Christos H. Papadimitriou},
}

@article{polynomial_hierarchy,
    title = {The polynomial-time hierarchy},
    journal = {Theoretical Computer Science},
    volume = {3},
    number = {1},
    pages = {1-22},
    year = {1976},
    issn = {0304-3975},
    doi = {https://doi.org/10.1016/0304-3975(76)90061-X},
    url = {https://www.sciencedirect.com/science/article/pii/030439757690061X},
    author = {Larry J. Stockmeyer},
}

@InProceedings{fernau_param_enum,
    author="Fernau, Henning",
    editor="Ibarra, Oscar H.
    and Zhang, Louxin",
    title="On Parameterized Enumeration",
    booktitle="Computing and Combinatorics",
    year="2002",
    publisher="Springer Berlin Heidelberg",
    address="Berlin, Heidelberg",
    pages="564--573",
    isbn="978-3-540-45655-1"
}

@InProceedings{buss_kernel,
    author="Buss, Jonathan F.
    and Goldsmith, Judy",
    editor="Choffrut, Christian
    and Jantzen, Matthias",
    title="{Nondeterminism within P}",
    booktitle="STACS 91",
    year="1991",
    publisher="Springer Berlin Heidelberg",
    address="Berlin, Heidelberg",
    pages="348--359",
    isbn="978-3-540-47002-1"
}

@article{damaschke_full_kernels,
    title = {Parameterized enumeration, transversals, and imperfect phylogeny reconstruction},
    journal = {Theoretical Computer Science},
    volume = {351},
    number = {3},
    pages = {337-350},
    year = {2006},
    note = {Parameterized and Exact Computation},
    issn = {0304-3975},
    doi = {https://doi.org/10.1016/j.tcs.2005.10.004},
    author = {Peter Damaschke},
}

@article{damaschke_cluster_editing,
	author = {Damaschke, Peter},
	date = {2010/02/01},
	doi = {10.1007/s00224-008-9130-1},
	isbn = {1433-0490},
	journal = {Theory of Computing Systems},
	number = {2},
	pages = {261--283},
	title = {Fixed-Parameter Enumerability of Cluster Editing and Related Problems},
	url = {https://doi.org/10.1007/s00224-008-9130-1},
	volume = {46},
	year = {2010},
}

@article{meeks_oracle,
	author = {Meeks, Kitty},
	date = {2019/02/01},
	doi = {10.1007/s00453-018-0404-y},
	id = {Meeks2019},
	isbn = {1432-0541},
	journal = {Algorithmica},
	number = {2},
	pages = {519--540},
	title = {Randomised Enumeration of Small Witnesses Using a Decision Oracle},
	url = {https://doi.org/10.1007/s00453-018-0404-y},
	volume = {81},
	year = {2019}
}

@InProceedings{multicut_ipec,
  author =	{C. M. Gomes, Guilherme and Juliano, Emanuel and Martins, Gabriel and F. dos Santos, Vinicius},
  title =	{{Matching (Multi)Cut: Algorithms, Complexity, and Enumeration}},
  booktitle =	{19th International Symposium on Parameterized and Exact Computation (IPEC 2024)},
  pages =	{25:1--25:15},
  series =	{Leibniz International Proceedings in Informatics (LIPIcs)},
  ISBN =	{978-3-95977-353-9},
  ISSN =	{1868-8969},
  year =	{2024},
  volume =	{321},
  address =	{Dagstuhl, Germany},
  URN =		{urn:nbn:de:0030-drops-222514},
  doi =		{10.4230/LIPIcs.IPEC.2024.25},
}

@inproceedings{oscar_degeneracy,
    author = {Bartier, Valentin and Defrain, Oscar and Mc Inerney, Fionn},
    title = {Hypergraph Dualization with FPT-delay Parameterized by the Degeneracy and Dimension},
    year = {2024},
    isbn = {978-3-031-63020-0},
    publisher = {Springer-Verlag},
    address = {Berlin, Heidelberg},
    url = {https://doi.org/10.1007/978-3-031-63021-7_9},
    doi = {10.1007/978-3-031-63021-7_9},
    booktitle = {Combinatorial Algorithms: 35th International Workshop, IWOCA 2024, Ischia, Italy, July 1–3, 2024, Proceedings},
    pages = {111–125},
    numpages = {15},
    location = {Ischia, Italy}
}

@misc{enum_dcut,
      title={Enumeration Kernels of Polynomial Size for Cuts of Bounded Degree}, 
      author={Christian Komusiewicz and Diptapriyo Majumdar},
      year={2025},
      eprint={2308.01286},
      archivePrefix={arXiv},
      primaryClass={cs.DS},
      url={https://arxiv.org/abs/2308.01286}, 
}

@article{eiter2003new,
  title={New results on monotone dualization and generating hypergraph transversals},
  author={Eiter, Thomas and Gottlob, Georg and Makino, Kazuhisa},
  journal={SIAM Journal on Computing},
  volume={32},
  number={2},
  pages={514--537},
  year={2003},
  publisher={SIAM}
}

@article{cross_composition,
    author = {Bodlaender, Hans L. and Jansen, Bart M. P. and Kratsch, Stefan},
    title = {Kernelization Lower Bounds by Cross-Composition},
    journal = {SIAM Journal on Discrete Mathematics},
    volume = {28},
    number = {1},
    pages = {277-305},
    year = {2014},
    doi = {10.1137/120880240},
    URL = {https://doi.org/10.1137/120880240},
    eprint = {https://doi.org/10.1137/120880240}
}

@article{crown_decomp_fellows,
	author = {Abu-Khzam, Faisal N. and Fellows, Michael R. and Langston, Michael A. and Suters, W. Henry},
	doi = {10.1007/s00224-007-1328-0},
	id = {Abu-Khzam2007},
	isbn = {1433-0490},
	journal = {Theory of Computing Systems},
	number = {3},
	pages = {411--430},
	title = {Crown Structures for Vertex Cover Kernelization},
	url = {https://doi.org/10.1007/s00224-007-1328-0},
	volume = {41},
	year = {2007},
}

@article{branch_bound,
    ISSN = {00129682, 14680262},
    URL = {http://www.jstor.org/stable/1910129},
    author = {A. H. Land and A. G. Doig},
    journal = {Econometrica},
    number = {3},
    pages = {497--520},
    publisher = {[Wiley, Econometric Society]},
    title = {An Automatic Method of Solving Discrete Programming Problems},
    volume = {28},
    year = {1960}
}

@book{flum_grohe,
  author       = {J{\"{o}}rg Flum and
                  Martin Grohe},
  title        = {Parameterized Complexity Theory},
  series       = {Texts in Theoretical Computer Science. An {EATCS} Series},
  publisher    = {Springer},
  year         = {2006},
  url          = {https://doi.org/10.1007/3-540-29953-X},
  doi          = {10.1007/3-540-29953-X},
  isbn         = {978-3-540-29952-3},
}

@book{niedermeier_book,
    author = {Niedermeier, Rolf},
    title = {Invitation to Fixed-Parameter Algorithms},
    publisher = {Oxford University Press},
    year = {2006},
    month = {02},
    isbn = {9780198566076},
    doi = {10.1093/acprof:oso/9780198566076.001.0001},
    url = {https://doi.org/10.1093/acprof:oso/9780198566076.001.0001},
}

@InProceedings{original_crown,
author="Chor, Benny
and Fellows, Mike
and Juedes, David",
editor="Hromkovi{\v{c}}, Juraj
and Nagl, Manfred
and Westfechtel, Bernhard",
title="Linear Kernels in Linear Time, or How to Save {$k$} Colors in {$O(n^2)$} Steps",
booktitle="Graph-Theoretic Concepts in Computer Science",
  url          = {https://doi.org/10.1007/978-3-540-30559-0\_22},

year="2005",
publisher="Springer Berlin Heidelberg",
address="Berlin, Heidelberg",
pages="257--269",
isbn="978-3-540-30559-0"
}

@inproceedings{vc_fvs_strong_pd_kernel,
  author       = {Marin Bougeret and
                  Guilherme C. M. Gomes and
                  Vin{\'{\i}}cius Fernandes dos Santos and
                  Ignasi Sau},
  title        = {{Enumeration Kernels for Vertex Cover and Feedback Vertex Set}},
  booktitle    = {Proc. of the 20th International Symposium on Parameterized and Exact Computation (IPEC)},
  series       = {LIPIcs},
  volume       = {358},
  pages        = {23:1--23:18}, 
  year         = {2025}, 
  note = {The full version is available at \url{https://arxiv.org/pdf/2509.08475}},
  doi          = {10.4230/LIPICS.IPEC.2025.23}, 
}

@article{strozecki_eatcs,
  title={Enumeration complexity},
  author={Strozecki, Yann},
  journal={Bulletin of EATCS},
  volume={3},
  number={129},
  year={2019}
}

@article{MajumdarRR18,
  author    = {Diptapriyo Majumdar and
               Venkatesh Raman and
               Saket Saurabh},
  title     = {Polynomial Kernels for Vertex Cover Parameterized by Small Degree
               Modulators},
  journal   = {Theory of Computing Systems},
  volume    = {62},
  number    = {8},
  pages     = {1910--1951},
  year      = {2018},
  doi       = {10.1007/s00224-018-9858-1},
}

@inproceedings{FominS16,
  author    = {Fedor V. Fomin and
               Torstein J. F. Str{\o}mme},
  title     = {Vertex Cover Structural Parameterization Revisited},
  booktitle = {Proc. of the 42nd International
                  Workshop on Graph-Theoretic Concepts in Computer Science (WG)},
  series    = {LNCS},
  volume    = {9941},
  pages     = {171--182},
  year      = {2016},
  doi       = {10.1007/978-3-662-53536-3_15},
}

@inproceedings{HolsK17,
  author    = {Eva{-}Maria C. Hols and
               Stefan Kratsch},
  title     = {Smaller Parameters for Vertex Cover Kernelization},
  booktitle = {Proc. of the 12th International Symposium on Parameterized and Exact Computation (IPEC)},
  series    = {LIPIcs},
  volume    = {89},
  pages     = {20:1--20:12},
  year      = {2017},
  doi       = {10.4230/LIPIcs.IPEC.2017.20},
}

@article{CyganLPPS14,
  author       = {Marek Cygan and
                  Daniel Lokshtanov and
                  Marcin Pilipczuk and
                  Michal Pilipczuk and
                  Saket Saurabh},
  title        = {On the Hardness of Losing Width},
  journal      = {Theory of Computing Systems},
  volume       = {54},
  number       = {1},
  pages        = {73--82},
  year         = {2014},  
  doi          = {10.1007/S00224-013-9480-1}, 
}

\end{document}